\documentclass[11pt]{article}

\usepackage[T1]{fontenc}
\usepackage[margin=1in]{geometry}
\usepackage{newunicodechar}
\newunicodechar{–}{--}
\newunicodechar{’}{'}
\newunicodechar{“}{``}
\newunicodechar{”}{''}

\usepackage{amsmath,amssymb,amsfonts,amsbsy,amsthm,mathtools,latexsym}
\usepackage{stmaryrd}

\usepackage{booktabs}
\usepackage{tabularx,makecell,multirow}
\usepackage{colortbl}
\usepackage{paralist}
\usepackage{enumitem}
\usepackage{float}
\usepackage{graphicx}
\usepackage{pdfpages}
\usepackage{mdframed}
\mdfsetup{nobreak=true}
\usepackage{boxedminipage}
\usepackage{pifont}
\usepackage{xcolor}
\usepackage{soul}
\usepackage[normalem]{ulem}
\usepackage{xspace}

\usepackage{algpseudocode}
\usepackage[
  lambda,
  advantage,
  operators,
  sets,
  adversary,
  landau,
  probability,
  notions,
  logic,
  ff,
  mm,
  primitives,
  events,
  complexity,
  asymptotics,
  keys
]{cryptocode}

\usepackage{tikz}
\usepackage{pgfplots}
\pgfplotsset{compat=newest}
\usepgfplotslibrary{fillbetween}
\usetikzlibrary{fit,shapes,positioning,patterns,arrows,shadows}
\usetikzlibrary{decorations.pathreplacing,calligraphy}
\tikzset{
  master/.style={
    execute at end picture={
      \coordinate (lower right) at (current bounding box.south east);
      \coordinate (upper left) at (current bounding box.north west);
    }
  },
  slave/.style={
    execute at end picture={
      \pgfresetboundingbox
      \path (upper left) rectangle (lower right);
    }
  }
}

\usepackage{cite}
\usepackage{hyperref}
\usepackage[capitalize,noabbrev]{cleveref}

\sethlcolor{lightgray}

\theoremstyle{plain}
\newtheorem{theorem}{Theorem}[section]
\newtheorem{lemma}[theorem]{Lemma}
\newtheorem{corollary}[theorem]{Corollary}

\newtheorem{claim}[theorem]{Claim}
\newtheorem{fact}[theorem]{Fact}

\newtheorem{mechanism}[theorem]{Mechanism}
\theoremstyle{definition}
\newtheorem{definition}[theorem]{Definition}

\crefname{claim}{claim}{claims}
\Crefname{claim}{Claim}{Claims}
\crefname{fact}{fact}{facts}
\Crefname{fact}{Fact}{Facts}
\crefname{assumption}{assumption}{assumptions}
\Crefname{assumption}{Assumption}{Assumptions}
\crefname{algorithm}{algorithm}{algorithms}
\Crefname{algorithm}{Algorithm}{Algorithms}
\crefname{mechanism}{mechanism}{mechanisms}
\Crefname{mechanism}{Mechanism}{Mechanisms}
\crefname{protocol}{protocol}{protocols}
\Crefname{protocol}{Protocol}{Protocols}

\newcommand{\mcal}[1]{\ensuremath{\mathcal{#1}}}

\newcommand{\M}{\mathcal M}

\newcommand{\eps}{\epsilon}
\newcommand{\R}{\mathbb{R}}

\newcommand{\E}{\mathbb{E}}
\newcommand{\C}{\mathcal{C}}

\definecolor{darkgreen}{rgb}{0.6,0.6,0.25}
\definecolor{lightblue}{RGB}{196,212,224}
\definecolor{darkblue}{RGB}{74,135,161}
\definecolor{lightpurple}{RGB}{191,181,201}
\definecolor{mygreen}{RGB}{194,207,162}
\definecolor{myred}{HTML}{C86733}
\definecolor{myblue}{HTML}{5B68FF}
\definecolor{myshadow}{HTML}{E6C5B4}
\definecolor{darkred}{rgb}{0.5,0,0}
\definecolor{maingreen}{RGB}{094,166,156}

\hypersetup{colorlinks=true,allcolors=darkred}

\newif\ifrevision
\revisionfalse

\ifrevision
\newcommand{\revision}[1]{{\color{red}#1}}
\newcommand{\delete}[1]{}
\else
\newcommand{\revision}[1]{#1}
\newcommand{\delete}[1]{}
\fi

\newcounter{task}

\newcommand{\util}{\ensuremath{{\sf util}}}
\newcommand{\mutil}{\ensuremath{{\sf mutil}}}
\newcommand{\bid}{\ensuremath{{\bf b}}}
\newcommand{\val}{\ensuremath{{\bf v}}}
\newcommand{\rminf}{\ensuremath{{\sf rmInf}}}

\newcommand{\tfm}{\ensuremath{({\bf I},{\bf C},{\bf P},{\bf R})}}
\newcommand{\res}{\ensuremath{{\sf res}}}

\newcommand{\papertitle}{Beyond Incentive Compatibility: Rational Harm-Proof Transaction Fee Mechanisms\thanks{Author order is randomized.}}

\title{\papertitle}
\author{Forest Zhang \qquad Elain Park \qquad Ke Wu\\[0.5em]
\small University of Michigan, Ann Arbor, MI, USA\\
\small\texttt{\{forestz,elainpk,kewucse\}@umich.edu}}
\date{}

\begin{document}
\raggedbottom
\setlength{\emergencystretch}{3em}
\maketitle
\begin{abstract}
On a blockchain, users compete for scarce block space in an auction run by the miner to get their transactions confirmed in the block. This auction is called \emph{transaction fee mechanism} (TFM). Recent work~\cite{roughgarden2021transaction,foundation-tfm,crypto-tfm} has been focused on \emph{incentive compatibility} (IC), requiring that honest behavior maximizes the payoff for each type of strategic player: users, the miner, or miner–user coalitions. 
In this work, we introduce \emph{rational-harm proofness (RHP)}, which rules out deviations that harm honest parties without also reducing the deviator’s own utility relative to the honest baseline.
RHP closes a gap left by IC: IC does not forbid utility neutral yet externally harmful deviations. For example, in a second-price auction, the second-highest bidder can increase the winner’s payment without affecting their own payoff. Such deviation is eliminated by RHP.

We characterize TFMs satisfying RHP alongside incentive compatibility for users (UIC) and miners (MIC). For finite block size, we develop a \emph{complete characterization} in two models:
\begin{itemize}[leftmargin=6mm]
\item In the \textit{plain model}---where a single miner unilaterally implements the auction---we prove a \emph{tetrilemma} (3-out-of-4 impossibility): among the four desired properties \{positive miner revenue, UIC, MIC, RHP against miner–user coalitions\}, no mechanism achieves all four simultaneously.
This trade-off is tight: any three are jointly achievable in the plain model.
\item In the \textit{MPC-assisted model}---where a committee of miners jointly implement the auction via multi-party computation (MPC)---we construct a randomized TFM with a positive miner revenue that achieves UIC, MIC, and RHP against all three types of strategic players. We further show that randomness is necessary: any deterministic TFM satisfying UIC and RHP in this model must confirm no transactions when the number of users exceeds the block size.
\end{itemize}
Finally, we show that IC and RHP are \emph{incomparable}: for each strategic role, there are mechanisms satisfying one but not the other in both models. 
Our results broaden the design objectives for TFMs: beyond incentive compatibility, mechanisms should also preclude costless harm to honest participants.

\end{abstract}
\clearpage
\tableofcontents
\clearpage

\section{Introduction}
Block space on a blockchain is a scarce resource, so users compete with each other to get their transactions confirmed in a block.
This process can essentially be viewed as an auction where the block producer, also called the miner, sells the limited $k$ number of block slots. Each user bids to get one slot for its transaction, and a \emph{transaction fee mechanism} (TFM) decides which transactions are confirmed, how much each confirmd transaction pays, and how much revenue the miner receives. 

Recent progress~\cite{zoharfeemech,yaofeemech,functional-fee-market,eip1559,roughgarden2021transaction,dynamicpostedprice} observed that, due to the decentralized environment of blockchain, TFM design departs significantly from classical auction design.
For example, classical auctions typically assume a trusted auctioneer and focus on the strategies of individual users. 
In contrast, on a blockchain, the auctioneer is a strategic miner who may deviate from the prescribed protocol to gain more revenue. 
Additionally, smart contracts make it easy for miners and users to enter binding side contracts and split their joint gains off-chain. 
A growing body of work therefore explores \emph{incentive compatibility} (IC) properties of TFM so that the honest behavior is provably best for all participants.
Ideally, a ``dream'' TFM should satisfy 1.) \emph{user incentive compatibility} (UIC): bidding one's true valuation is optimal for each individual user; 2.) \emph{miner incentive compatibility} (MIC): the miner is incentivized to execute the mechanism honestly; and 3.) \emph{side-contract proofness} (SCP): no miner-user coalition can deviate to jointly profit.

However, incentive compatibility alone does not protect honest participants. 
There are deviations that leave a strategic player's utility unchanged compared to honest behavior, yet strictly harm others. 
Standard incentive compatibility notions do not rule out such utility-neutral but externally harmful actions.
For example, the gold-standard second-price auction~\cite{vickrey1961counterspeculation} is known to be incentive compatible for individual users: the highest bidder wins and pays the second-highest price. However, the second-highest bidder can raise their bid slightly without changing their own utility (the bidder still loses and pays nothing) but increase the winner's payment.
Such manipulations are costless for the deviator but impose real externalities on others. 

\revision{Concern about such harm is well established in blockchain protocol design.
For example, Herlihy's atomic-swap~\cite{herlihy2018atomic} formulation separately requires that coalition deviations not leave conforming parties worse off in addition to IC. Subsequent work~\cite{xue2021hedging} studies attacks that impose losses on honest participants in cross-chain transactions.
More broadly, griefing-style attacks~\cite{buterin2018griefing}, in which a participant can impose disproportionate harm on others at relatively low cost, have been studied in blockchain protocols including hashed timelock contracts~\cite{tsabary2021madhtlc} and payment channels~\cite{aumayr2023domino}.
These works motivate treating harm to other participants as a design concern separate from whether a deviation is profitable under the mechanism's modeled utility.
In particular, when several strategies yield the same utility as honest behavior, protecting honest participants should not rely on the deviator choosing a benign strategy.}
This motivates us to ask
\begin{center}
\emph{Can we design TFMs that achieve desired incentive compatibilities while also protecting honest participants from being harmed by rational deviations?}
\end{center}
Related concepts have been studied in classical mechanism design under non-bossiness and robustness to secondary goals~\cite{satterthwaite1981strategy, thomson2016non, leme2023nonbossy, duque2024local,carroll2019robustness, jehiel2005allocative, lopomo2021uncertainty}.
Thus, the underlying concern is not unique to TFMs.
However, those notions apply to individual agents and do not capture the decentralized setting of TFMs with strategic miners and miner-user coalitions. 
In this work, we initiate a systematic study of the above question in the context of TFMs where the mechanism implementer is strategic and can collude with the users.
We formalize the requirement that no strategic players (whether a user, miner, or miner-user coalition) can harm any honest participants unless they also harm the deviator themselves as a property called {\bf rational-harm proofness} (RHP). 
We then characterize when TFMs can achieve RHP alongside the standard incentive guarantees.
For practical viability, we also target \emph{positive miner revenue} so that participation is profitable for miners.

Our results apply to two settings:  
(1) the \textbf{plain model}, where a single miner unilaterally determines the block contents, capturing today’s mainstream blockchain architecture; and  
(2) the \textbf{MPC-assisted model}~\cite{crypto-tfm}, where a committee of $M$ miners jointly execute a cryptographic primitive, called the multi-party computation (MPC), to implement the TFM. We assume $M\ge3$ in this model. 
The detailed results are summarized below.

\subsection{Our Results}
\label{sec:contribution}
\paragraph{Model at a Glance.}
We begin with a high-level model of TFM to contextualize our results. Formal definitions appear in \Cref{sec:model}. 
We focus on a single block of finite size $k$. 
One can view the mechanism as selling $k$ identical block slots.
Each user $i$ has a private true value $v_i$ indicating the maximum they are willing to pay for a block slot. Users submit bids for inclusion, and a miner proposes a block of up to $k$ bids. 
A TFM specifies, upon a received bid vector, 1.) up to $k$ bids to include from the input bids; 2.) the set of confirmed bids and the payment charge to each confirmed bid; 3.) the miner's revenue.
Once the set of up to $k$ bids to include are fixed\footnote{There is a subtle distinction between a bid being included in a block versus being confirmed. A bid must be included to have a chance of confirmation, but depending on the mechanism, not all included bids are necessarily confirmed; see \Cref{sec:model} for detailed explanation.}, the confirmation, payment, and miner revenue rules are honestly executed on-chain. 
Importantly, a TFM must work for \emph{any} number of input bids due to the open and permissionless environment of blockchain.
Moreover, on blockchain, not all payments need to go to the miner: some or even all payments can be \emph{burned}.

In the \emph{plain model}, a single miner unilaterally determines which bids to include. 
In the \emph{MPC-assisted model}, a committee of $M$ miners jointly execute the mechanism via an MPC, as if a trusted party were running the mechanism honestly on all bids delivered to the committee.
\revision{We assume that less than half of the miners are corrupted, and the MPC provides guaranteed output delivery of the outcome prescribed by the protocol}\footnote{When instantiating with real-world cryptography, there is a negligible probability of breaking cryptography, but our results are stated in the ideal world assuming an ideal functionality.}. 

We require the mechanism to satisfy three basic properties: 1.) \textit{Individual rationality}: no bid pays more than its value; 2.) \textit{Budget feasibility}: the miner’s revenue does not exceed the total payments paid by users, i.e., the mechanism does not create money.\footnote{The miner gets a fixed block reward independent from the transaction fees, so we do not model the block reward.}; 3.) \textit{Weak symmetry}: metadata such as the identity or timestamp of a bid is used only for tie-breaking; otherwise, outcomes depend only on bid values. 
For later reference, we note that all of our impossibility results hold assuming only weak symmetry, but our posted-price feasibility constructions are strongly symmetric, meaning that all bids with the same value have identical outcome distributions, without a tie-breaking based on metadata.

\smallskip{\it Strategy Space.}
We focus on direct-revelation mechanisms, in which honest users submit a single bid equal to their true values. 
An honest miner includes up to $k$ bids as prescribed in the mechanism, without censoring honest bids or injecting fake ones.

Strategic players may deviate to increase their utility.
We consider three types of strategic players: 1.) A single user. A strategic user can bid an arbitrary value for its transaction, refuse to bid, and/or inject any number of \emph{fake bids} since pseudonyms are easy to register on-chain. 
2.) Strategic miner(s). In the plain model, the single miner can inject fake bids, drop honest users' bids, and select any set of up to $k$ bids for the block. In the MPC-assisted model, a coalition of $m<M/2$ miners can inject fake bids, but they cannot selectively censor honest users' bids or change the outcome, and the MPC provides guaranteed output delivery. 3.) A miner-user coalition, containing the miner and some users in the plain model, or $m<M/2$ miners together with some users. A coalition's strategy is a combination of its members' strategies.

\smallskip{\it Utility.} A user with true value $v$ for a transaction gets utility $v - p$ if that transaction is confirmed and pays $p$, and $0$ if not confirmed. A miner’s utility is its revenue minus any cost it pays for its own fake bids. A coalition’s utility is the sum of its members’ utilities.

\paragraph{Conceptual Contribution: Defining RHP} 
Our first contribution is conceptual: we elevate the principle of ``no costless harm'' to a first-class requirement for TFM design.  
As illustrated by the earlier second-price auction example, a user may raise the winner’s payment and thereby reduce a rival’s utility without affecting its own payoff. 

We formalize \emph{rational-harm proofness (RHP)} as the requirement that rational players cannot make any honest participant strictly worse off unless the deviators also make themselves worse off, compared to the honest behavior. We define RHP as a property of the mechanism and take the honest profile as its baseline rather than tying the notion to equilibrium outcomes of the induced TFM game. 
An equilibrium-based formulation would depend on equilibrium existence and selection choices, whereas our goal is a notion that is well-defined for \emph{every} TFM instance independently of such choices.

We adopt \emph{ex-post} notions because they capture the \emph{public-mempool} reality of many blockchains:
bids are broadcast and observable by strategic participants, in particular by the miner before deciding what to include, and often by other users who can react via last move. Ex-post definitions ensure robustness in such environments where a strategic players can make their decisions based on the honest users' bids, rather than relying on bid secrecy.


Formally, for any strategic player $\C$, whether an individual user, miner, or miner-user coalition, let $H_{\mcal{C}}$ denote $\C$'s honest strategy.
For any true value profile $\val = (v_1,...,v_n)$ of all the users, let $\util_i(\val;S_{\C})$ and $\util_{\C}(\val;S_{\C})$ denote the expected utility of a player $i$ and the coalition $\C$, respectively, in the randomized experiment where players in $\C$ adopts some (possibly randomized) strategy $S_{\C}$ and players outside $\C$ behaves honestly.

\begin{definition}[Rational-harm proofness]
\label{defn:intro-rhp}
Fix a coalition of strategic players $\C$ and a protected set of honest players $\mcal H$ such that $\C\cap \mcal H=\emptyset$. We say that a TFM satisfies rational-harm proofness (RHP) against $\C$ for protecting $\mcal H$ if, for any true value vector $\val$ of users and any strategy $S_{\C}$ for coalition $\C$ such that
\[
    \util_{\C}(\val;S_{\C})\geq \util_{\C}(\val;H_{\C}),
\]
it holds that
\[
    \util_i(\val;S_{\C})\geq \util_i(\val;H_{\C})
\]
for every protected honest player $i\in\mcal H$. Here $\util_a$ and $\util_{\C}$ denote expected utilities in the randomized experiment where the users' true value vector is $\val$, players in $\C$ adopt $S_{\C}$, players outside $\C$ behave honestly, and $H_{\C}$ denotes the honest strategy of coalition $\C$.
\end{definition}
\revision{To help understand RHP, let us take the first-price auction in the plain model as an example.
Here, the $k$ highest bidder gets confirmed and pays their own bid.
Under honest behavior, every confirmed user bids its true value and therefore obtains utility zero, while every unconfirmed user also obtains utility zero.
Thus, first-price auction satisfies RHP in the sense that no user can be harmed compared to honest case.
Intuitively, although strategic players may change the allocation or payments, they cannot make an honest user worse off than its zero-utility honest baseline.

This example also illustrates what RHP does \emph{not} guarantee.
The first-price auction is not UIC: a winning user may lower its bid to make it barely enough to be the top $k$ bids.
This strictly increases its utility.
Thus, RHP protects honest participants relative to the prescribed honest baseline, but does not by itself ensure that the honest baselline is incentive compatible.
Conversely, as illustrated by the earlier second-price example, IC does not prevent utility-neutral deviations that harm others.

The first-price and second-price examples together demonstrate necessity of requiring both RHP and IC: 
They are complementary requirements.
IC makes the prescribed honest behavior optimal for the strategic player, while RHP protects honest participants against the external effects of rational deviations.}

\paragraph{RHP for different strategic roles.}
\revision{
TFMs have three natural classes of strategic players already reflected in the IC notions defined in prior works:
an individual user, the miner, and a miner-user coalition.
We take the same distinction for RHP and instantiate \Cref{defn:intro-rhp} for these three strategic roles.
These are not separate underlying notions, but three instantiations of the same RHP definition with different deviating strategy spaces and protected participants.
Distinguishing them lets us separate precisely which sources of rational harm a mechanism protects against, and whether these protections can be achieved simultaneously:}
\begin{itemize}
    \item \emph{User rational-harm proofness (URHP)} iff \Cref{defn:intro-rhp} holds for any $\C=\{i\}$ containing an individual user $i$ and $\mcal H$ contains all users except $i$. URHP protects honest users from an individual deviating user.

    \item \emph{Miner rational-harm proofness (MRHP)} iff \Cref{defn:intro-rhp} holds for the single miner in the plain model, with $\mcal H$ equal to all users, and for any nonempty coalition of $m<M/2$ miners in the MPC-assisted model, with $\mcal H$ equal to all users and all miners outside the coalition.

    \item $d$-\emph{coalition rational-harm proofness ($d$-CRHP)} iff \Cref{defn:intro-rhp} holds for coalitions containing the miner and at least one but at most $d$ users in the plain model, or a nonempty coalition of $m<M/2$ miners and at least one but at most $d$ users in the MPC-assisted model, protecting all players outside the coalition.
\end{itemize}
Our URHP notion focuses only on protecting honest users from rival-user: we do not require an individual user deviation to preserve miner revenue.
A user may reduce miner revenue simply by choosing not to participate, which we view as ordinary demand response rather than harm to be protected against.
In contrast, MRHP and CRHP concern deviations involving miners and therefore also protect non-colluding honest miners in the MPC-assisted model.

\paragraph{Our design target.}
Ideally, a TFM would satisfy all of the standard incentive guarantees (UIC, MIC, and SCP) together with all RHP guarantees (URHP, MRHP, CRHP), as well as positive miner revenue.
However, for finite block size, prior work~\cite{foundation-tfm,crypto-tfm} implies that UIC and $d$-SCP, for $d\ge2$, can coexist in either model only if every user and miner has zero expected utility under honest behavior on every valuation profile.
Since UIC is tightly coupled to user experience, 
we therefore prioritize UIC and aim to achieve UIC, MIC, all three RHP guarantees, and positive miner revenue.

In this case, UIC and MIC ensures that the honest protocol forms an equilibrium, and RHP then provides protections of honest participants compared to this honest baseline.
Now CRHP becomes especially useful when SCP is impossible:
Even when we cannot prevent a coalition from benefiting, CRHP ensures that the coalition not obtain such a benefit at the expense of honest participants.

\revision{This full collection of desired properties---{\bf UIC, MIC, URHP, MRHP, CRHP, and positive miner revenue}---is our common design target in both the plain model and the MPC-assisted model.
Our impossibility results are stated under smaller subsets of these properties because they identify stronger obstructions. 
We state only the assumptions needed for each impossibility.}

We also define a weaker notion called \emph{weak RHP}, which only requires that any \emph{strictly profitable} deviation for $\C$ does not harm honest participants. This is weaker than RHP and weaker than IC, since IC rule out the existence of any strictly profitable deviations. 

\paragraph{Landscape in the Plain Model.}
We first consider the plain model, where a single miner determines which bids to include. This setting reflects the architecture of most mainstream blockchains.
The game proceeds as follows: honest users submit their bids truthfully; any strategic players then choose their strategies after observing the honest bids; the miner selects up to $k$ bids for the block. The blockchain protocol then honestly confirms some of these bids and determines payments and miner revenue according to the mechanism’s rules.

Unfortunately, it is impossible to achieve all desired properties in the plain model.
In fact, the impossibility already follows from only four of them: UIC, MIC, $1$-CRHP, and positive miner revenue. 
Any TFM satisfying all four must be \emph{degenerate}, meaning that it confirms bids only when there are no more than $k$ bids.
Formally,
\begin{theorem}[Impossibility in the plain model]
\label{thm:intro-impossibility-plain}
    Let $k$ be the finite block size. In the plain model, any UIC, MIC, and $1$-CRHP TFM must be degenerate: it confirms no bids when there are more than $k$ bids.
    Moreover, miner revenue must always be zero.
\end{theorem}
In fact, even if we relax $1$-CRHP to weak $1$-CRHP, the miner revenue still must always be zero. 
\begin{theorem}[Impossibility under weak-CRHP]
    \label{intro:impossibility-weak-cRHP}
    The miner revenue in any UIC, MIC, and weak $1$-CRHP TFM in the plain model must always be zero. 
\end{theorem}
\begin{corollary}
    Only trivial TFMs, where no bids are ever confirmed, satisfy UIC, MIC, MRHP, and (weak) $1$-CRHP in the plain model.
\end{corollary}
\revision{On the other hand, \cref{thm:intro-impossibility-plain} is tight: any three of the four desired properties can be achieved simultaneously.}
We present these witness mechanisms in \cref{sec:plain-ub}.
\begin{theorem} 
\label{thm:intro-ub}
    Any three of the four properties of UIC, MIC, $1$-CRHP, and positive miner revenue can simultaneously be achieved.
\end{theorem}

\newcommand{\cmark}{\textcolor{darkgreen}{\ding{51}}}
\newcommand{\xmark}{\textcolor{darkred}{\ding{55}}}
\newcommand{\unknown}{\textcolor{black!50}{--}} 

\renewcommand{\arraystretch}{1.2}
\renewcommand{\arraystretch}{1.15}
\paragraph{Landscape in the MPC-assisted Model.}
In the MPC-assisted model, $M$ miners jointly execute the TFM and share the total revenue. Strategic miners can still inject fake bids, but the MPC prevents them from altering the computed outcome once all bids are submitted. 
It is convenient to think of an ideal functionality $\mathcal{F}$ that honestly implements the mechanism as follows:

Fix an arbitrary strategic player $\C$, which in the MPC-assisted model could be an individual user, a coalition of \revision{$m<M/2$ miners}, or a miner-user coalition of \revision{$m<M/2$ miners} and some users. 
Honest users submit their bids to $\mathcal{F}$. After seeing these honest bids, the strategic coalition $\C$ decides what bids to submit. Then $\mathcal{F}$ computes the TFM outcome (which bids are confirmed, payments, miner revenue) based on all submitted bids \revision{and sends the outcome to all players}. 
This ideal functionality can be instantiated by real-world protocols \cite{gmw87,uc} \revision{with guaranteed output delivery when $m<M/2$ miners are colluding}. 

In the MPC-assisted model, we can simultaneously achieve all the desire properties, UIC, MIC, all three RHP variants, as well as positive miner revenue.
\begin{mdframed}
\begin{mechanism}[{Posted-price with random selection}]
\label{intro-mec:posted-w-random}
    Let ${\sf res}>0$ be the reserve and $0<\eps\leq \res$ be the revenue parameter. Let $E$ be the set of bids at least ${\sf res}$. Randomly choose $\min\{k,|E|\}$ bids from $E$ to include and confirm. Each confirmed bid pays the reserve price ${\sf res}$. The miner gets total revenue $\epsilon$ if at least one bid is confirmed, and zero otherwise.
\end{mechanism} 
\end{mdframed}

\revision{
\begin{theorem}[Feasibility in the MPC-assisted model]
\label{thm:intro-mpc}
    \Cref{intro-mec:posted-w-random} (posted-price with random selection) achieves UIC, MIC, URHP, MRHP, $d$-CRHP for any $d\leq k$, and positive miner revenue in the MPC-assisted model.
\end{theorem}
}
We stress that practical realization of \Cref{intro-mec:posted-w-random} remains simple.
It requires a censorship resistance inclusion list and an independent non-biasable randomness beacon. We discuss this further in \cref{sec:pratical}. 

A natural question is whether the use of randomness is necessary for the above result. Can we achieve all properties with a deterministic mechanism? Interestingly, there is a deterministic degenerated mechanism that achieves all properties in the MPC-assisted model. 
In fact, this is unavoidable for any deterministic mechanism:
\begin{theorem}[Characterization of deterministic mechanism in the MPC-assisted model]
\label{thm:intro-impossibility-mpc}
    Let $k$ denote the block size. 
    Any deterministic UIC and URHP TFM in the MPC-assisted model must be degenerate: no bid can be confirmed whenever the number of bids exceeds $k$.
\end{theorem}

\paragraph{IC and RHP are Incomparable.}
The previous first-price and second-price examples already show that UIC and URHP are incomparable.
More generally, this incomparability holds for each of the three strategic roles:
\begin{theorem}[IC and RHP are incomparable]
\label{thm:intro-incomparable}
In both the plain and the MPC-assisted model, IC and RHP are incomparable.
For each strategic role $X$ (user, miner, or miner--user coalition), there exist mechanisms that are $X$-IC but not $X$-RHP, and vice versa.
\end{theorem}
The mechanisms demonstrating this incomparability are given in \Cref{sec:compare}, where we summarize the results in \Cref{tab:compare-plain} for the plain model and \Cref{tab:compare-mpc} for the MPC-assisted model.

\paragraph{Organization.}
We give an overview of the techniques used in our impossibility results in \Cref{sec:overview}.
The detailed model and preliminaries are given in \Cref{sec:model}.
Characterizations in the plain model and the MPC-assisted model are given in \Cref{sec:plain} and \Cref{sec:mpc}, respectively.
We present the comparison between IC and RHP in \Cref{sec:compare}. Our design rationale and additional related work appear in \Cref{sec:pratical,sec:related-work}.

\subsection{Design Implications and Rationale}
\label{sec:pratical}

\smallskip\noindent\textbf{Implications for protocol design.}
Our results identify a concrete design tradeoffs for TFMs.
In today's mainstream plain model, if congestion can occur (i.e., the number of pending transactions exceed the block capacity $k$),
then no mechanism can simultaneously satisfy all four desiderata---UIC, MIC, CRHP, and positive miner revenue.
Thus, designers must either provision the system so that the capacity constraint is effectively non-binding
(e.g., operating in an ``always non-congested'' regime), or accept an explicit compromise among these properties during congestion.

Technically, the main driving force behind our plain-model impossibilities is the miner's ability to \emph{censor} honest bids without reducing its own payoff.
This reveals a practical route around the barrier in the MPC-assisted model.
\revision{
Because bids are public, realizing \Cref{intro-mec:posted-w-random} does not require private computation.
Under the honest-majority assumption, it suffices to have 1.) a censorship-resistant inclusion mechanism that establishes the set of submitted bids, and 2.) an independent non-biasable randomness beacon sampled after the submitted bids are fixed.
Once the bid set and randomness are fixed, the prescribed random selection is publicly computable and verifiable, while the honest-majority committee ensures that the resulting outcome is executed and delivered.
Thus, the mechanism remains lightweight and does not inherently require a generic MPC implementation.
Committee-based inclusion proposals such as FOCIL~\cite{eip7805} and public-verfiable randomness~\cite{ethereum_randao_docs} already provide natural building blocks toward such a realization.
}

\smallskip\noindent{\bf Prioritize UIC over coalition-IC.}
For finite block size, prior work~\cite{foundation-tfm,crypto-tfm} implies that UIC and $d$-SCP (incentive compatibility against miner-user coalitions with up to $d$ users), for $d\ge2$, can coexist in either the plain or the MPC-assisted model only if every user and miner has zero expected utility under honest behavior on every valuation profile.
Moreover, other collusion-resistance goals such as \emph{off-chain agreement proofness} (OCA-proofness~\cite{roughgarden2021transaction}) remain impossible with UIC and MIC during congestion~\cite{chung2024collusion}.
Given this tension between UIC and coalition resistance, we prioritize UIC and ask what protection can still be provided against miner--user coalitions.
From a practical perspective, UIC avoids the bid guessing in mechanisms like first-price auctions, which leads to poor user experience~\cite{lopp2017feeestimation}; Ethereum's move to the posted-price structure in EIP-1559 was motivated in part by simplifying fee estimation for users~\cite{eip1559}.
CRHP provides a complementary residual guarantee: even if a miner--user coalition can profit from deviating, it cannot obtain this benefit at the expense of protected honest participants.

\smallskip\noindent{\bf Comparing CRHP with other coalition-resistant notions.}
In addition to SCP and OCA-proofness, a broad TFM literature studies alternative notions of coalition-resistance and their
compatibility with UIC. As we already shown in \Cref{thm:intro-incomparable}, CRHP is incomparable to SCP.
OCA-proofness is primarily concerned with preventing the miner and \emph{all} users from ``stealing from the protocol'',
but it does not in general rule out deviations by a miner colluding with a subset of users that harm non-colluding users.
To capture such off-chain coordination more directly, \cite{ganesh2025characterizing,GaneshThomasWeinberg-EC24} study \emph{off-chain influence-proofness} (OffC-IP), which rules out deviations where the miner shifts price discovery or bidding off-chain to extract additional profit via influence/threat channels.
A different weakening is \emph{(IC+IR)-collusion resilience}~\cite{ferreira2024incentive}, which restricts attention to
off-chain agreements implementable by side mechanisms that remain incentive-compatible and individually rational for the colluding players.
In contrast to these notions, CRHP approaches coalition-resistance from a different angle: it does not attempt to preclude profitable collusion, but rather, rules out coalitions imposing \emph{costless harm} on
non-colluding users.

\smallskip\noindent{\bf Scope of the posted-price with random selection.}
\Cref{intro-mec:posted-w-random} satisfies UIC, MIC, and all variants of RHP while achieving positive revenue in the MPC-assisted model, but it does \emph{not} satisfy SCP.
\revision{To see this, suppose no honest user has value at least $\res$, and consider a coalition containing $m<M/2$ miners and one user with value $v\in(\res-m\epsilon/M,\res)$.

Under honest behavior, no bid is eligible, so no transaction is confirmed and the coalition obtains utility zero.
If the colluding user instead bids $\res$, its transaction is confirmed and pays $\res$, while the colluding miners receive a total revenue share of $m\epsilon/M$.
The coalition's utility becomes $v-\res+\frac{m\epsilon}{M}>0.$
Thus, the deviation is strictly profitable and the mechanism is not SCP.}

Importantly, this deviation does not harm any protected participant: all honest users remain unconfirmed with utility zero, while miners outside the coalition gain positive revenue.
This illustrates the distinction between SCP and CRHP.
CRHP does not rule out every profitable coalition deviation; rather, it ensures that a rational coalition cannot profit at the expense of protected honest participants.

\subsection{Additional Related Work}
\label{sec:related-work}
\smallskip\noindent\textbf{Transaction fee mechanisms.} 
There is a rich line of literature of TFM focusing on different models considered in this work.
In response to the UIC + SCP impossibility results, \cite{reasonable-tfm, ganesh2025characterizing, chen2025bayesian}
relax ex-post incentive compatibilities and instead study Bayesian notions of incentive compatibility. 

\revision{
Conversely, \cite{crypto-tfm} studies approximate-IC as another way to circumvent impossibility results for exact incentive compatibility.
Approximate IC and RHP constrain different aspects of deviations.
An $\epsilon$-IC guarantees that strategic players gain at most $\eps$ more utility compared to honest behavior by deviating.
But approximate-IC does not bound the harm that such a deviation may impose on honest participants.
RHP instead constrains this external harm whenever the deviation is rational, while allowing the deviator's utility gain itself to be arbitrarily large as long as protected participants are not harmed.
Thus, neither notion substitutes for the other.

Prior work also shows limitations of approximate IC under finite block size.
In the plain model, \cite{crypto-tfm} proves a scalability barrier for mechanisms satisfying $\epsilon$-UIC, $\epsilon$-MIC, and $\epsilon$-SCP.
In the MPC-assisted model, the diluted posted-price mechanism of~\cite{crypto-tfm} satisfies exact UIC and MIC together with $\epsilon$-SCP.
It pads the candidate pool with dummy bids before random selection and may therefore leave block capacity unused even when at least $k$ eligible real bids are present.
Our randomized posted-price mechanism also achieves exact UIC and MIC, but provides $d$-CRHP for $1\le d\le k$ as its coalition-side guarantee.
It confirms exactly $\min\{k,\ell\}$ eligible submitted bids when there are $\ell$ such bids, without padding the candidate pool.}

\revision{
\smallskip\noindent{\bf Randomness in the plain model.}
Our plain model follows the strategy space of~\cite{foundation-tfm, chung2024collusion}: a strategic miner controls which bids enter the block and may choose arbitrary randomness and deviate from any prescribed randomized inclusion rule.
A public randomness beacon alone does not enforce that the miner applies the prescribed random selection to the intended set of candidate bids.
An alternative model, as in~\cite{gafni2024barriers}, allows the miner to omit bids and inject fake bids but evaluates deviations under the prescribed randomized allocation rule on the resulting bid vector.
Unlike our plain model, it does not give the miner direct control over the realized allocation randomness.
Characterizing RHP under this alternative strategy space is an interesting direction for future work.
}

\smallskip\noindent{\bf TFM for other consensus.}
Many papers also study TFM for other types of consensus protocols.
\cite{garimidi2025transaction} and \cite{stouka2025multiple} consider blockchain models in which multiple block producers participate in block construction such as DAG-based consensus (see survey \cite{raikwar2024sok}) and FOCIL~\cite{eip7805}. 
Another line of work extends the classical single-resource TFM model
to settings in which transactions consume multiple scarce blockchain resources.
In these multidimensional fee markets,
distinct resources (e.g., computation, data availability, or storage)
are priced separately, and transactions may require varying quantities of each.
\cite{diamandis2023designing} proposes formal multidimensional fee mechanisms,
while \cite{angeris2024multidimensional, lavee2025does}
analyze their efficiency and optimality properties.
Empirical analyses of Ethereum’s EIP-4844 (blob) fee market
appear in \cite{heimbach2025early}.

\smallskip\noindent\textbf{Robustness to secondary goals, non-bossiness, and bounded maximin fairness.}
There is a rich literature on making mechanisms robust to \emph{secondary goals} beyond direct payoff~\cite{carroll2019robustness,jehiel2005allocative,lopomo2021uncertainty,leme2023nonbossy}.
For example, \cite{jehiel2005allocative} studies auctions with externalities where bidders' utilities depend on others' allocations or information, and recent work formalizes robustness to hidden secondary goals via non-bossy mechanisms~\cite{leme2023nonbossy}.
We refer the reader to Carroll's survey~\cite{carroll2019robustness}.

RHP is also related to the classical notion of \emph{non-bossiness}~\cite{satterthwaite1981strategy,thomson2016non,duque2024local}, which restricts an agent from changing others' outcomes without changing its own.
Most formulations focus on allocations, while TFMs involve both allocation and payments.
Utility-based variants have also been studied~\cite{Klaus2001,saijo2007secure,ritz1983restricted,leme2023nonbossy}, but they typically concern individual agents, restrict deviations that leave the deviator's own outcome unchanged, and do not distinguish whether the induced change to others is harmful or beneficial.
In contrast, RHP operates directly on utility, applies to both individual and coalitional deviations, and specifically protects honest participants against harmful deviations that are no worse for the deviator.

Recently, \cite{cryptoeprint:2025/1086} introduced the notion of \emph{bounded-maximin fairness} in the context of atomic swaps.
This notion ensures honest parties have non-negative utility from any polynomial-time profitable strategy performed by any coalition with arbitrary, but bounded external incentives, such as locked collateral, outside the mechanism's allocations.
In contrast, our notion uses the honest outcome as a baseline for honest parties and does not depend on external incentives.


\revision{
\smallskip\noindent\textbf{Griefing and griefing factors.}
RHP is closely related to the blockchain literature on \emph{griefing}, where a participant takes an action that harms other participants even when doing so may also impose a cost on the attacker.
Buterin introduced \emph{griefing factors} as a way to quantify the harm inflicted on others relative to the attacker's own loss~\cite{buterin2017triangle,buterin2018griefing}.
This perspective has subsequently been studied in blockchain settings including mining games~\cite{leonardos2023griefing}, payment channels~\cite{mazumdar2023strategic}, and cross-chain transactions~\cite{xue2021hedging}; see also the broader analysis in~\cite{george2023griefs}.

RHP addresses a complementary boundary of this concern.
Griefing-factor analyses ask how much harm can be imposed relative to the attacker's cost, including attacks for which the attacker is willing to incur a positive utility loss.
RHP does not rule out all such costly griefing attacks.
Instead, it requires that whenever a deviation is no worse for the deviator according to the modeled utility, including both utility-neutral and strictly profitable deviations, the deviation cannot make a protected honest participant worse off.
Thus, RHP focuses on the particularly conservative case in which harming others requires no sacrifice in the deviator's modeled utility.
}

\smallskip\noindent\textbf{Miner deviations and credible auctions.}
Akbarpour and Li~\cite{credibleauction} introduce \emph{credible auctions}, requiring that the auctioneer (miner) has no incentive to \emph{safely} deviate from the prescribed mechanism, where a deviation is \emph{safe} if it admits a plausible explanation to all users. 
Credibility is related to, but distinct from, MIC. 
As shown in~\cite{foundation-tfm}, a TFM can be MIC yet fail credibility. 
Moreover,~\cite{credibleauction} proves a trilemma: a revenue-optimal auction cannot be simultaneously \emph{credible}, \emph{strategy-proof}, and \emph{static}. 
By comparison, our impossibility results target a different scope: we do not impose revenue optimality as an objective. 
Subsequent work~\cite{credibleauction-comm00,credibleauction-comm01} demonstrates that cryptographic techniques can circumvent the trilemma impossibility, but these constructions do not ensure harm-prevention as provided by RHP or explicitly address the miner-user coalition in TFMs.

\section{Technical Overview} 
\label{sec:overview}
This section highlights the core ideas behind our results. 
The main technical challenge in analyzing RHP is that we must track not only how a deviation changes the deviator’s own utility, but also its effect on every honest participant’s utility. This is different from the analysis of incentive compatibility, which considers only the deviator’s payoff. 
We outline below the key idea under  our impossibility theorems and main feasibility in the MPC-assisted model. 

\smallskip\noindent{\bf Notation.} For a bid vector $\bid = (b_1,...,b_n)$, we use $\max(\bid)$ to denote $\max_{i\in[n]}(b_i)$. 
We use $\bid_{-i}$ to denote $(b_1,\dots,b_{i-1},b_{i+1}\dots,b_n)$, and we use $(\bid_{-i},b_i)$ and $\bid$ interchangeably. 
For any $i\in[n]$, we use $x_i(\bid)$, $p_i(\bid)$, and $\util_i(\bid)$ to denote the confirmation probability, expected payment, and expected utility, respectively, of user $i$'s true value $b_i$, when everyone behaves honestly with an input bid vector $\bid$.
We use $\mu(\bid)$ to denote the expected miner revenue in the honest case.

We focus on bids with \emph{unique} values since we only require weak symmetry for the mechanisms for the impossibility. We call a bid $b_i$ \emph{unique in $\bid$} if $b_j \neq b_i$ for all $j\neq i$. For a bid vector $\bid$ and a unique value $z$ in $\bid$, let $\operatorname{idx}_{\bid}(z)$ be the index $i$ such that $b_i=z$. We write $\widetilde{x}_z(\bid):=x_{\operatorname{idx}_{\bid}(z)}(\bid)$, $\widetilde{p}_z(\bid):=p_{\operatorname{idx}_{\bid}(z)}(\bid)$, and $\widetilde{\util}_z(\bid):=\util_{\operatorname{idx}_{\bid}(z)}(\bid)$. These value-indexed quantities are used only when $z$ is unique in $\bid$.
For two sets $S$ and $T$, we use $S \prec T$ to show that $\sup S < \inf T$.

\subsection{Impossibility of UIC+MIC+CRHP in the Plain Model}
\label{sec:overview-impossibility}
We start by sketching the proof of \cref{thm:intro-impossibility-plain} that any TFM satisfying UIC, MIC, and 1-CRHP in the plain model cannot confirm any bids when the input bid vector $|\bid| > k$. We first explain the proof for deterministic TFMs to illustrate the main idea, and then discuss how to generalize it to randomized TFMs.

\smallskip\noindent
{\bf Impossibility for Deterministic TFM.} 
The proof uses an induction on a carefully constructed sequence of bid vectors. 
We use a simple \emph{toy example} to illustrate the key idea behind this sequence.

\smallskip\noindent{\it Toy Example:} 
Consider a 1-CRHP TFM. Consider two users with true values $\bid = (b_1, b_2)$. We assume that $b_1 < b_2$, making them unique in $\bid$. Furthermore, to isolate the core idea behind the proofs, we assume that in the honest execution user~2’s bid is confirmed ($x_2(\bid)=1$) and pays $p_2(\bid) < b_1$, 
so user~2 gets a positive utility.
Under these conditions, user~1 must also be confirmed in the honest execution. 
If not, consider a coalition of the miner and user~1. In the honest outcome, their joint utility is just the miner’s revenue $\mu(\bid)$. Now the coalition performs this {\bf censor-then-replace} attack:
the miner ignores user~2’s bid $b_2$, and user~1 changes its own bid to $b_2$ while also injecting a fake bid of value $b_1$. The resulting bid vector is still $\bid$ but now both $b_1$ and $b_2$ come from user~1. By weak symmetry, the bid of value $b_2$ (now coming from user~1) would still be confirmed and charged $p_2(\bid)$, while the bid of value $b_1$ is not confirmed. Under this strategy, the coalition’s joint utility strictly increases: the miner’s revenue is unchanged, and user~1 now gets utility $b_1 - p_2(\bid) > 0$, while the honest user~2 is harmed. This contradicts 1-CRHP. Thus $b_1$ must be confirmed in the honest execution.

The above toy example can be generalized to a useful ``above-payment-guarantee'' lemma: \emph{In an honest execution of a (weak) 1-CRHP TFM, if some unique bid $b_i$ earns positive utility, then \emph{any} unique bid in $b$ that is higher than $p_i(\bid)$ must also be confirmed.} Intuitively, if a higher bid is not confirmed, then it can collude with the miner and perform the censor-then-replace attack on $b_i$, which contradicts $1$-CRHP. 
Here, we require uniqueness of the bid to get rid of the potential tie breaking that can depend on the metadata.

\smallskip\noindent{\bf Induction.}
Next, we show how to prove \cref{thm:intro-impossibility-plain} for deterministic TFM.
Assume for sake of contradiction that there exists a deterministic TFM satisfying UIC, MIC, and 1-CRHP in which some bid $v_{i^*}$ is confirmed when the input bid vector is $\val$ with $|\val| >k$. 
We will construct a sequence of bid vectors $\val^{(m)}$ for $m = 1,2,...$, and prove by induction that $m$ bids must be confirmed in $\val^{(m)}$. 
This directly contradicts the block size $k$ when $m = k+1$.
The sequence is constructed as follows:
\begin{align*}
b_1& =  \max(\val) + 1,&  \val^{(1)} &= (\val_{-i^*}, b_1), \\
b_m &= \frac{1}{2}(\max(\val) + b_{m-1}),& \val^{(m)} &= \left(\val^{(m-1)}, b_m\right) \quad \text{for } m> 1.
\end{align*}
The inductive invariant states that 
\begin{center}
    $b_m$ must be confirmed in $\val^{(m)}$ and its payment $\widetilde{p}_{b_m}(\val^{(m)}) \leq \max(\val)$.
\end{center}
Effectively, 
in each step, we add a bid $b_m$ that lies between the payment $\widetilde{p}_{b_{m-1}}\left(\val^{(m-1)}\right)\leq \max(\val)$ and $b_{m-1}$ to $\val^{{(m-1)}}$. 
By the inductive statement, $b_m$ gets a positive utility in $\val^{(m)}$, and all $b_z$ for $z\in[m]$ is higher than $b_m$'s payment. 
Thus, by the above-payment-guarantee lemma, all $b_z$ for $z\in[m]$ must be confirmed. 

The induction proof uses Myerson's lemma (\cref{Myerson's}). Roughly speaking, for a deterministic UIC TFM, for any fixed other users' bids $\bid_{-i}$, there exists a critical value $v^*$ such that any $b_i > v^*$ gets confirmed and pays $v^*$, i.e., $x_i(b_i,\bid_{-i}) = 1$ and $p_i(b_i,\bid_{-i}) = v^*$ for any $b_i > v^*$.

\smallskip\noindent\underline{\it Base case:}
The base case follows directly from UIC and the Myerson's lemma: 
Raising user $i^*$'s bid in $\val$ to $\max(\val)+1$ must maintain its confirmation probability and payment.
Thus, $\widetilde{x}_{b_1}(\val^{(1)}) = 1$ and $\widetilde{p}_{b_1}(\val^{(1)})  = p_{i^*}(\val)\le \max(\val)$ where the inequality follows from individual rationality.

\smallskip\noindent\underline{\it Inductive step:}
Now assume we have $\val^{(m-1)}$ where $b_{m-1}$ is confirmed and pays no more than $\max(\val)$.
By Myerson's Lemma, it suffices to show that for any $b$ such that $\max(\val) < b < b_{m-1}$, 
bid $b$ must be confirmed in $(\val^{(m-1)},b)$.
This implies that $\widetilde{x}_{b_m}(\val^{(m)}) = 1$ and $\widetilde{p}_{b_m}(\val^{(m)}) \le \max(\val)$.

For the sake of contradiction, assume that there exists some $\max(\val) < b^* < b_{m-1}$ such that $b^*$ is not confirmed in $(\val^{(m-1)},b^*)$.
Imagine a world where $\val^{(m-1)}$ is the honest bid vector.
Since $|\val^{(m-1)}|\ge |\val|> k$, some user $j$ must be unconfirmed in the honest case.
Consider a coalition consisting of the miner and user $j$ who injects a fake bid $b^*$.
Since $b^*$ is unconfirmed, the miner revenue must remain unchanged by MIC: $\mu(\val^{(m-1)}) = \mu(\val^{(m-1)},b^*)$.
Therefore, this strategy does not decrease the coalition's utility compared to the honest case, since user $j$'s original bid still gives nonnegative utility by individual rationality.\footnote{This is why this proof does not work for the impossibility w.r.t.  weak $1$-CRHP. \label{weak-challenge}}
That said, by $1$-CRHP, this strategy should harm no honest bid in $\val^{(m-1)}$, including $b_{m-1}$.
Therefore, $b_{m-1}$ should still get a positive utility in $(\val^{(m-1)},b^*)$ and pays no more than $\max(\val)$.
However, by the above-payment-guarantee, since $b^*$ is greater that the payment of $b_{m-1}$ in $(\val^{(m-1)},b^*)$, we know that $b^*$ must be confirmed, which contradicts our assumption. 
This completes the induction proof.


\smallskip\noindent{\bf Generalizing to Randomized TFMs.} 
The above proof idea needs significant adaptation for randomized mechanisms. 
If the TFM uses randomness, then having more than $k$ bids, each with a positive \emph{confirmation probability} is no longer an immediate contradiction. 
In a randomized TFM, different random outcomes might confirm different sets of at most $k$ bids.
Therefore, we derive a stronger condition called the \emph{$k$-winners lemma}: If a TFM is MIC and (weak) $1$-CRHP, then it must be that
\begin{equation}
    \label{eqn:overview-k-winner}
    |W(\bid)|\leq k, \text{ where } W(\bid):= \{i \colon \util_i(\bid) >0\}.
\end{equation}
Here, $W(\bid)$ denotes the set of \emph{winners} who have positive expected utility in the honest case when the input bid vector is $\bid$. 
Note that positive \emph{utility} is a strictly stronger requirement than positive \emph{confirmation probability}.
We prove this lemma formally in \cref{sec:impossibility}. 
For the overview, we assume it is true and show how to prove \cref{thm:intro-impossibility-plain}.

For possibly randomized TFMs, the ``above-payment-guarantee'' can be generalized to ``above-ratio-guarantee'': \emph{In a 1-CRHP TFM , if 
a unique bid $b_i$ has positive expected utility in $\bid$, then any unique bid in $\bid$ above the ratio $p_i(\bid)/x_i(\bid)$ must have a positive confirmation probability}. 
In the deterministic case, this ratio $p_i(\bid)/x_i(\bid)$ is simply the payment itself. We define, for any bid $b_i$ such that $x_i(\bid)>0$, the \emph{ratio} as
\[r_i(\bid):=p_i(\bid)/x_i(\bid), \qquad \text{and if $b_i$ unique, } \widetilde{r}_{b_i}(\bid):=\widetilde{p}_{b_i}(\bid)/\widetilde{x}_{b_i}(\bid).\]
If a unique bid $b > \widetilde{r}_{b}(\bid)$, then $b$ gets a positive utility in $\bid$.

To prove \cref{thm:intro-impossibility-plain}, suppose for the sake of contradiction, that for some UIC, MIC, and $1$-CRHP TFM, 
there exists some $\val$ where $|\val| = n >k$ and $x_{i^*}(\val) > 0$ for some $i^* \in[n]$. As in the deterministic case, we will construct a sequence of bid vectors $\val^{(m)} = (\val^{(m-1)},b_m)$, each adding a new bid $b_m$ to $\val^{(m-1)}$, where $\widetilde{r}_{b_{m-1}}(\val^{(m-1)}) < b_m < b_{m-1}$. 
We then prove that $|W(\val^{(m)})|\geq m$ by induction, which leads to a contradiction of the $k$-winner lemma for $m > k$.

\smallskip\noindent{\bf Challenges.}  
However, simply adding one bid in each step is {\bf not} enough.
Using the above-ratio-guarantee, we can argue similarly as the deterministic case that $b_m$ gets a positive confirmation probability in $(\val^{(m-1)}, b_m)$.
Then raising $b_m$ to some $b'_m$ ends up with a positive \emph{utility} for $b'_m$ in $(\val^{(m-1)}, b'_m)$ by a corollary of Myerson's lemma.

The challenge is that raising $b_m$ to $b'_m$ can alter payments for others: this naive approach might make earlier winners $b_1,...,b_{m-1}$ lose their utility, breaking the inductive invariant. Note that $1$-CRHP does not rule out such strategies since raising $b_m$ to $b'_m$ may harm the coalition themselves in a randomized TFMs. 
We overcome this by carefully choosing a set $S_m$ of two possible values $\alpha^{(m)} < \beta^{(m)}$ that the new bid $b_m$ can take.
Effectively, we build a sequence of \emph{sets} of bid vectors $V^{(m)} = \{\val_{-i^*}\}\times S_1\times\dots\times S_m = \{(\val_{-i^*}, b_1,...,b_m):b_j\in S_j\text{ for }j\in[m]\}$. 
The induction will show that for each $m$, we have some $\bid \in V^{(m)}$ such that $|W(\bid)|\ge m$. 


\smallskip\noindent\underline{\it Base case:} 
Pick $\max(\val) < \alpha^{(1)} < \beta^{(1)}$ to be two values large enough such that $b_1$ has a positive utility in $(\val_{-i^*},b_1)$ for any $b_1\in S_1= \{\alpha^{(1)}, \beta^{(1)}\}$. Such values must exist since $x_{i^*}(\val) > 0$. Specifically, choose $v^*>\max(\val)$, put $\tau=\widetilde r_{v^*}(\val_{-i^*},v^*)<v^*$, and choose $\max\{\max(\val),\tau\}<\alpha^{(1)}<\beta^{(1)}<v^*$. Monotonicity of the ratio then gives the bound $\widetilde r_{b_1}(\val_{-i^*},b_1)\le\tau<\alpha^{(1)}$. Let $V^{(1)} = \{\val_{-i^*}\}\times S_1$. Then for any $\bid\in V^{(1)}$, we have $|W(\bid)|\geq 1$.




\smallskip\noindent\underline{\it Inductive step:} 
We demonstrate the main idea with $m = 2$, and summarize the general construction afterwards.
Let $r^{(1)}:=\underset{\val^{(1)}\in V^{(1)}}{\max} r_{n+1}(\val^{(1)},\ \alpha^{(1)})$ and $\alpha^{(2)}$ and $\beta^{(2)}$ 
be such that $\max\{\max(\val),r^{(1)}\} < \alpha^{(2)} < \beta^{(2)} < \alpha^{(1)}$.
By a similar argument as in the deterministic case, $\alpha^{(1)}$ must have a positive confirmation probability in $(\val^{(1)},\alpha^{(1)})$ for any $\val^{(1)} \in V^{(1)}$, and therefore, $r^{(1)}$ is well-defined.
More precisely, the censor-then-replace argument gives positive confirmation probability to every newly appended bid $t\in(\alpha^{(1)}-\epsilon,\alpha^{(1)})$ for sufficiently small $\epsilon>0$, independent of the choice $\val^{(1)}\in V^{(1)}$ (\Cref{claim:confirmed}). Raising this bid to $\alpha^{(1)}$ at the fixed index $n+1$ gives strictly positive utility, even when it ties the existing bid $b_1$. Thus each ratio $ r_{n+1}(\val^{(1)},\ \alpha^{(1)})$ for $\val^{(1)}\in V^{(1)}$ is strictly below $\alpha^{(1)}$, so $r^{(1)}<\alpha^{(1)}$.
The value $r^{(1)}$ is an upper bound for the ratio $\widetilde{r}_b(\val^{(1)},b)$ for any bid $b < \alpha^{(1)}$ with a positive confirmation probability.
This maximum guarantees that the new bid $b_2\in S_2 = \{\alpha^{(2)}, \beta^{(2)}\}$ \emph{always} lies between $b_1\geq \alpha^{(1)}$ and its corresponding ratio, no matter which value $b_1$ takes on. 
Let $V^{(2)} = \{\val_{-i^*}\}\times S_1 \times S_2$.
We show that 
\begin{align}
    &\text{for any }\val^{(2)} = (\val_{-i^*},b_1,b_2)\in V^{(2)},\  \widetilde{\util}_{b_2} (\val^{(2)}) > 0 \text{ for any } b_2\in S_2;\label{eqn:overview-new-pos}\\
    &\text{for any } i \in [2],\  \widetilde{\util}_{\beta^{(i)}} \left(\val_{-i^*}, \beta^{(1)},\beta^{(2)}\right) > 0, \label{eqn:overview-old-pos}
\end{align}
which completes the induction proof for $m=2$ since $|W(\val_{-i^*}, \beta^{(1)},\beta^{(2)})|\geq 2$.

Property (\ref{eqn:overview-new-pos}) directly follows from UIC, since any bid $b > r^{(1)}$ must have a positive utility in $(\val^{(1)},b)$ for any $\val^{(1)} \in V^{(1)}$. 
Otherwise, when the honest bid vector is $(\val^{(1)},b)$, the user with true value $b$ is incentivized to overbid to $\alpha^{(1)}$ to gain a positive utility.


To see why (\ref{eqn:overview-old-pos}) is true, consider $\bid' =(\val_{-i^*}, \alpha^{(1)},\beta^{(2)})$. Since $\beta^{(2)}$ has a positive utility in $\bid'$ by (\ref{eqn:overview-new-pos}), and $\alpha^{(1)} > \beta^{(2)} > r^{(1)} \geq \widetilde{r}_{\beta^{(2)}}(\bid')$,
we know $\alpha^{(1)}$ must have a positive confirmation probability in $\bid'$by the above-ratio-guarantee. 
Therefore, raising $\alpha^{(1)}$ to $\beta^{(1)}$ gives $\beta^{(1)}$ a positive utility. This completes the proof for $m=2$.


In general, at each step $m$, we define $r^{(m-1)} = \underset{\val^{(m-1)}\in V^{(m-1)}}{\max}\ r_{n+m-1}\left(\val^{(m-1)},\alpha^{(m-1)}\right)$ as an upper bound, over all possible $\val^{(m-1)}\in V^{(m-1)}$, of the ratio that a new bid smaller than $\alpha^{(m-1)}$ could have if it gets a positive confirmation probability. For every $\val^{(m-1)}\in V^{(m-1)}$, the newly appended bidder has strictly positive utility when bidding $\alpha^{(m-1)}$, and hence
\[
r_{n+m-1}\left(\val^{(m-1)},\alpha^{(m-1)}\right)<\alpha^{(m-1)}.
\]
Since $V^{(m-1)}$ is finite, it follows that $r^{(m-1)}<\alpha^{(m-1)}$. We then pick $\max\{\max(\val),r^{(m-1)}\} < \alpha^{(m)} < \beta^{(m)} < \alpha^{(m-1)}$ and define $S_m = \{\alpha^{(m)}, \beta^{(m)}\}$ and $V^{(m)} = V^{(m-1)}\times S_m$.
Roughly speaking, in the $m$-th step, for any $(b_1,...,b_m)\in S_1\times\dots\times S_m$, the new bid $b_m$ must get a positive utility in $(\val_{-i^*}, b_1,...,b_m)$. 
Now consider specifically $\bid' = (\val_{-i^*},\beta^{(1)}, ...,\beta^{(m)})$.
Using the above-ratio-guarantee and our choice of the sets $S_j$, we can conclude that for any $n\leq j\leq n+m-1$, 
we have $x_{j}(\bid_{-j}',\alpha^{(j-n+1)}) > 0$.
Therefore, $\util_{j}(\bid') > 0$ for each $n\leq j\leq n+m-1$.
When $m = k+1$, this directly contradicts the $k$-winners lemma. This completes the impossibility proof for randomized TFMs in the plain model.
We omit the details of the inductive reasoning here.
See \cref{sec:impossibility} for the complete proof.

\subsection{Impossibility under Weak-CRHP}
\label{sec:overview-weak-impossibility}
We first give a high-level overview of the proof techniques as this proof is quite involved. 
Recall that weak $1$-CRHP means that strictly profitable deviations must not harm honest players, but it does not rule out utility-neutral deviations that are harmful to others.

This relaxation complicates the proof as mentioned in \cref{weak-challenge}. In the $1$-CRHP setting, each step of the induction used a contradiction argument to show that a newly added bid $b_m$ must receive a positive confirmation probability: if $b_m$ were not confirmed, a miner–user coalition could instead inject $b_m$ without changing joint utility, and under $1$-CRHP such a deviation cannot harm any honest player. 
Under only weak $1$-CRHP, however, this utility-neutral strategy can harm honest players, so the earlier argument breaks down.

To overcome this, we adopt a slightly different approach. The key observation is that by MIC, if we start with some bid vector $\val$ with a positive miner revenue, then adding any set of additional bids $\bid$ will still produce $\mu(\val,\bid) > 0$. 
By budget feasibility, this implies that some bid must have a positive confirmation probability in $(\val,\bid)$. We use the following lemma to argue that this confirmed bid must be one of the newly added bids.
\begin{lemma}[no-persistent-old-winner, informal]
\label{lem:overview-winning-set}
Fix a TFM in the plain model satisfying UIC, MIC, and weak $1$-CRHP.
Suppose \( \val = (v_1,\dots,v_n) \) and a sequence of sets $S_1,...,S_k\subset \R_{\geq 0}$ such that each $|S_j| \geq 2$, and that $\{\max(\val)\} \prec S_1 \prec \dots \prec S_k$.
\revision{Then there is no $i\in[n]$ such that $x_i(\val,\bid)>0$ for all \( \bid = (b_1,\dots,b_k) \in S_1 \times \dots \times S_k \).}
\end{lemma}
Suppose, for the sake of contradiction, that there exists a UIC, MIC, weak $1$-CRHP TFM and an initial bid vector $\val$ with $\mu(\val)>0$. 
We inductively construct a sequence of bid-vector sets $V^{(m)}$ by adding $k$ new sets, instead of one set, at each step, then pruning the old sets and fixing $k-1$ fresh coordinates in a background vector $\val^{(m)}$. Thus $V^{(m)}=\{\val^{(m)}\}\times S_1^{(m)}\times\cdots\times S_m^{(m)}$, where $|\val^{(m)}|=n+(k-1)m$ and $\val\subseteq\val^{(1)}\subseteq\cdots\subseteq\val^{(m)}$.
We show by induction that for each $m$ there exists some $\bid\in V^{(m)}$, such that $\mu(\bid) > 0$ and $|W(\bid)|\geq m$.
For $m = k+1$ this yields a contradiction to the $k$-winner lemma (which remains valid under weak 1-CRHP).
In the induction proof, we will make use of the following dense Cartesian product lemma of Erd\H{o}s:
\begin{lemma}[{\cite[Corollary, p.~188]{Erdos1964}}]\label{lem:overview-G-R}
For all positive integers $p,a$ and $\xi\in(0,1]$, there exists $N$ such that, if
$T \subseteq A_1 \times \cdots \times A_p$, each $|A_i| = N$ for $i\in[p]$, 
and $|T| \ge \xi\cdot N^p$,
then there exist subsets $B_i \subseteq A_i$ with
$|B_i| \ge a$ such that
\[
B_1 \times \cdots \times B_p \subseteq T.
\]
\end{lemma}
Looking ahead, in the induction, each newly added set $A_j^{(m)}$ in the $m$-th step has size $N_m$.
The invariant maintains $\{\max(\val^{(m)})\}\prec S_m^{(m)}\prec\cdots\prec S_1^{(m)}$ and a bound $r^{(m)}<\min S_m^{(m)}$ on every ratio $\widetilde r_{b_\ell}(\val^{(\ell)},b_1,\ldots,b_\ell)$, for all $\ell\le m$ and $(b_1,\ldots,b_\ell)\in S_1^{(m)}\times\cdots\times S_\ell^{(m)}$; all these ratios have positive denominators. In particular, $b_m$ has positive utility, so the above-ratio-guarantee gives every larger $b_i$ positive confirmation probability.
The sequence of $N_m$ is defined backwards starting from $N_{k+2}=2$, such that at the end of the $m$-th step, there exist $m$ sets $S_1^{(m)},...,S_m^{(m)}$, each with size $N_{m+1}$, such that any $b_i$ in $\bid = (b_1,...,b_m)\in S_1^{(m)}\times ...\times S_m^{(m)}$ must have a positive confirmation probability in $(\val^{(m)},\bid)$. Therefore, $b_i$ must have a positive utility in all possible $(\val^{(m)},\bid)$ as long as $b_i>\min S_i^{(m)}$.
We show how to choose $N_m$ given $N_{m+1}$ in the base case and inductive step below.

\smallskip\noindent\underline{\it Base case:} Let $\{\max(\val)\} \prec A_1^{(1)} \prec...\prec A_k^{(1)}$ be $k$ sets, each of size $N_1$, where $N_1$ is the $N$ in \cref{lem:overview-G-R} with $p = k$, $\xi = \frac{1}{n+k}$, and $a = N_2$.
By MIC, $\mu(\val, \bid)>0$ for any $\bid\in A_1^{(1)}\times\dots\times A_k^{(1)}$. Therefore, some bid must have a positive confirmation probability in each $(\val, \bid)$.
For any bid vector $\widehat{\bid}$ with at least one bid of positive confirmation probability, let $\phi(\widehat{\bid})$ denote the index of a lowest-valued bid among those with positive confirmation probability, breaking ties by index.
The $n+k$ possible values of $\phi(\val,\bid)$ give a natural partition of $A_1^{(1)}\times\dots\times A_k^{(1)}$:
\[T_i = \{\bid \in A_1^{(1)}\times\dots\times A_k^{(1)}: \phi(\val,\bid)=i\} \text{ for } i\in[n+k].\]
At least one of these sets $|T_{i^*}|\geq \frac{1}{n+k}\cdot N_1^k$ for some $i^*\in[n+k]$. 
By \cref{lem:overview-G-R}, there exists $B_j^{(1)}\subseteq A_j^{(1)}$ for all $1\le j\le k$ such that for each such set, $|B_j^{(1)}|\geq N_2$, such that $B_1^{(1)}\times \dots \times B_k^{(1)} \subseteq T_{i^*}$, i.e., the $i^*$-th bid has positive confirmation probability in $(\val,\bid)$ for any $\bid\in B_1^{(1)}\times \dots \times B_k^{(1)}$.
By \cref{lem:overview-winning-set}, $i^*$ cannot be one of the old bids from $\val$, i.e., $i^* > n$. Let $\ell = i^* - n$. 

Pick an arbitrary $\bid^*\in B_1^{(1)}\times \dots \times B_k^{(1)} \subseteq T_{i^*}$, and let $\val^{(1)} = (\val, \bid^*_{-\ell})$. Choose $v^*>\max A_k^{(1)}$. Raising the selected bid to $v^*$ gives positive utility by Myerson's lemma, so $\tau:=\widetilde r_{v^*}(\val^{(1)},v^*)<v^*$. Let $S_1^{(1)}$ be a set with $N_2$ values such that $\{\max\{\max A_k^{(1)},\tau\}\}\prec S_1^{(1)}\prec\{v^*\}$, and the ratio $r^{(1)}$ (defined analogously as in 
\cref{sec:overview-impossibility})
$r^{(1)}=\widetilde{r}_{\max S_1^{(1)}}(\val^{(1)},\max S_1^{(1)})$.
By Myerson's lemma, $\widetilde{x}_{b_1}(\val^{(1)},b_1) > 0$ for any $b_1\in S_1^{(1)}$.
By monotonicity of the payment-to-confirmation ratio (\cref{cor:p/x-increasing}), $r^{(1)}\le\max\{\max A_k^{(1)},\tau\}<\min S_1^{(1)}$.




\smallskip\noindent\underline{\it Inductive step:} Pick $k$ sets $A_1^{(2)},...,A_k^{(2)}$ such that $\{r^{(1)}, \max(\val^{(1)})\}\prec A_1^{(2)}\prec\dots\prec A_k^{(2)}\prec S_1^{(1)}$ and each $|A_j^{(2)}| = N_2$.
This sequence of sets exist because $\{r^{(1)}, \max(\val^{(1)})\} \prec S_1^{(1)}$ by construction.
Here, $N_2$ is the $N$ in \cref{lem:overview-G-R} with $p = k+1$, $\xi = \frac{1}{n+2k}$, and $a = N_3$.
Still, by MIC, $\mu(\val^{(1)},\bid,b_1)>0$ for any $(\bid, b_1) \in A_1^{(2)}\times\dots\times A_k^{(2)} \times S_1^{(1)}$, meaning that some bid must have a positive confirmation probability. 
The $n+2k$ possible values of $\phi(\val^{(1)},\bid,b_1)$ give a natural partition of $A_1^{(2)}\times\dots\times A_k^{(2)}\times S_1^{(1)}$ for $ i\in[n+2k]$:

\[T_i = \left\{(\bid,b_1) \in A_1^{(2)}\times\dots\times A_k^{(2)}\times S_1^{(1)}: \phi(\val^{(1)},\bid,b_1)=i\right\}.\]
By a similar reasoning as in the base case, there must exist an $i^*\in[n+2k]$, such that there exists $B_j^{(2)}\subseteq A_j^{(2)}$ and $S_1^{(2)}\subseteq S_1^{(1)}$, each of size $N_3$, such that $B_1^{(2)}\times \dots \times B_k^{(2)} \times S_1^{(2)} \subseteq T_{i^*}$, i.e., the $i^*$'th bid has a positive confirmation probability in $(\val^{(1)},\bid, b_1)$ for all $(\bid, b_1)\in B_1^{(2)}\times \dots \times B_k^{(2)}\times S_1^{(2)}$.
The set-to-winners argument underlying \cref{lem:overview-winning-set}, with $b_1$ held fixed, shows that $i^*$ cannot be smaller than $n+k$. 
On the other hand, because $\phi$ selects a lowest-valued bid with positive confirmation probability, the censor-then-replace argument shows that $i^* < n+2k$.
That is, $i^*$ must represent one of the new bids from $ B_1^{(2)}\times\dots\times B_k^{(2)}$. Let $\ell = i^* - |\val^{(1)}|$.

Pick an arbitrary $(\bid^*, b_1)\in B_1^{(2)}\times \dots \times B_k^{(2)}\times S_1^{(2)} \subseteq T_{i^*}$, and let $\val^{(2)} = (\val^{(1)}, \bid^*_{-\ell})$. Put $\gamma:=\min S_1^{(1)}$ and $J:=|\val^{(2)}|+1$. For each $b_1\in S_1^{(2)}$, reinsert the selected fresh bid at index $J$; weak symmetry preserves its positive confirmation probability. Raising it to $\gamma$ at this fixed index gives positive utility by Myerson's lemma, even if $\gamma=b_1$. Since $S_1^{(2)}$ is finite, $\tau_2:=\max_{b_1\in S_1^{(2)}}r_J(\val^{(2)},\gamma,b_1)<\gamma$. Let $S_2^{(2)}$ be a set with $N_3$ values such that $\{\max\{\max A_k^{(2)},r^{(1)},\tau_2\}\}\prec S_2^{(2)}\prec\{\gamma\}$, and set $r^{(2)}:=\max\{\max A_k^{(2)},r^{(1)},\tau_2\}$. Ratio monotonicity preserves the invariant, so $b_2$ has positive utility and the above-ratio-guarantee applies to $b_1$. Then $b_1$ and $b_2$ must have a positive confirmation probability in $(\val^{(2)},b_1,b_2)$ for any $(b_1,b_2)\in S_1^{(2)}\times S_2^{(2)}$.

In general, for all $m \ge 2$, the inductive step proceeds analogously. Define $N_i$ backward as follows: set $N_{k+2}:=2$, and for each $i=k+1,k,\ldots,1$, let $N_i$ be the number $N$ given by \cref{lem:overview-G-R} with parameters
$p = k+i-1$, $\xi = \tfrac{1}{\,n + k i\,}$, and $a = N_{i+1}$.
This choice ensures that at the end of the $(k+1)$-th step we obtain sets $S^{(k+1)}_1,\ldots,S^{(k+1)}_{k+1}$, each of size two, such that for every
$\bid=(b_1,\ldots,b_{k+1}) \in S^{(k+1)}_1 \times \cdots \times S^{(k+1)}_{k+1}$,
each coordinate $b_i$ has a positive confirmation probability in $(\val^{(k+1)},\bid)$.
Consequently, taking $b_i^* := \max S^{(k+1)}_i$ for all $i\in[k+1]$ yields a bid vector
$(\val^{(k+1)}, b_1^*,\ldots,b_{k+1}^*)$ in which all $b_i^*$ enjoy positive utility, by a similar reasoning in \cref{sec:overview-impossibility}, which contradicts the $k$-winner lemma.

\revision{
\subsection{Feasibility in the MPC-Assisted Model}
\label{sec:overview-feasibility}
We next explain why \Cref{intro-mec:posted-w-random} simultaneously achieves all desired properties once inclusion and correct output delivery are enforced by the MPC.
For $N>0$ eligible bids, let $q(N):=\min\{1,k/N\}$ denote the confirmation probability of each eligible bid. Set $q(0):=0$.
Every confirmed bid pays the fixed reserve $\res$, while the total miner revenue is the fixed amount $\eps$ whenever the block is nonempty.
UIC and MIC follows from prior work, so we focus on explaining why it achieves $d$-CRHP.

Fix a coalition containing $1\le m<M/2$ miners and at most $d\le k$ users.
Let $h$ be the number of eligible honest users outside the coalition and $t$ be the number of eligible colluding users.
Let the sum of these eligible colluding users' utility had they be confirmed be $s$. 
Then under any deviation, besides the miner revenue, the total utility contribution from colluding bids is at most $s$.

In the honest case, as long as $h+t > 0$, colluding miners gets revenue $m\eps/M$. 
Each eligible user gets confirmed with probability $q(h+t)$, so colluding users' expected utility is $q(h+t)s$.
Let $N$ be the resulting number of eligible bids under deviation. Then $N\geq h$, and the colluding users' utility is at most $\E[q(N)]s$.

There are two possible cases. If $h > 0$, the block is always non-empty under every deviation, so the revenue of honest miners is always the same. 
If the coalition wants to harm an eligible honest user, the coalition has to decrease its expected confirmation probability, i.e., make $\E[q(N)]<q(h+t)$.
If $s>0$, the bound $\E[q(N)]s$ then makes the coalition strictly worse off. If $s=0$, every colluding user has value at most $\res$; a deviation that does not lower coalition utility must therefore use no eligible fake bids or below-reserve users' eligible bids. Hence $N\le h+t$ almost surely, so no honest user is harmed.
If $h = 0$, i.e., there is no honest eligible user.
Then only honest miners need to be protected. If $t=0$, their honest revenue is zero and cannot decrease, so assume $t>0$. This is where we need $d\leq k$.
In the honest case, since $t\leq d\leq k$, all eligible colluding users are already confirmed in the honest case.
Therefore, any rational deviation must keep the block nonempty with probability one: each nonempty outcome gives conditional coalition utility at most $m\eps/M+s$, its honest utility, while an empty outcome gives zero and $m\eps/M>0$. This guarantees honest miners' revenue.
The condition $d\leq k$ ensures that all eligible colluding users already fit in the block in the honest case, so the coalition cannot profit by dropping some colluding users to increase the confirmation probability of the others.
}

\subsection{Impossibility of UIC+URHP for Deterministic MPC-assisted TFMs}
The proof of \cref{thm:intro-impossibility-mpc} is very similar to the deterministic impossibility in the plain model.
However, we cannot directly use the above-payment-guarantee lemma because it utilizes a censor-then-replace attack, which cannot be performed in the MPC-assisted model.
Instead we prove the above-payment-guarantee for deterministic UIC and URHP mechanisms in the MPC-assisted model.
The crux of the proof is to start with a slightly weaker claim: 
\emph{for a deterministic TFM that satisfies UIC and URHP, if $b_i$ gets a positive utility in $\bid$, then any unique bid higher than $b_i$ must also be confirmed.}
Then we only need to show that any unique $b_j$ between $p_i(\bid)$ and $b_i$ must be confirmed to get the above-payment-guarantee lemma.

Suppose for the sake of contradiction that, in some deterministic UIC and URHP TFM in the MPC-assisted model, a unique bid $b_j>p_i(\bid)$ is not confirmed even though user $i$ has positive utility in $\bid$.
First let user $j$ lower its bid to some $b'_j\in(p_i(\bid),b_j)\setminus\{b_\ell:\ell\in[n]\}$.
By UIC, user $j$ still has zero utility under this deviation, so URHP implies that user $i$ remains unharmed and therefore still has positive utility.
Now choose $v\in(p_i(\bid),b'_j)$ and consider the honest profile $(\bid_{-i},v)$.
Using the weaker claim above and Myerson's lemma, user $j$ has positive utility in this profile.
However, user $i$ can deviate from the truthful bid $v$ to its original bid $b_i$ without changing its own utility, while the resulting bid vector is $\bid$, where user $j$ has zero utility.
This contradicts URHP.
This proves the above-payment-guarantee for UIC and URHP TFMs in the MPC-assisted model.
The rest of the proof follows by a similar reasoning of the deterministic case in \cref{sec:overview-impossibility}.
Here an unconfirmed user alone injects the unconfirmed fake bid, without decreasing its utility by individual rationality. The above-payment-guarantee forces the previous winner either to lose confirmation or to pay at least the fake bid, which exceeds its old payment; this harms that winner and contradicts URHP.
The proof of this impossibility is given in \cref{sec:mpc-det-imp}.

\section{Model and Preliminaries}
\label{sec:model}
\subsection{Notation}

Let $\mathbb{N}$ denote the set of natural numbers and $\mathbb{R}_{\ge0}$ denote the non-negative real numbers. 
Given a bid vector $\bid = (b_1, \ldots,b_n)$ and any $i \in [n]$, we use $\mathbf{b}_{-i}$ to denote $(b_1, \ldots, b_{i-1}, b_{i+1}, \ldots, b_n)$. We use $(\mathbf{b}_{-i}, b'_i)$ to denote the vector obtained by replacing $b_i$ with $b'_i$ in $\mathbf{b}$, and treat this as equivalent to $(b_1, \ldots, b_{i-1},b'_i,b_{i+1}, \ldots, b_n)$. For vectors $\mathbf{a}$ and $\mathbf{b}$, we write $\mathbf{a} \subseteq \mathbf{b}$ if every entry of $\mathbf{a}$ appears in $\mathbf{b}$. We use $|\mathbf{b}|$ to denote the number of elements in vector $\mathbf{b}$. For any vector $\mathbf{v} = (v_1, v_2, \dots, v_n)$, we define 
$\max(\mathbf{v}) \coloneqq \max_{i \in [n]} v_i$.
For a vector $\val$ and $S\subseteq[n]$, let $\val_S:=(v_i)_{i\in S}$.
For a bid vector $\bid$ and a value $z$ that appears exactly once in $\bid$, let $\operatorname{idx}_{\bid}(z)$ be the unique index $i$ such that $b_i=z$. We write $\widetilde{x}_z(\bid):=x_{\operatorname{idx}_{\bid}(z)}(\bid)$, $\widetilde{p}_z(\bid):=p_{\operatorname{idx}_{\bid}(z)}(\bid)$, and $\widetilde{\util}_z(\bid):=\util_{\operatorname{idx}_{\bid}(z)}(\bid)$. These value-indexed quantities are used only when $z$ is unique in $\bid$.
For two sets $S$ and $T$, we use $S \prec T$ to show that $\sup S < \inf T$.

\subsection{Transaction Fee Mechanism}
\label{sec:tfm-model}
We consider a transaction fee mechanism (TFM) where each block has a finite capacity of $k$ transactions. Assume each user has a true value of $v_i \in \mathbb{R}_{\ge0}$ that measures the maximum amount a user is willing to pay to get its transaction confirmed in the block.
Each user $i$ submits their transaction together with a \emph{bid} $b_i$. 
We assume each transaction takes one slot in the block.
In this paper, ``bid'' and ``transaction'' are used interchangeably.
A TFM with capacity $k$ is defined as a tuple $\tfm$ where: 
\begin{itemize}[leftmargin = 5mm]
    \item {\bf Inclusion rule ${\bf I}(\cdot)$}: takes a bid vector $\textbf{b}$ as input, and outputs a \emph{block} $B \subseteq \bid$ of at most $k$ bids to include. 
    \item {\bf Confirmation rule ${\bf C}(\cdot)$}: takes as input a block $B$ of the included bids and chooses a subset of included bids to confirm. Specifically, ${\bf C}(B)$ outputs a vector $(x_1,...,x_{|B|})\in\{0,1\}^{|B|}$, indicating whether each bid is confirmed. 
    \item {\bf Payment rule ${\bf P}(\cdot)$}: takes a block $B$ as input and outputs a vector of $(p_1,...,p_{|B|})\in\R_{\geq 0}^{|B|}$, indicating the price paid by each transaction in $B$. 
    \item {\bf Revenue rule ${\bf R}(\cdot)$}: takes a block $B$ as input, and outputs the miner's revenue $\mu\in \R_{\geq 0}$. 
\end{itemize}
A feasible TFM $\tfm$ should satisfy the following properties: 
\label{def: feasibility}
\begin{itemize}[leftmargin = 5mm]
    \item  \emph{Space feasibility:} the size of the block $|B|\leq k$.
    \item \emph{Individual rationality:} a user's payment shall not exceed the bid amount, i.e., for any $b_i\in B$, its payment $p_i\leq b_i$.
    Unconfirmed bids pay zero.
    \item \emph{budget feasibility}: the miner's revenue cannot exceed the total payment collected from all confirmed bids, i.e., $\mu\leq \sum_{i=1}^{|B|} p_i$.    
    When the miner's revenue is strictly less than the total payments, we say that the difference is \emph{burnt}. 
\end{itemize}
Among the four rules, the confirmation rule, payment rule, and revenue rule are all executed by the blockchain based on the block, and we treat these three rules as always correctly implemented based on the input block $B$.
The inclusion rule is either implemented by the miner in the plain model or implemented by an MPC in the MPC-assisted model. 

\paragraph{Weak Symmetry.} 
Given a bid $\bid= (b_1,...,b_n)$,
let $x_i$ and $p_i$ denote the random variable representing the probability of bid $b_i$ getting confirmed and the payment it needs to pay in an honest execution. 
We say that a TFM satisfies \emph{weak symmetry} iff the joint distribution of $(\{(b_i,x_i,p_i)\}_{i\in[n]},\mu)$ is the same for input bid vector $\pi(\bid)$ for any permutation $\pi$ on $\bid$.

An operational view of weak symmetry assumes the following: Given a bid vector $\bid$ where each bid may carry some metadata such as identity, public keys of the user, or timestamp. The honest mechanism first sorts the bids based on the amount, and can perform arbitrary tie-breaking rule based on the metadata if multiple bids have the same amount. 
After the sorting step, all the four rules $\tfm$ depend only on bid amounts and positions in the sorted vector.
Weak symmetry is a natural requirement in practice and does not require two bids of the same amount to always receive the same treatment. 

\subsection{Game Induced by TFM }
\label{sec:model-tfm-game}
Henceforth, let $\C$ denotes the strategic players.
Specifically, $\C$ can be a strategic user, strategic miner(s) in control of the current block, or a coalition of the miner(s) and one or more users.

\subsubsection{TFM Game in the Plain Model.}
In the plain model, the TFM game is defined as follows:
\begin{enumerate}[leftmargin = 5mm]
    \item Each honest users not in $\C$ submits their bid represented by a single real-value bid. Let $\bid_{-\C}$ denote the bids from these honest users. 
    \item The coalition $\C$ decides their bids $\bid_{\C}$ according to $\bid_{-\C}$.
    \item The miner selects up to $k$ bids from $(\bid_{\C}, \bid_{-\C})$ to form a block $B$.
    \item The blockchain protocol executes the confirmation rule $\mathbf{C}$, payment rule $\mathbf{P}$, and miner revenue rule ${\bf R}$ to the block $B$ created by the miner.
\end{enumerate}
\paragraph{Strategy Space.} 
In this paper, we focus on direct-revelation mechanism, i.e., for an user $i$ with true value $v_i$, the honest strategy $H_i(v_i) = v_i$ is to submit a single bid representing its true value. 
A strategic user may choose to submit a bid vector $\bid^*_i$ that contains zero to multiple bids which do not necessarily reflect their true value. 
We call all the additional bids that the user injects as \emph{fake} bids. 

A miner's honest behavior $H_{\M}$ is to implement the prescribed inclusion rule on the input bid vector $\bid$ without submitting any fake bids, i.e., if users' honest bid vector is $\bid$, honest miner strategy $H_{\M}(\bid)$ outputs a block ${\bf I}(\bid)$.
In the plain model, a strategic miner may choose to deviate from the prescribed inclusion rule arbitrarily by dropping bids, injecting fake bids, and arbitrarily choosing the randomness used in the inclusion rule, i.e., the strategy $S_{\M}(\bid)$ outputs an arbitrary block $B^*$ of size at most $k$.
Players can adopt mixed strategies.
A strategic coalition can adopt a combination of the strategies of its members.

\subsubsection{MPC-Assisted Model}
In the MPC-assisted model, $M$ miners jointly run a multi-party computation (MPC) to realize an ideal functionality that honestly implements the inclusion rule, and the $M$ miners share the total revenue. 
\revision{We assume that the strategic coalition contains $m<M/2$ miners and that the MPC provides guaranteed output delivery.}
Concretely, there is an ideal functionality $\mathcal{F}_{\rm TFM}$ that, on input $\bid$, implements the honest inclusion, confirmation, payment and revenue rule on $\bid$.
In this paper, we analyze incentives in the ``ideal'' world with $\mcal{F}_{\rm TFM}$. 
The real-world MPC-assisted model can be instantiated under standard cryptographic assumptions.
In the MPC-assisted model, the game runs as follows:
\begin{enumerate}[leftmargin = 5mm]
    \item Users not in $\C$ submit bids represented by a single real value to $\mcal{F}_{\rm TFM}$. Let $\bid_{-\mcal{C}}$ denote the bids from users outside $\C$.
    \item $\mcal{F}_{\rm TFM}$ sends $\bid_{-\mcal{C}}$ to $\C$ and receives a bid vector $\bid_{\C}$ from the coalition.
    \item $\mcal{F}_{\rm TFM}$ implements the inclusion, confirmation, payment, and revenue rule on $(\bid_{-\mcal{C}},\bid_{\mcal{C}})$. The outputs include a vector indicating whether each bid is confirmed or not, a payment vector of every bid's payment, and the total miner revenue. 
    \item \revision{Send the outcome to every player.}
\end{enumerate}
\paragraph{Strategy Space.} 
User's strategy space is the same as in the plain model.
An honest miner does not submit any bids, whereas strategic miner(s) may inject one or more fake bids. 
Unlike in the plain model, now strategic miner(s) can no longer drop honest users' bids or arbitrarily choose which subset of bids to include since now the inclusion rule is implemented by the ideal functionality $\mcal{F}_{\rm TFM}$.
A strategic coalition can adopt a combination of the strategies of its members.

\subsection{Utility and Incentive Compatibility}

\paragraph{Utility.} Each user $i$ has a true value $v_i\in\R_{\geq 0}$ if its primary bid representing its transaction gets confirmed. All the fake bids have true value $0$. Let $p_i$ denote the total payment user $i$ needs to pay. Then user $i$'s utility is $v_i-p_i$ if its primary bid gets confirmed and $-p_i$ otherwise. 
A miner's utility is its revenue $\mu - p_{\M}$, where $p_{\M}$ denotes the total payment from the miner if they inject any fake bids. 
A coalition's joint utility is the sum of all coalition members' utilities.

Below, for a strategic player or coalition, denoted as $\C$, we use $H_{\mcal{C}}$ to denote $\C$'s honest strategy of $\mcal{C}$. 
Let $\val$ represent the vector of true values of all users. 
We use $\util_i(\val;S_{\C})$ to denote the expected utility of user $i$ in the following randomized experiment: 
\begin{itemize}[leftmargin = *]
    \item In the mechanism, players in $\C$ adopts (possibly randomized) strategy $S$,
    while all other players act honestly, where all users' true values are represented by $\val$.
    \item Output utility of player $i$.
\end{itemize}
A miner $j$'s expected utility $\mutil_{j}(\val;S_{\C})$ and coalition's joint utility $\util_{\C}(\val;S_{\C})$ are defined analogously w.r.t the above randomized experiment.
\begin{definition}[Incentive compatibility]
\label{defn:ic}
    Given a TFM $\tfm$, we say that the TFM satisfies incentive compatibility (IC) w.r.t. a strategic player or coalition, denoted as $\C$, iff for any true value vector $\val$ of users, for any strategy $S_{\C}$ of coalition $\C$, we have 
    \[\util_{\C}(\val;S_{\C})\leq \util_{\C}(\val;H_{\C}).\]
    Specifically, we say that a TFM satisfies 
    \begin{itemize}[leftmargin = 5mm]
        \item \emph{User incentive compatibility (UIC)} if the above holds when $\C$ contains an individual user, and the miner acts honestly.
        \item \emph{Miner incentive compatibility (MIC)} in the plain model if the above holds when $\C$ only contains the miner. 
        
        In the MPC-assisted model, we say that a TFM satisfies MIC if the above holds when $\C$ contains \revision{$m<M/2$ miners jointly running} the MPC implementing $\mcal{F}_{\rm TFM}$.
        \item $d$-\emph{side-contract-proofness ($d$-SCP)} in the plain model for some integer $d\geq 1$ if the above holds when $\C$ contains the miner and at least one but no more than $d$ number of users.

        In the MPC-assisted model, we say that a TFM satisfies $d$-SCP for some integer $d\geq 1$ if the above holds when $\C$ contains \revision{$m<M/2$ miners} running the MPC and at least one but no more than $d$ number of users.
    \end{itemize}
\end{definition}
In this work, we introduce a new notion which requires that a rational strategic player cannot harm other honest players without harming themselves. 
We formalize this as \emph{Rational-Harm Proofness} (RHP). 

\begin{definition}[Restatement of \cref{defn:intro-rhp}]
\label{defn:rhp}
Fix a coalition of strategic players $\C$ and a protected set of honest players $\mcal H$ such that $\C\cap \mcal H=\emptyset$. We say that a TFM satisfies rational-harm proofness (RHP) against $\C$ for protecting $\mcal H$ if, for any true value vector $\val$ of users and any strategy $S_{\C}$ for coalition $\C$ such that
\[
    \util_{\C}(\val;S_{\C})\geq \util_{\C}(\val;H_{\C}),
\]
it holds that
\[
    \util_i(\val;S_{\C})\geq \util_i(\val;H_{\C})
\]
for every protected honest player $i\in\mcal H$. Here $\util_a$ and $\util_{\C}$ denote expected utilities in the randomized experiment where the users' true value vector is $\val$, players in $\C$ adopt $S_{\C}$, players outside $\C$ behave honestly, and $H_{\C}$ denotes the honest strategy of coalition $\C$.
\end{definition}

Based on the three types of strategic players, we say that a TFM satisfies
    \begin{itemize}[leftmargin = 5mm]
        \item \emph{User rational-harm proofness (URHP)} iff \cref{defn:intro-rhp} holds for any $\C=\{i\}$ containing an individual user $i$ and $\mcal H$ contains all users except $i$. URHP protects honest users from an individual deviating user; it does not protect the miner.

    \item \emph{Miner rational-harm proofness (MRHP)} iff \cref{defn:intro-rhp} holds for the single miner in the plain model, with $\mcal H$ equal to all users, and for any nonempty coalition of \revision{$m<M/2$ miners} in the MPC-assisted model, with $\mcal H$ equal to all users and all miners outside the coalition.

    \item $d$-\emph{coalition rational-harm proofness ($d$-CRHP)} iff \cref{defn:intro-rhp} holds for coalitions containing the miner and at least one but at most $d$ users in the plain model, protecting all users outside the coalition, and for coalitions containing a nonempty coalition of \revision{$m<M/2$ miners} and at least one but at most $d$ users in the MPC-assisted model, protecting all users outside the coalition and all miners outside the coalition. User-only deviations are covered by URHP.
    \end{itemize}
We also define a weaker notion called \emph{weak RHP}: any strategy that \emph{strictly} benefits the coalition must not harm any players outside the coalition.
\begin{definition}[Weak RHP]
\label{def:weak-RHP}
A TFM is \emph{weakly RHP} w.r.t. $\mathcal{C}$ if \cref{defn:rhp} holds with the strict inequality
$u_{\C}(\val;S_{\mathcal{C}}) > u_{\C}(\val;H_{\C})$
in place of $\ge$. 
Equivalently, every strictly profitable deviation for $\mathcal C$ must weakly preserve the utility of every protected honest player.
\end{definition}
\begin{fact}
    If a TFM is incentive compatible w.r.t. some $\C$, then it is weak RHP w.r.t. $\C$.
\end{fact}
\begin{proof}
    Since incentive compatibility rule out the existence of any strictly profitable strategies of $\C$ compared to honest behavior, it automatically achieves weak RHP w.r.t. $\C$.
\end{proof}
\subsection{Myerson's Lemma}
Our impossibility will rely on the famous Myerson's Lemma.
Below, we use $x_i(\bid)$ and $p_i(\bid)$ to denote user $i$'s probability of getting confirmed and its expected payment under input bid vector $\bid$ when everyone behaves honestly. 
We define $\util_i(\bid):= b_i \cdot x_i(\bid)-p_i(\bid)$ as the honest expected utility when $\bid$ is also the true value vector. 
\begin{lemma}[Myerson's Lemma \cite{myerson}]
\label{Myerson's}
If a TFM satisfies UIC, then
\begin{itemize}[leftmargin = 5mm]
    \item {\bf Monotone allocation}: For any user \( i \), any other users' bids \( \mathbf{b}_{-i} \), any \( b'_i > b_i \), it must be \( x_i(\mathbf{b}_{-i}, b'_i) \geq x_i(\mathbf{b}_{-i}, b_i) \).
    
    \item {\bf Unique payment}: For any user \( i \), any other users' bids \( \mathbf{b}_{-i} \), and bid \( b_i \) from user \( i \), user \( i \)'s expected payment can be uniquely determined as
    \[
    p_i(\mathbf{b}_{-i}, b_i) = b_i \cdot x_i(\mathbf{b}_{-i}, b_i) - \int_0^{b_i} x_i(\mathbf{b}_{-i}, t) dt,
    \]
    with respect to the normalization condition: \( p_i(\mathbf{b}_{-i}, 0) = 0 \), i.e., user \( i \)'s payment must be zero when \( b_i = 0 \).
\end{itemize}
When the mechanism is deterministic, the confirmation probability \( x_i \) is either 0 or 1. In this case, user \( i \)'s payment can be simplified as
\[
p_i(\mathbf{b}_{-i}, b_i) =
\begin{cases}
    \inf\{z \in [0, b_i] : x_i(\mathbf{b}_{-i}, z) = 1\}, & \text{if } x_i(\mathbf{b}_{-i}, b_i) = 1; \\
    0, & \text{if } x_i(\mathbf{b}_{-i}, b_i) = 0.
\end{cases}
\]
\end{lemma}

\section{Characterization in the Plain Model}
\label{sec:plain}
\subsection{Feasibility Results}
\label{sec:plain-ub}
In this section, we present the mechanisms referenced in \cref{thm:intro-ub} and show how each mechanism achieves any three of the four target properties: UIC, MIC, CRHP, and positive miner revenue.
For completeness, we restate each mechanism with respect to the four rules for defining a TFM.
\paragraph{UIC + MIC + CRHP}
\label{sec:aon-zero}
\begin{mdframed}
    \begin{center}
        \textbf{Singleton posted-price} \hfill //Reserve price ${\sf res} \ge 0$.
    \end{center}
    \begin{itemize}[leftmargin = 5mm, label=\textbullet]
        \item \textbf{Inclusion \& Confirmation}: If there is only one bid, and it is at least $\res$, include and confirm that bid. 
        Otherwise, include and confirm no bids. 
        \item \textbf{Payment \& Revenue:} The confirmed bid pays the reserve $\res$, and the miner gets nothing.
    \end{itemize}
\end{mdframed}
\begin{lemma} \label{plain-ub}
The above singleton posted-price satisfies UIC, MIC, URHP, and $d$-CRHP for any $d\geq 1$ in the plain model.
\end{lemma}
\begin{proof} 
For UIC, fix a user. If there are no other bids, bidding at least reserve gives utility $v_i-\res$ and bidding below reserve gives zero; this is exactly the posted-price threshold rule. If there is at least one other bid, the user cannot get its own transaction confirmed by adding bids because the honest inclusion rule confirms a bid only when the submitted vector contains exactly one bid, which is eligible. Fake bids only create additional bids and cannot improve the user's utility. Thus truthful bidding with no fake bids is optimal.

\smallskip\noindent\textbf{MIC:} Miner revenue is always zero. Fake bids can only create payment costs for the miner and cannot increase revenue, so honest behavior is weakly optimal.

\smallskip\noindent\textbf{URHP:} A deviating user can harm another user only if that other user has positive honest utility and becomes unconfirmed or pays more. Payments are fixed at reserve. An honest user can have positive utility in the singleton mechanism only when it is the only submitted bid and its value is above reserve. If another strategic user is present, this cannot be the honest baseline. Hence no rational user deviation can harm a protected user.

\smallskip\noindent\textbf{$d$-CRHP:} A protected honest user can have positive honest utility only when it is the unique submitted bid. But a $d$-CRHP coalition contains at least one user in addition to the miner, so under honest play there is more than one submitted bid and no protected user is confirmed. More generally, any deviation that creates a singleton confirmed bid can only benefit coalition users and cannot make an outside user worse off relative to the honest baseline.
\end{proof}

\smallskip\noindent{\bf Not $1$-SCP:} If there are multiple users and a colluding user has value above reserve, the miner can include only the colluding user's bid, giving the colluding user positive utility while miner revenue remains zero.

\smallskip\noindent{\bf Not MRHP:}  If there is a single honest user with value above reserve, the miner can censor the bid, keep revenue zero, and reduce that user's utility to zero.
\paragraph{UIC + CRPH + Positive Revenue}
\label{sec:aon-posted}
\begin{mdframed}
    \begin{center}
        \textbf{All-or-nothing posted-price} \hfill //Reserve price ${\sf res} \ge 0$.
        \\ \hfill //Revenue parameter $0 <\epsilon \le \res$.
    \end{center}
    \begin{itemize}[leftmargin = 5mm, label=\textbullet]
        \item \textbf{Inclusion \& Confirmation}: If there are no more than $k$ bids and all bids are strictly greater than $\res$, include and confirm all bids. Otherwise, include and confirm no bids. The strict inequality is intentional and avoids zero-utility reserve-boundary deviations in the URHP proof. 
        \item \textbf{Payment \& Revenue:} Each confirmed bid pays the reserve price $\res$, and the miner receives a revenue share of $\epsilon$ per confirmed bid.
    \end{itemize}
\end{mdframed}
\begin{lemma} \label{plain-AorN}
The above all-or-nothing posted-price satisfies UIC, $d$-CRHP for any $d\ge 1$ and positive revenue in the plain model.
Additionally, it satisfies URHP and MRHP.
\end{lemma}
\begin{proof}
UIC follows from the same posted-price threshold argument as in \cref{plain-ub}: a bid strictly above reserve is needed for confirmation, confirmed bids pay the reserve, and fake bids either preserve the all-or-nothing success condition while adding nonpositive fake-bid surplus or make the condition fail. Positive revenue is immediate whenever the mechanism confirms at least one bid.

\smallskip\noindent\textbf{URHP:} Because strict eligibility implies any honestly confirmed strategic user has strictly positive utility, any user-only deviation that makes the all-or-nothing condition fail and harms another confirmed user also makes the deviating user lose strictly positive utility. Fake bids that remain confirmed create non-positive fake-bid surplus and cannot help.

\smallskip\noindent\textbf{MRHP:} In the honest baseline, if any protected honest user has positive utility, then all submitted bids are above reserve, there are at most $k$ bids, all users are confirmed, and miner revenue is $\epsilon$ times the number of confirmed bids. Any miner deviation that makes a protected honest user unconfirmed must either make the all-or-nothing test fail or omit that user. In either case the total number of confirmed real honest bids falls, so the miner's revenue falls by at least $\epsilon$ for each omitted protected confirmed user. Fake bids cannot compensate because each confirmed fake bid costs the miner $\res$ and increases miner revenue by at most $\epsilon\leq\res$.

\smallskip\noindent\textbf{$d$-CRHP:} The same revenue-loss argument applies to miner-user coalitions. Colluding users who were honestly confirmed have strictly positive utility because their values are strictly above reserve. Confirmed fake bids have nonpositive net contribution because they cost $\res$ and increase miner revenue by at most $\epsilon\leq\res$. Therefore a deviation that harms a protected honest user strictly lowers the coalition's joint utility.
These arguments also cover mixed deviations: when any protected user has positive honest utility, every pure deviation weakly lowers the strategic party's utility, and every pure deviation that harms a protected user lowers it strictly. Thus a mixture that preserves the strategic party's expected utility cannot harm a protected user.
\end{proof}

\smallskip\noindent\textbf{Not MIC or $1$-SCP:} Let $S$ be the honest bid vector. If $|S| > k$ and every bid is above $\res$, the miner gets 0 revenue. By ignoring $|S| - k$ bids from $S$, the miner can include exactly $k$ bids and earn $k\eps > 0$ revenue. The above mechanism does not achieve $1$-SCP for the same reason as the miner can collude with any of the included players for the same effect.
\paragraph{UIC + MIC + Positive Revenue}
\label{sec:ppw.random}
\begin{mdframed}
    \begin{center}
        \textbf{Posted-price with random selection} \hfill //Reserve price ${\sf res} > 0$. 
        \\ \hfill //Revenue parameter $0 < \epsilon \le \res$.
    \end{center}
    \begin{itemize}[leftmargin = 5mm, label=\textbullet]
        \item \textbf{Inclusion \& Confirmation}: Let $S$ be the set of bids at least the reserve price ${\sf res}$. If $|S|\le k$, include and confirm all bids in $S$. Otherwise, uniformly at random choose $k$ bids in $S$ to include and confirm.
        \item \textbf{Payment \& Revenue}: Each confirmed bid pays the reserve price $\res$. The miner receives total revenue $\epsilon$ if at least one bid is confirmed, and zero otherwise.
    \end{itemize}
\end{mdframed}
\begin{lemma} \label{plain-ub-pp}
The above posted-price with random selection TFM satisfies UIC, MIC, and positive revenue in the plain model.
Additionally, it achieves URHP.
\end{lemma}
\begin{proof}
\smallskip\noindent{\bf UIC:} The allocation and payment rule satisfies Myerson's lemma: for a fixed user and fixed other bids, the allocation probability is zero below reserve and equals the uniform-selection probability at or above reserve, while the expected payment is the reserve times that probability. To also rule out fake bids, let $h$ be the number of eligible outside bids. A user of value $v$ submitting an eligible primary bid and $f$ eligible fake bids has utility $q(v-\res-f\res)$, where $q=\min\{1,k/(h+1+f)\}$. For $v\geq\res$, this is at most the truthful utility $q_0(v-\res)$, where $q_0=\min\{1,k/(h+1)\}$: the bound is immediate if the former utility is negative, and otherwise follows from $q\leq q_0$. For $v<\res$ it is nonpositive. With no eligible primary bid the utility is also nonpositive, and ineligible bids have no effect. These bounds establish UIC, including for mixed deviations.

\smallskip\noindent{\bf MIC:} If there is at least one eligible real bid, honest miner revenue is already $\epsilon$, the maximum possible revenue. Censoring some eligible bids while leaving at least one eligible bid can keep revenue at $\epsilon$ but cannot increase it. Censoring all eligible bids gives zero revenue unless the miner injects a fake eligible bid; an injected fake bid costs at least $\res$ and yields at most $\epsilon$, so it is not profitable since $\epsilon\leq \res$. If there are no eligible real bids, injecting fake eligible bids yields revenue at most $\epsilon$ but costs at least $\res$ per confirmed fake bid, so it cannot strictly improve miner utility.

\smallskip\noindent{\bf URHP:} A user-only deviation can harm another user only by reducing that user's confirmation probability. Under random selection this requires adding eligible competition, either by overbidding with an ineligible real bid or injecting eligible fake bids. Turning a real bid with value below reserve into an eligible bid gives non-positive expected surplus and is strictly negative if it is ever confirmed. Injected fake eligible bids have true value zero and impose expected payment costs. A user with value exactly reserve is already eligible under truthful bidding and has zero surplus, so changing the real bid within the eligible range does not add competition; fake bids still cost money. Therefore any user-harming deviation lowers the deviating user's utility.
Since UIC bounds every pure deviation by truthful utility, a mixed deviation cannot offset such a strict loss with a gain from another pure deviation. Hence the argument also covers mixed deviations.
\end{proof}

\smallskip\noindent{\bf Not $1$-SCP:} The miner can choose the random seed or inclusion subset to favor a colluding eligible user, increasing the coalition's utility while keeping miner revenue unchanged.

\smallskip\noindent{\bf Not MRHP:} The plain miner can manipulate inclusion or randomness, or censor eligible honest users, while keeping fixed revenue $\epsilon$ as long as at least one eligible bid remains.

\smallskip\noindent{\bf Not $1$-CRHP:} The same deviation, with a colluding eligible user kept confirmed or favored by the miner, harms an outside eligible user while preserving the coalition's utility.
\paragraph{MIC + CRHP + Positive Revenue}
\label{sec:fpa}
\begin{mdframed}
\label{FPA}
    \begin{center}
        {\bf First-price auction}
    \end{center}
    \begin{itemize}[label=\textbullet, leftmargin = 5mm]
        \item \textbf{Inclusion \& Confirmation:} Include and confirm the top $k$ bids, breaking ties arbitrarily. 
        \item \textbf{Payment \& Revenue:} All confirmed bids pay their bid price and all payments go to the miner.
    \end{itemize}
\end{mdframed}
\begin{lemma}
The above first-price auction satisfies URHP, MIC, MRHP, $d$-CRHP for any $d\ge 1$, $d$-SCP for any $d\ge1$, and has positive miner revenue in the plain model.
\end{lemma}
\begin{proof}
    \textbf{MIC:} The miner receives all payments, so they would have to increase payments from users to increase their utility. In the honest inclusion rule, the top bids are confirmed and paid, and there is no way to have payments higher than the bids.
    
    \smallskip\noindent\textbf{URHP, $d$-CRHP \& MRHP:} Under truthful bidding in a first-price auction, every confirmed user pays its value and every unconfirmed user gets zero, so every user's honest utility is zero. No deviation can make an honest user's utility strictly below zero because individual rationality and non-confirmation both give utility at least zero. Hence URHP, MRHP, and $d$-CRHP hold.
    
    \smallskip\noindent\textbf{$d$-SCP:} In the honest case, because all payments go to the miner, the joint utility of the miner and any set of users is equal to the summation of the top $\min\{k,n\}$ valuations. Under any deviation, a confirmed colluding user's payment cancels within the coalition, leaving its true value; a confirmed outside user contributes its truthful bid, also its true value; and a confirmed fake bid contributes zero after its payment cancels. Thus coalition utility is the sum of the true values of at most $k$ confirmed real users, outcome by outcome and hence also in expectation under mixed deviations.
    It is impossible to gain more utility than this for any coalition due to the block size limit, budget feasibility, and individual rationality (See \cref{def: feasibility}).
\end{proof}
\smallskip\noindent{\bf Not UIC:} Since users pay their own bid, a confirmed user may profit by underbidding while still remaining confirmed. This deviation harms miner revenue, not a protected user under URHP.

\subsection{Impossibilities in the Plain Model}
\label{sec:impossibility}
We first introduce the following useful lemmas in our impossibility proofs.
Recall that, for a bid vector $\bid = (b_1,\dots, b_n)$, for $i \in[n]$, $x_i(\bid)$, $p_i(\bid)$ and $\util_i(\bid)$ denote the confirmation probability, expected payment, and expected utility for bid $b_i$ in $\bid$ when the honest inclusion rule is followed. 
For unique bid values we use the value-indexed convention from \cref{sec:model}: $x_z(\bid)$, $p_z(\bid)$, and $\util_z(\bid)$ denote the allocation probability, expected payment, and expected utility of the unique bid with value $z$. We keep the decorated notation $\widetilde{x}_z$, $\widetilde{p}_z$, and $\widetilde{\util}_z$ only as typographic reminders of this value-indexed convention.
\begin{lemma}\label{cor:increasing}
    For any TFM that satisfies UIC, 
    given any $\bid_{-i}$, let $b^* = \inf  \{b:x_i(\bid_{-i}, b) > 0\} < +\infty$. Then for any $v^* > v \ge b^*$, we have $\util_i(\bid_{-i},v^*) > \util_i(\bid_{-i},v)$.
\end{lemma}
\begin{proof}
    Suppose that for some $\mathbf{b}_{-i}$ and $v$, we have $x_i(\mathbf{b}_{-i}, v) > 0$. Then for any $v^*>v$, 
    \begin{align*}
            \util_i(\mathbf{b}_{-i}, v^*) &= v^* \cdot x_i(\mathbf{b}_{-i} , v^*) - p_i(\bid_{-i},v^*) \\
            & =\int^{v^*}_0 x_i(\bid_{-i},t) dt \tag*{by Myerson's Lemma (\cref{Myerson's}).}\\
            & > \int^{v}_0 x_i(\bid_{-i},t) dt = \util_i(\mathbf{b}_{-i} , v) 
    \end{align*}
    The only remaining case is $v=b^*$ with $x_i(\bid_{-i},v)=0$. Choose $w\in(v,v^*)$. By the definition of $b^*$ and monotonicity, $x_i(\bid_{-i},w)>0$. Myerson's payment identity gives
    \[
    \util_i(\bid_{-i},v^*)-\util_i(\bid_{-i},v)
    =\int_v^{v^*}x_i(\bid_{-i},t)\,dt
    \ge (v^*-w)x_i(\bid_{-i},w)>0.
    \]
\end{proof}
\begin{lemma}\label{cor:p/x-increasing}
    For any TFM that satisfies UIC, 
    given any $\bid_{-i}$, let $b^* = \inf\{b:x_i(\bid_{-i}, b) > 0\} < +\infty$. Then for any $v^*\geq v > b^*$, the ratio $\frac{p_i(\mathbf{b}_{-i}, v^*)}{x_i(\mathbf{b}_{-i}, v^*)} \geq \frac{p_i(\mathbf{b}_{-i}, v)}{x_i(\mathbf{b}_{-i}, v)}$.
\end{lemma}
\begin{proof}
For this proof, abbreviate $x(t):=x_i(\bid_{-i},t)$ and $p(t):=p_i(\bid_{-i},t)$.
By Myerson's Lemma,
\[
\frac{p(v)}{x(v)}
=
v-\frac{\int_0^v x(t)\,dt}{x(v)}.
\]
Therefore,
\begin{align*}
\frac{p(v^*)}{x(v^*)}-\frac{p(v)}{x(v)}
&=
\left(\frac{1}{x(v)}-\frac{1}{x(v^*)}\right)\int_0^v x(t)\,dt
+
(v^*-v)
-
\frac{1}{x(v^*)}\int_v^{v^*}x(t)\,dt.
\end{align*}
Both terms are nonnegative.
The first is nonnegative because $x(v^*)\ge x(v)>0$ by monotonicity.
For the second, monotonicity gives $x(t)\le x(v^*)$ for every $t\in[v,v^*]$, and hence
\[
\int_v^{v^*}x(t)\,dt\le (v^*-v)x(v^*).
\]
Thus,
\[
\frac{p(v^*)}{x(v^*)}\ge \frac{p(v)}{x(v)}.
\]
\end{proof}
\begin{lemma}\label{lem:subvec-mr}
Suppose a TFM $({\bf I},{\bf C},{\bf P},{\bf R})$ satisfies MIC in the plain model. For any $\bid \subseteq \mathbf{b'},$ it must be that $\mu({\bf b'}) \ge \mu(\bid)$.
\end{lemma}
\begin{proof}
     For the sake of contradiction, suppose this was not the case, so $\mu({\bf b'}) <\mu({\bf b})$. 
     Then, when the miner sees the honest bid vector  ${\bf b'}$, they censor the bids not in $\bid$ perform the honest inclusion rule as if the input bid vector is ${\bf b}$ to gain more revenue. This contradicts MIC.
\end{proof}
\begin{lemma}\label{lem:miner-revenue-bound}
    Suppose a TFM $({\bf I},{\bf C},{\bf P},{\bf R})$ satisfies MIC in the plain model. Then for any bid vectors $\bid, \bid', \val$, we have $\mu(\bid,\mathbf{v}) \ge \mu(\bid',\mathbf{v}) - \sum^{|\bid'|}_{i = 1} p_i(\bid',\mathbf{v})$. 
\end{lemma}
\begin{proof}
     For the sake of contradiction, suppose  $\mu(\bid,\mathbf{v}) < \mu(\bid',\mathbf{v}) - \sum^{|\bid'|}_{i = 1} p_i(\bid',\mathbf{v})$. 
    Consider the world where $(\bid,\mathbf{v})$ is the honest bid vector. The miner's honest utility is $\mu(\bid,{\bf v})$.
    Consider the following miner strategy: it deviates from the honest inclusion rule by ignoring the bids in $\bid$ and submitting the fake bids in $\bid'$ to achieve the bid vector $(\bid',\val).$ Inserting the fake bids costs the miner $\sum^{|\bid'|}_{i = 1} p_i(\bid',\mathbf{v})$ so the miner's utility is 
    \begin{align*}
        \mu(\bid',\mathbf{v}) - \sum^{|\bid'|}_{i = 1} p_i(\bid',\mathbf{v}) &> \mu(\bid,\mathbf{v}), \tag*{by assumption}
    \end{align*}
    which contradicts MIC.
\end{proof}
Recall that, for a bid vector $\bid = (b_1,\dots, b_n)$, for $i \in[n]$, $x_i(\bid)$, $p_i(\bid)$ and $\util_i(\bid)$ denote the confirmation probability, expected payment, and expected utility for bid $b_i$ in $\bid$ when the honest inclusion rule is followed. 
Also, $\widetilde{x}_{b_i}(\bid)$, $\widetilde{p}_{b_i}(\bid)$, and $\widetilde{\util}_{b_i}(\bid)$ are defined analogously for a unique bid $b_i$ in bid vector $\bid$. 

In a (possibly randomized) mechanism, the inclusion rule, confirmation rule, payment rule and miner revenue rule can be randomized. 
However, the randomness used in the inclusion rule is freely chosen by the miner, whereas the randomness used in other rules come from the blockchain.
Let $\Omega$ and $\mathcal{D}$ be the sample space and distribution of the randomness specified by the honest inclusion rule. 
For each $r \in \Omega $, define $x_i(\bid \mid r) $, $ p_i(\bid \mid r)$, and $\mu(\bid \mid r) $ as the expected confirmation probability of bid \( b_i \), the expected payment of bid \( b_i \), and the expected miner revenue, respectively, conditioned on the randomness used in the inclusion rule being \( r \). 
The expectation in each case is now taken over the randomness used in the confirmation, payment, and miner revenue rules, respectively.

By definition,
\[x_i(\bid) = \underset{r\leftarrow \mathcal{D}}{\E}[x_i(\bid\mid r)], \qquad p_i(\bid) = \underset{r\leftarrow \mathcal{D}}{\E}[p_i(\bid\mid r)], \qquad \mu(\bid) = \underset{r\leftarrow \mathcal{D}}{\E}[\mu(\bid\mid r)].\]
We define $\util_i(\bid\mid r)$ analogously as user $i$'s honest expected utility conditioned on randomness $r$ used in the inclusion rule when the input bid vector is $\bid$. 
Value-indexed utilities such as $\widetilde{\util}_v(\bid)$ are used only when $v$ is unique in $\bid$. 
Let
\[
W(\bid) \colon= \left\{ i \in [n] : \util_i(\bid) > 0 \right\}
\]
to be the set of users whose honest expected utility is strictly positive under the bid vector $\bid$.
\begin{lemma} \label{lem:rand-include-above}
Suppose a TFM $({\bf I},{\bf C},{\bf P},{\bf R})$ satisfies weak $1$-CRHP in the plain model. For any bid vector $\bid = (b_1,\dots, b_n)$ where there exists a unique $b_i$ in $\bid$ such that $\util_{i}(\bid) >0$, for all unique $b_j$ in $\bid$ where $b_j > \frac{p_i(\bid)}{x_i(\bid)}$, we have $x_{j}(\bid) > 0$. 
\end{lemma}
In words, if there is a unique bid $b_i$ in $\bid$ with positive utility, all unique bids greater than the ratio $p_i(\bid)/x_i(\bid)$ must have positive confirmation probability as well. 
\begin{proof}
Seeking contradiction, suppose there exists a bid vector $\bid = (b_1,\dots, b_n)$ where $b_i$ is a unique bid in $\bid$ and $\util_{i}(\bid) >0$, but $x_{j}(\bid) = 0$ for some unique $b_j> \frac{p_i(\bid)}{x_i(\bid)}$ in $\bid$. 

Consider a world where $\bid$ is the honest bid vector. Consider a coalition  consisting of user $j$ with true value $b_j$ and the miner, whose joint expected utility in the honest case equals $\mu(\bid)$.
The coalition can perform the following strategy: 
\begin{itemize}
    \item The miner ignores the bid $b_i$ from user $i$;
    \item User $j$ changes its bid to $b_i$ and injects a fake bid of value $b_j$ under an arbitrary fake identity.
    \item The miner performs the honest inclusion rule on the remaining bids.
\end{itemize}
Because the bids $b_i$ and $b_j$ are unique, by weak symmetry, the fake bid at $b_j$ will not be confirmed, while user $j$'s real bid at $b_i$ 
will have positive confirmation probability with expected payment $p_i(\bid)$. 
Under the deviation, the coalition's joint utility equals
\begin{align*}
    \mu(\bid)+b_j\cdot x_i(\bid)-p_i(\bid) &> \mu(\bid) + \frac{p_i(\bid)}{x_i(\bid)}\cdot x_i(\bid)-p_i(\bid) = \mu(\bid), 
\end{align*}
where the rightmost expression is the coalition's expected joint utility in the honest case. 
However, user $i$ in the honest case has positive utility, but is now unconfirmed.
Thus, this strategy increases the coalition's joint utility while harming the honest user $i$, contradicting weak 1-CRHP.
\end{proof}

\begin{lemma} \label{lem:same-mr}
Suppose a TFM \( (\mathbf{I}, \mathbf{C}, \mathbf{P}, {\bf R}) \) satisfies MIC in the plain model. Then, for any $\bid$, we have $\mu(\bid\mid r) = \mu(\bid)$ almost surely:
\[\underset{r\leftarrow\mcal{D}}{\Pr}[\mu(\bid\mid r) = \mu(\bid)]=1.\]
\end{lemma} 

\begin{proof}
By MIC, it must be that $\mu(\bid\mid r)\leq \mu(\bid)$ for any $r\in\Omega$. 
Otherwise, the miner can fix $r$ as the randomness used in the inclusion rule instead of sampling the randomness from $\mcal{D}$ and strictly increases its expected revenue.
Since $\mu(\bid) = \underset{r\leftarrow \mathcal{D}}{\E}[\mu(\bid\mid r)]$, we have $\underset{r\leftarrow\mcal{D}}{\Pr}[\mu(\bid\mid r) = \mu(\bid)]=1$.
\end{proof}


\begin{lemma} \label{lem:limited-space}
Suppose a TFM \( (\mathbf{I}, \mathbf{C}, \mathbf{P}, {\bf R}) \) satisfies MIC and weak \( 1 \)-CRHP, and the block size is bounded by \( k \) in the plain model. Then, for every bid vector \( \bid = (b_1, \dots, b_n) \),
\[
|W(\bid)| \leq k,
\]
i.e., at most \( k \) users can have positive expected utility under the honest mechanism.
\end{lemma}

\begin{proof}
Assume, for the sake of contradiction, that there exists a bid vector $\bid$ where \( |W(\bid)| \ge k+1 \). Since the block size is at most \( k \), the inclusion rule cannot include more than \( k \) bids for any randomness \( r \in \Omega \). 
Therefore, there must exists some $i\in W(\bid)$ and a set $\Omega_{i}:=\{r\in\Omega: x_{i}(\bid\mid r) = 0\}$ such that $\Pr_{r\leftarrow \mathcal{D}}[\Omega_{i}]>0$, this implies that
\begin{equation}
    \label{eqn:0-util-prob}
    \Pr_{r\leftarrow \mathcal{D}}[\util_{i}(\bid\mid r) = 0] > 0
\end{equation}
Recall that by definition, $\util_{i}(\bid) = \E_{r \leftarrow \mcal{D}}[\util_{i}(\bid \mid r)] > 0.$
Together with \cref{eqn:0-util-prob}, this means that $\Pr_{r \leftarrow \mcal{D}}[\Omega'_i]>0$, where $\Omega'_i = \{r\in\Omega: \util_{i}(\bid\mid r) > \util_{i}(\bid)\}$.

For every fixed randomness, at most $k$ bids can be included. Thus, for each $r\in \Omega'_i$, there exists some $j\in W(\bid)$ such that $x_j(\bid \mid r) = 0$. 
By the pigeonhole principle, there exists some $j^* \in W(\bid)$ such that 
\[\underset{r\leftarrow \mcal{D}}{\Pr}[\util_{i}(\bid\mid r) > \util_{i}(\bid) \text{ and } x_{j^*} (\bid \mid r) = 0] > 0.\]
Define set $\Omega^* :=\{r\in \Omega: \util_{i}(\bid\mid r) > \util_{i}(\bid) \text{ and } x_{j^*} (\bid \mid r) = 0 \text{ and } \mu(\bid\mid r) = \mu(\bid)\}$. By \cref{lem:same-mr}, we have $\Pr_{r\leftarrow \mcal{D}}[\Omega^*]>0$.

Consider the following strategy by a coalition of the miner and user \( i \) as follows: the miner picks an arbitrary randomness \( r^* \in \Omega^*\) in the inclusion rule. By the choice of $\Omega^*$, the coalition's expected joint utility becomes:
\begin{align*}
\mu(\bid\mid r^*) + \util_{i}(\bid \mid r^*) 
&= \mu(\bid) + \util_{i}(\bid \mid r^*) > \mu(\bid) + \util_{i}(\bid) ,
\end{align*}
which strictly improves the coalition's expected joint utility. 
Meanwhile, user \( j^* \), who, in the honest case, has positive expected utility, is now excluded and receives zero utility, violating weak 1-CRHP. 
\end{proof}

Throughout the proof, we will rely on bids having strictly positive utility. Therefore, in our argument, we often need to remove the infimum value from a set, which might have zero utility. For this purpose, we define an operator $\rminf(\cdot)$. For any nonempty set $S\subseteq \R_{\ge0}$, $\rminf(S) = S\setminus\{\inf S\}$.
Recall that for two sets $S$ and $S'$, we say $S \prec S'$ if $\sup S < \inf S'$.



\begin{lemma} 
\label{lem:W-size-set}
Let \( (\mathbf{I}, \mathbf{C}, \mathbf{P}, {\bf R}) \) be a TFM that satisfies UIC, MIC, and weak \( 1 \)-CRHP. Suppose there exists a bid vector \( \val = (v_1,\dots,v_n) \) and a sequence of sets $S_1,\ldots,S_m\subseteq\R_{\ge0}$ such that, for every \( \bid=(b_1,\ldots,b_m)\in S_1\times\cdots\times S_m \), each $b_\ell$ is unique in $(\val,\bid)$ and $x_{b_1}(\val,\bid)>0$. If $|S_\ell|\ge2$ for all $\ell$ and $S_1\prec\cdots\prec S_m$, then there exists a bid vector \( \bid^* \) with \( |W(\bid^*)|\ge m \).
\end{lemma}

\begin{proof}
Pick an arbitrary $\bid=(b_1,\ldots,b_m)\in\rminf(S_1)\times\cdots\times\rminf(S_m)$ and set $\bid^*=(\val,\bid)$.

First consider $b_1$. Choose $\widehat b_1\in S_1$ with $\widehat b_1<b_1$. By assumption,
\[
x_{\widehat b_1}(\val,\widehat b_1,b_2,\ldots,b_m)>0.
\]
Hence, by \cref{cor:increasing}, $\util_{b_1}(\bid^*)>0$.

Now fix any $\ell>1$, and choose $\widehat b_\ell\in S_\ell$ with $\widehat b_\ell<b_\ell$. Let $\widehat{\bid}$ be obtained from $\bid^*$ by replacing $b_\ell$ with $\widehat b_\ell$. Choose $\widehat b_1\in S_1$ with $\widehat b_1<b_1$. By assumption, the bid $\widehat b_1$ has positive confirmation probability when the first and $\ell$-th coordinates are $\widehat b_1$ and $\widehat b_\ell$, respectively. Raising $\widehat b_1$ to $b_1$ and applying \cref{cor:increasing} shows that $b_1$ has positive utility in $\widehat{\bid}$.

Since $S_1\prec S_\ell$, we have
\[
\widehat b_\ell>b_1>\frac{p_{b_1}(\widehat{\bid})}{x_{b_1}(\widehat{\bid})}.
\]
Thus, by \cref{lem:rand-include-above}, $\widehat b_\ell$ has positive confirmation probability in $\widehat{\bid}$. Raising $\widehat b_\ell$ to $b_\ell$ and applying \cref{cor:increasing} gives $\util_{b_\ell}(\bid^*)>0$.

Therefore all $b_1,\ldots,b_m$ have positive utility in $\bid^*$, so $|W(\bid^*)|\ge m$.
\end{proof}

Our impossibility relies on the following key technical lemma. 
\begin{lemma}
\label{lem:rand-sCRHP-ind-sets}
Suppose a TFM $(\mathbf{I}, \mathbf{C}, \mathbf{P}, {\bf R})$ with block size $k$ satisfies UIC, MIC, and $1$-CRHP, and let $\val$ be a bid vector such that $|\val| \ge k+1$ and index $i^*$ where $x_{i^*}(\val )>0$.
Then there exists a bid vector ${\bf b}^*$ such that 
\[|W(\bid^*)| \ge k+1.\]
\end{lemma}
Combining \cref{lem:rand-sCRHP-ind-sets} and \cref{lem:limited-space}, we immediately get the following impossibility result:
\begin{theorem} \label{thm:rand-sCRHP-imp-k+1}
    Let $k$ denote the block size.
    If a TFM $({\bf I},{\bf C},{\bf P},\mu)$ satisfies UIC, MIC, and $1$-CRHP,
    then it must be degenerate.
\end{theorem}
\begin{proof}
By \cref{lem:rand-sCRHP-ind-sets} and \cref{lem:limited-space}.
\end{proof}

\subsubsection[Proof of the plain-model induction lemma]{Proof of \cref{lem:rand-sCRHP-ind-sets}}
Suppose, for the sake of contradiction, that there exists $\val$ where $|\val|\ge k+1$ and $x_{i}(\val) >0$ for an index $i\in[n]$. Here $n=|\val|$. We will prove a stronger result using induction to construct $\bid^*$ such that $|W(\bid^*)|\ge k+1$. 

For any $m\in[k+1]$, our induction constructs a sequence of positive real numbers $\eps^{(1)},...,\eps^{(m)}$, a sequence of sets $S^{(1)},\dots, S^{(m)}\subset \R_{\geq 0}$ 
such that properties \ref{srand:prop1}-\ref{srand:prop3} below hold. 
    \begin{enumerate} [label=(${\sf Prop'\text{-}\arabic*}$), leftmargin = 15mm]
        \item \label{srand:prop1} Each $S^{(\ell)}$ contains two values: $S^{(\ell)} = \{\alpha^{(\ell)}, \beta^{(\ell)}\}$ with $\alpha^{(\ell)} < \beta^{(\ell)}$.

        \item \label{srand:prop2} 
            $\max(\val)+\eps^{(m)} < \alpha^{(m)} < \beta^{(m)}<\alpha^{(m-1)}<\beta^{(m-1)}<\dots < \alpha^{(1)} < \beta^{(1)}$.
        
        This implies that $\set{\max(\val)+\eps^{(m)}} \prec S^{(m)} \prec \dots \prec S^{(1)}$
        \item \label{srand:prop3}  For all \( \bid  = (b_1, \dots, b_m) \in S^{(1)}\times \cdots \times S^{(m)} \), $b_m$ is unique in $(\val_{-i}, \bid)$. Moreover,
        \[\widetilde{x}_{b_m}\left(\val_{-i}, \bid\right) > 0 \quad \text{and}  \quad \frac{\widetilde{p}_{b_m}\left(\val_{-i}, \bid\right)}{\widetilde{x}_{b_m}\left(\val_{-i}, \bid\right)} \le \alpha^{(m)} -\eps^{(m)}.
        \]
        As an implication, $\widetilde{\util}_{b_m}(\val_{-i},\bid) > 0$.
\end{enumerate}

        
        
    We first show how the lemma follows assuming that the above properties holds for any $m\in[k+1]$ and give the induction proof afterwards.
    Let $m= k+1$. Consider $\val_{-i}$ and the sequence of sets $S^{(1)},...,S^{(k+1)}$.
    By \ref{srand:prop1}, $|S^{(\ell)}|\ge2$ for all $\ell \in [k+1]$. From \ref{srand:prop2}, we also know that $S^{(k+1)} \prec \dots \prec S^{(1)}$, and that for $\bid = (b_1,\dots,b_m) \in S^{(1)} \times \dots \times S^{(m)}$, each $b_\ell$ is unique in $(\val_{-i}, \bid)$. 
    Furthermore, $\widetilde{x}_{b_m}(\val_{-i}, \bid)> 0$ by \ref{srand:prop3}.
 Thus, the assumptions of \cref{lem:W-size-set} hold for $\val_{-i}$ and the sets  $S^{(k+1)},...,S^{(1)}$, which means that some $\bid^*$ exists where $|W(\bid^*)| \ge k+1$.
 

    \paragraph{Induction Proof.} In the rest of the proof, we focus on the induction for proving properties \ref{srand:prop1} - \ref{srand:prop3}.
    For clarity, in our induction, we use superscripts to refer to the propositions for a particular $m \in [k+1]$, i.e. \ref{srand:prop1}$^{(m)}$ represents property \ref{srand:prop1} for $m$.

    \smallskip\noindent \underline {\it Base Case: $m=1$.} 
    Let $v^* = \max(\val)+1$.
    Consider a bid vector $\val' = (\val_{-i}, v^*)$. Since $x_i(\val) > 0$, by \cref{cor:increasing}, we know that $\util_{i}(\val_{-i}, v^*)>\util_{i}(\val)\ge 0$, and thus $\tau:=\frac{p_{i}(\val')}{x_{i}(\val')} < v^*$. 
    Define 
    \[\eps^{(1)} :=\frac{1}{3}(v^*-\max  (\val,\tau )) > 0 \quad \text{ and } \quad
    S^{(1)} := \set{\alpha^{(1)} = v^*- \eps^{(1)}, \beta^{(1)}= v^*- \frac{1}{2}\eps^{(1)}}. \]
    We now prove the properties:
    \begin{itemize}
        \item \ref{srand:prop1}$^{(1)}$ and \ref{srand:prop2}$^{(1)}$ are true by construction. In particular, $\alpha^{(1)}-\eps^{(1)}=\max(\val,\tau)+\eps^{(1)}>\max(\val)$.
        \item Pick an arbitrary $b_1 \in S^{(1)}$. By Myerson's Lemma (\cref{Myerson's}), $\widetilde{x}_{b_1}(\val_{-i},b_1)\ge x_{i}(\val)>0.$ Furthermore, we have 
        \begin{align*}
            \frac{\widetilde{p}_{b_1}(\val_{-i},b_1)}{\widetilde{x}_{b_1}(\val_{-i},b_1)} &\le \frac{p_{i}(\val')}{x_{i}(\val')} \tag*{by \cref{cor:p/x-increasing}}\\
            &= \tau \le \max  (\val,\tau ) \leq \alpha^{(1)} - \eps^{(1)}.
        \end{align*}
        Thus, \ref{srand:prop3}$^{(1)}$ is satisfied.
    \end{itemize}
 
 \smallskip\noindent \underline{\it Inductive Steps: $1<m \le k+1$.}
        Suppose  $S^{(1)}, \dots, S^{(m-1)}$ and $\eps^{(m-1)}$ are the sequence of sets and the real number respectively that satisfy \ref{srand:prop1}$^{(m-1)}$-\ref{srand:prop3}$^{(m-1)}$. 
    We now construct $S^{(m)}$ and $\eps^{(m)}$. 
    Recall that $\alpha^{(m-1)} = \min S^{(m-1)}$.    
    \begin{claim} \label{claim:confirmed}
        For all $\bid= (b_1,\dots,b_{m-1}) \in S^{(1)}\times\dots\times S^{(m-1)}$ and any $\tau \in (\alpha^{(m-1)} - \eps^{(m-1)}, \alpha^{(m-1)})$, let $\bid_{\tau} = (\val_{-i},\bid,\tau)$.
        Then $\tau$ is a unique bid in  $\bid_{\tau}$ and $\widetilde{x}_{\tau}(\bid_{\tau}) > 0.$
    \end{claim}
    \begin{proof}
    By \ref{srand:prop2}$^{(m-1)}$, $\tau$ and $b_{m-1}$ are both unique bids in $\bid_{\tau}$.
    In the rest of this proof, we focus on proving that $\widetilde{x}_{\tau}(\bid_{\tau}) > 0.$  Suppose for the sake of contradiction that $\widetilde{x}_{\tau}(\bid_{\tau}) = 0.$ 
    We first show that
    \begin{equation}
        \label{claim:rand-strong-previous-bid}
        \widetilde{\util}_{b_{m-1}}(\bid_{\tau}) \ge \widetilde{\util}_{b_{m-1}}(\val_{-i},\bid).
    \end{equation}

Suppose for the sake of contradiction that $\widetilde{\util}_{b_{m-1}}(\bid_{\tau}) < \widetilde{\util}_{b_{m-1}}(\val_{-i},\bid)$. 
Consider a world where \( (\val_{-i},\bid) \) is the honest input bid vector. 
Since $|(\val_{-i},\bid)|=n+m-2\ge n>k$, by \cref{lem:limited-space}, there must exist some user $j$ whose expected utility is $0$ in the honest case. By \ref{srand:prop3}$^{(m-1)}$, this user is distinct from the user bidding $b_{m-1}$.
Now consider a coalition made up of the miner and user $j$ that performs the following strategy: they inject a fake bid at $\tau$, faking a world with input bid vector $\bid_{\tau}$, and performs the honest mechanism on $\bid_{\tau}$.
The fake bid $\tau$ must be unconfirmed by assumption.
In the honest case, their expected joint utility is just the expected miner revenue $\mu(\val_{-i},\bid)$.
The coalition's strategic utility is $\mu(\bid_{\tau}) + \util_j(\bid_{\tau}) \geq \mu(\val_{-i},\bid)$
by \cref{lem:subvec-mr} and individual rationality.
However, the user bidding $b_{m-1}$ is harmed by the assumption that $\widetilde{\util}_{b_{m-1}}(\bid_{\tau}) < \widetilde{\util}_{b_{m-1}}(\val_{-i},\bid)$, which contradicts $1$-CRHP.   

Now we are ready to prove the claim statement.
Consider a world where $\bid_{\tau}$ is the honest bid vector. Let \( u \) denote the user whose true value is \( b_{m-1} \) and let \( u' \) denote the user who's valuation is \( \tau \). 
Then user $u$'s expected utility in the honest case is
\begin{align*}
    \widetilde{\util}_{b_{m-1}}(\bid_{\tau}) &\ge \widetilde{\util}_{b_{m-1}}(\val_{-i},\bid) &\text{by \cref{claim:rand-strong-previous-bid}}
    & > 0 &\text{by \ref{srand:prop3}}^{(m-1)}
\end{align*}
Consider a coalition made of the user \( u'\) and the miner. 
Since $\widetilde{x}_{\tau}(\bid_{\tau})=0$ by the assumption, the coalition's expected joint utility in the honest case is simply the expected miner revenue $\mu(\bid_{\tau})$. 
Consider the following strategy of the coalition:
\begin{itemize}
    \item The miner ignores the bid $b_{m-1}$ from user $u$;
    \item User $u'$ raises its bid to $b_{m-1}$. The miner then performs the honest inclusion rule on $(\val_{-i},\bid)$ where the bid at $b_{m-1}$ now comes from the coalition.
\end{itemize}
Under this strategy, the coalition's expected joint utility becomes
\begin{align*}
&\mu\left(\val_{-i}, \bid\right) + \left( \tau \cdot \widetilde{x}_{b_{m-1}}\left(\val_{-i}, \bid\right) - \widetilde{p}_{b_{m-1}}\left(\val_{-i}, \bid\right) \right) \\
>\; & \mu\left(\val_{-i}, \bid\right) +  \left( \alpha^{(m-1)} - \eps^{(m-1)} \right) \cdot \widetilde{x}_{b_{m-1}}\left(\val_{-i}, \bid\right) - \widetilde{p}_{b_{m-1}}\left(\val_{-i}, \bid\right) &\text{by choice of }\tau\\
\ge\; & \mu\left(\val_{-i}, \bid\right) +  \left( \frac{\widetilde{p}_{b_{m-1}}\left(\val_{-i}, \bid\right)}{\widetilde{x}_{b_{m-1}}\left(\val_{-i}, \bid\right)} \right) \cdot \widetilde{x}_{b_{m-1}}\left(\val_{-i}, \bid\right) - \widetilde{p}_{b_{m-1}}\left(\val_{-i}, \bid\right) &\text{by \ref{srand:prop3}$^{(m-1)}$}\\
= \; & \mu\left(\val_{-i}, \bid\right) \geq \mu(\bid_{\tau}) &\text{by \cref{lem:miner-revenue-bound}}
\end{align*}
The last step comes from the assumption that  $\widetilde{x}_{\tau}(\bid_{\tau})=0$.
This strategy strictly benefits the coalition while harming honest user $u$, which contradicts $1$-CRHP.
\end{proof}   
Let $q=n+m-1$ denote the final bidder in the concatenated vector $(\val_{-i},\bid,t)$; keep this bidder's identity fixed as its bid $t$ varies.
Consider the value
\begin{equation}\label{eq:srand-delta*}
  \delta:= \min_{\bid \in S^{(1)}\times\dots\times S^{(m-1)}} \left(\alpha^{(m-1)} - \frac{p_{q}(\val_{-i},\bid,\alpha^{(m-1)})}{x_{q}(\val_{-i},\bid,\alpha^{(m-1)})}\right)   .
\end{equation}
By \cref{claim:confirmed} and \cref{cor:increasing}, $\frac{p_{q}(\val_{-i},\bid,\alpha^{(m-1)})}{x_{q}(\val_{-i},\bid,\alpha^{(m-1)})}  < \alpha^{(m-1)}$ for all $\bid \in S^{(1)} \times \dots \times S^{(m-1)}$. Indeed, raising this bidder's bid from any $t\in(\alpha^{(m-1)}-\eps^{(m-1)},\alpha^{(m-1)})$ to $\alpha^{(m-1)}$ strictly increases its nonnegative utility. This uses the fixed index $q$ even when $\alpha^{(m-1)}=b_{m-1}$, so it does not assume equal treatment of tied bids. 
Since there is only a finite number of possible $\bid$ by \ref{srand:prop1}$^{(m-1)},$ we have $\delta> 0.$  
Define 
    \[\eps^{(m)} := \frac{1}{2}\min \left( \delta,\eps^{(m-1)}\right) \quad \text{ and } \quad
    S^{(m)} :=\left\{ \alpha^{(m)} = \alpha^{(m-1)} - \eps^{(m)} , \beta^{(m)} = \alpha^{(m-1)}-\frac{1}{2} \eps^{(m)}\right\} \]
\begin{itemize}
    \item \ref{srand:prop1}$^{(m)}$ follows by construction.
    \item By construction, $\beta^{(m)} < \alpha^{(m-1)}$. Moreover, 
    \begin{align*}
        \alpha^{(m)}  = \alpha^{(m-1)} -\eps^{(m)} &> \max(\val) + \eps^{(m-1)} -\eps^{(m)}  &\text{by \ref{srand:prop2}$^{(m-1)}$}\\
        &\ge \max(\val) + \eps^{(m)}. &\text{by choice of $\eps^{(m)}$}
    \end{align*}
    \ref{srand:prop2}$^{(m)}$ thus follows by combining $\max(\val) + \eps^{(m)}< \alpha^{(m)}<\beta^{(m)}<\alpha^{(m-1)}$ and \ref{srand:prop2}$^{(m-1)}$. 
    \item Let $\bid = (b_1,\dots,b_m) \in S^{(1)} \times\dots\times S^{(m)}$. By \cref{claim:confirmed}, $\widetilde{x}_{b_m}(\val_{-i},\bid)>0.$ Furthermore, we have 
        \begin{align*}
        \frac{\widetilde{p}_{b_m}(\val_{-i}, \bid)}{\widetilde{x}_{b_m}(\val_{-i}, \bid)}
        &  \le  \frac{p_{q}(\val_{-i}, \bid_{-m},\alpha^{(m-1)})}{x_{q}(\val_{-i}, \bid_{-m},\alpha^{(m-1)})} \tag*{by \cref{cor:p/x-increasing}}\\
        &\le \alpha^{(m-1)}- \delta \tag*{by \cref{eq:srand-delta*}}\\
        &\le \alpha^{(m-1)}-2\eps^{(m)} = \alpha^{(m)}- \eps^{(m)}. \tag*{by choice of \text{$\eps^{(m)}$}}
        \end{align*}
        Thus, \ref{srand:prop3}$^{(m)}$ is satisfied.
\end{itemize}

\begin{corollary} \label{cor:rand-no-rev}
    Let $k$ denote the block size.
    If a TFM $({\bf I},{\bf C},{\bf P},\mu)$ satisfies UIC, MIC, and $1$-CRHP, then the miner revenue must be zero, i.e., $\mu(\bid) = 0$ for any $\bid$.
\end{corollary}
\begin{proof}
    Seeking contradiction, suppose that there was a transaction fee mechanism $({\bf I},{\bf C},{\bf P},{\bf R})$ that satisfies UIC, MIC, and $1$-CRHP, and has finite block size $k$, and for a bid vector $\bid,$ $\mu(\bid) >0$. Now consider $\mathbf{b'} = (\bid, \underbrace{0, \dots, 0}_{k+1}) = (b'_1,\dots,b'_n).$ By \cref{lem:subvec-mr}, $\mu(\bid')\ge \mu(\bid)>0.$ Thus, by budget feasibility (see \cref{sec:tfm-model}), there exists $i \in [n]$ where $x_i(\mathbf{b'}) > 0$. Additionally, $n>k,$ so by \cref{thm:rand-sCRHP-imp-k+1}, this is impossible, leading to a contradiction.
\end{proof}

\subsection{Impossibility Under Weak-CRHP}
\label{sec:weak-imp}



In this section, we present the proof of the stronger impossibility result under weak-CRHP. The proof relies on the following technical lemma. 
\begin{lemma}
\label{lem:rand-CRHP-ind-sets}
Suppose a TFM $(\mathbf{I}, \mathbf{C}, \mathbf{P}, \mathbf{R})$ satisfies UIC, MIC, and weak $1$-CRHP, and let $\val$ be a bid vector such that $\mu(\val) > 0$. Then, there exists a bid vector ${\bf b}^*$ such that 
\[|W(\bid^*)| > k.\]
\end{lemma}
Combining \cref{lem:rand-CRHP-ind-sets} and \cref{lem:limited-space}, we immediately get the following impossibility result:
\begin{theorem}[Restatement of \cref{intro:impossibility-weak-cRHP}]
\label{thm:rand-CRHP-imp-k+1}
    Let $k$ denote the block size.
    If a TFM $({\bf I},{\bf C},{\bf P},{\bf R})$ satisfies UIC, MIC, and weak $1$-CRHP, then miner revenue must be zero, i.e. $\mu(\bid) = 0$ for any $\bid$.
\end{theorem}
\begin{proof}
Note that \cref{lem:limited-space} holds for weak $1$-CRHP TFMs. The theorem thus follows from \cref{lem:rand-CRHP-ind-sets} and \cref{lem:limited-space}. 
\end{proof}
\paragraph{Proof of \cref{lem:rand-CRHP-ind-sets}}
As with previous proofs, we start with a bid vector $\val = (v_1,\dots,v_n)$ such that $\mu(\val) > 0$ and construct inductively a sequence of sets of vectors, such that in the $m$-th step, we can obtain a vector $\bid^{(m)}$ such that $|W(\bid^{(m)})|\geq m$.
Then when $m = k+1$, this contradicts \cref{lem:limited-space}.
To aide in our proof, we utilize the following implication of Erd\H{o}s's corollary~\cite[Corollary, p.~188]{Erdos1964} to determine the size of our sets.
\begin{lemma}{(from \cite[Corollary, p.~188]{Erdos1964})}\label{thm:G-R}
For all positive integers $p,a$ and all $\xi\in(0,1]$ there exists $N$ such that, if 
$G \subseteq A_1 \times \cdots \times A_p$, 
$|A_i| = N$ for each $i$, and if 
$|G| \ge \xi N^p$, 
then there exist subsets $B_i \subseteq A_i$ with 
$|B_i| \ge a$ such that 
\[
B_1 \times \cdots \times B_p \subseteq G.
\]
\end{lemma}
Like in the proof of \cref{lem:rand-sCRHP-ind-sets}, in each inductive step, we build a sequence of sets $S_1^{(m)},...,S_m^{(m)}$, each of size at least two.
To ensure that we can pick sets of size at least two in step $k+1$, we define the following sequence  $N_m$ for $m\in[k+2]$ backwards:
\begin{equation} \label{eqn:def-N}
\begin{split}
    &N_{k+2} := 2, \\
    &N_m := N\ \text{in }\text{\cref{thm:G-R} with }(p=k+m-1,\, 
    \xi{=}\tfrac{1}{n{+}m\cdot k},\,a{=}N_{m+1}), \text{ for }1\le\!m\!\le\!k+1.
\end{split}
\end{equation}
Specifically, for any \( m \in [k+1] \), we construct a sequences of sets \( S^{(m)}_1, \dots, S^{(m)}_m \), bid vectors $\val^{(m)}$,
and real numbers $\eps^{(m)}$ such that they satisfy the properties listed in \ref{rand:prop1}–\ref{rand:prop4} below.
    \begin{enumerate} [label=(${\sf Prop^*\text{-}\arabic*}$),leftmargin = 18mm]
        \item \label{rand:prop1} For all \( \ell \in [m] \), 
        $|S^{(m)}_{\ell}| =  N_{m+1}$.
        \item \label{rand:prop2_val} $\val \subseteq \val^{(1)}
        \subseteq \dots\subseteq\val^{(m)}$, and $|\val^{(m)}| = n + (k-1)m$. 
        \item \label{rand:prop2}For $m\geq 2$ and $\ell \in [m-1]$, 
        \[S_{\ell}^{(m)}\subseteq S_{\ell}^{(m-1)}\subseteq...\subseteq S_{\ell}^{(\ell)}.\]
        \item 
        \label{rand:prop3}
        $S^{(m)}_m \prec \cdots \prec S^{(m)}_1.$
        Furthermore,
        \[\max(\val^{(m)}) \le \min S^{(m)}_m - \eps^{(m)}\]
\item \label{rand:prop4} 
Fix any \( \ell \in [m] \) and let
$\bid = (b_1, \dots, b_\ell) \in S^{(m)}_1 \times \cdots \times S^{(m)}_\ell.$
Then, the following hold:
\[
\widetilde{x}_{b_\ell}\left(\val^{(\ell)},\bid\right) > 0,
\quad \text{and} \quad 
\frac{\widetilde{p}_{b_\ell}\left(\val^{(\ell)},\bid\right)}{\widetilde{x}_{b_\ell}\left(\val^{(\ell)},\bid\right)} \le \min S^{(m)}_m - \eps^{(m)}.
\]
\end{enumerate}
    We first show how the lemma follows assuming that the above properties holds for any $m\in[k+1]$ and give the induction proof afterwards.
    Let $m= k+1$. Consider $\val^{(k+1)}$ and the sequence of sets $S^{(k+1)}_1,...,S^{(k+1)}_{k+1}$.
    By \ref{rand:prop1} and \cref{eqn:def-N}, $|S^{(k+1)}_\ell|\ge2$ for all $\ell \in [k+1]$. 
    From \ref{rand:prop3}, we also know that $S^{(k+1)}_{k+1} \prec\cdots\prec S^{(k+1)}_1$, 
    and that for $\bid = (b_1,\dots,b_m) \in S^{(k+1)}_1\times\cdots\times S^{(k+1)}_{k+1}$, each $b_\ell$ is unique in $(\val^{(k+1)}, \bid)$. 
    Furthermore, $\widetilde{x}_{b_m}(\val^{(k+1)}, \bid)> 0$ by \ref{rand:prop4} with $\ell=m=k+1$.
 By weak symmetry, we may reorder these unique bids in increasing order of their sets. Thus, the assumptions of \cref{lem:W-size-set} hold for $\val^{(k+1)}$ and the sets  $S^{(k+1)}_{k+1},...,S^{(k+1)}_1$, which means that some $\bid^*$ exists where $|W(\bid^*)| \ge k+1$.
 
\paragraph{Induction Proof.} In the rest of the proof, we focus on the induction for proving properties \ref{rand:prop1} - \ref{rand:prop4}.
    For clarity, in our induction, we use superscripts to refer to the propositions for a particular $m \in [k+1]$, i.e. \ref{rand:prop1}$^{(m)}$ represents property \ref{rand:prop1} for $m$.
    
\smallskip\noindent \underline {\it Base Case: $m=1$.} Define a sequence of sets $A_z$ for $z\in[k]$ such that 
$$\set{\max(\val) +z+1} \prec A_z \prec \set{\max(\val) + z+2}$$
     and $|A_z | = N_1$ for all $z \in [k].$ 
     By construction $A_1\prec \dots\prec A_k$.
    We define $\phi(\cdot)$ to select a lowest-valued bid among those with positive confirmation probability, breaking ties by the index:
\begin{equation}
\label{eqn:defn-phi-B-rand}
\phi\left({\bf \widehat{b}}\right)
:=
\min\left\{
i:
x_i({\bf \widehat{b}})>0
\text{ and }
\widehat b_i
=
\min_{j:x_j({\bf \widehat{b}})>0}\widehat b_j
\right\}.
\end{equation}
    Consider the bid vector $(\val,\val')$ for $\val' \in A_1\times\dots\times A_{k}.$ 
    By \cref{lem:subvec-mr}, $\mu(\val,\val')\ge\mu(\val)>0$, and therefore, by budget feasibility, there exists an index $i$ such that $x_i(\val,\val')>0$. 
    Thus, $\phi(\val,\val')$ is well defined. 
    The $n+k$ possible values of $\phi(\val,\val')$ naturally give the following partition of $A_1\times\dots\times A_{k}:$
    \begin{align*}
        T_i &:= \left\{ \val' \in A_1 \times \dots \times A_{k} : \phi(\val,\val') = i\right\}
    \end{align*}
    for $i\in[ n+k]$. 

    Since this is a partition of  $k$ sets $A_1,\dots,A_k$ all of size $N_1$, there must exist a smallest index $i^*$ such that $|T_{i^*}| \ge \frac{1}{n + k} \cdot N_1^{ k}$.
    By \cref{thm:G-R} and the definition of $N_1$ and $N_2$, we know that there exists subsets $B_z \subseteq A_z$ for all $z \in[k]$ with $|B_z|= N_2$ such that $B_1 \times\dots\times B_k \subseteq T_{i^*}$.
    \begin{claim}\label{claim:rand-i-range}
        $i^* > n$.
    \end{claim}
    \begin{proof}
       Suppose for the sake of contradiction that $i^* \le n.$ 
       Choose a two-element set $V \subset (v_{i^*},\min B_1)\setminus\{v_1,\ldots,v_n\}$. This is possible because $\min B_1>\max(\val)$, and gives $V\prec B_1$.
       From Myerson's Lemma (\cref{Myerson's}), for $\val' \in B_1 \times \dots \times B_k$ and  $v^* \in V$, $x_{i^*}((\val_{-i^*},v^*), \val') \ge x_{i^*}(\val,\val') > 0$ where the last step comes from the definition of $i^*$.
       Thus, by construction, the hypothesis of \cref{lem:W-size-set} is satisfied with respect to $\val_{-i^*}$ and the sets $V,B_1,\dots,B_k$, and we can conclude that there exists $\bid^*$ such that $|W(\bid^*)|\ge k+1,$ which contradicts \cref{lem:limited-space}. 
    \end{proof} 
    Thus, we know that for every bid vector $(\val,\widehat{\val})$ where $\widehat{\val} \in B_1\times\cdots\times B_k\subseteq T_{i^*}$, we have $x_{i^*}(\val,\widehat{\val}) > 0$, so the bid falls in $B_{i^*-n}\subseteq A_{i^* - n}$ has a positive confirmation probability.
    Pick an arbitrary $\widehat{\val}\in B_1\times\cdots\times B_k\subseteq T_{i^*}$.
    Let $v^* = \max(\val) + k+ 3 > \widehat{v}_{i^*-n}$. 
    Consider a bid vector $\val^* = ((\val,\widehat{\val})_{-i^*}, v^*)$.
    Since $x_{i^*}(\val,\widehat{\val}) > 0$, by \cref{cor:increasing}, we know that $\util_{i^*}(\val^*)>\util_{i^*}(\val,\widehat{\val})\ge 0$, and thus $\tau:=\frac{p_{i^*}(\val^*)}{x_{i^*}(\val^*)} < v^*$. 
We define the following:
\begin{itemize}[leftmargin = 5mm]
    \item $\val^{(1)} := (\val,\widehat{\val})_{-i^*}$;
    \item $\eps^{(1)} :=\frac{1}{2}(v^*-\max\{\max A_k,\tau\}) > 0$;
    \item \( S^{(1)}_1  \text{ s.t. } \set{v^* - \eps^{(1)}}\prec S^{(1)}_1\prec \set{v^*} \text{ and } |S^{(1)}_1| = N_2\).
\end{itemize}
    We now prove the properties:
\begin{itemize}[leftmargin = 5mm]
    \item \ref{rand:prop1}$^{(1)}$ and \ref{rand:prop2_val}$^{(1)}$ follow by construction.
    \item \ref{rand:prop2}$^{(1)}$ is trivially satisfied.
    \item  \ref{rand:prop3}$^{(1)}$ follows from construction and the fact that 
    \begin{align*}
        \min S^{(1)}_{1} - \eps^{(1)} &> v^*-2\eps^{(1)} \\
        &= \max\{\max A_k,\tau\}\\
        &\ge \max(\val^{(1)}). \\
    \end{align*}
    \item \ref{rand:prop4}$^{(1)}$, note that for any $b_1 \in S^{(1)}_1$, and $b'\in B_{i^*-n}$, we know that $b_1>v^*-\eps^{(1)}>\max A_k\ge b'$. Reinserting $b'$ into $\val^{(1)}$ gives, up to a permutation, a vector in $\{\val\}\times B_1\times\cdots\times B_k$, so weak symmetry and the definition of $T_{i^*}$ imply that this bid has positive confirmation probability. Thus,
    \begin{align*}
        \widetilde{x}_{b_1}(\val^{(1)},b_1) &\ge \widetilde{x}_{b'}(\val^{(1)},b') \tag*{by Myerson's Lemma (\cref{Myerson's})}\\
        &> 0. \tag*{by definition of \text{$i^*$}}
    \end{align*}
    Furthermore,
        \begin{align*}
        \frac{\widetilde{p}_{b_1}(\val^{(1)},b_1)}{\widetilde{x}_{b_1}(\val^{(1)},b_1)}
        &\le \frac{\widetilde{p}_{v^*}(\val^{(1)},v^*)}{\widetilde{x}_{v^*}(\val^{(1)},v^*)} \tag*{by \cref{cor:p/x-increasing}}\\
        & = \tau \\
        & \le \max\{\max A_k,\tau\}\\
        &\le \min S^{(1)}_1 - \epsilon^{(1)}.
        \end{align*}
        Thus, \ref{rand:prop4}$^{(1)}$ is satisfied.
\end{itemize}
\smallskip\noindent \underline {\it Inductive Step: $m>1$.} Suppose $S^{(m-1)}_1, \dots, S^{(m-1)}_{m-1}$, $\val^{(1)}, \dots, \val^{(m-1)}$, and $\eps^{(m-1)}$ are the sequences of sets, bid vectors, and the real number respectively that satisfy \ref{rand:prop1}$^{(m-1)}$-\ref{rand:prop4}$^{(m-1)}$. 
For simplicity, let $\gamma := \min S^{(m-1)}_{m-1}$, and 
\begin{equation}\label{def:alpha-delta}
    \alpha := \gamma -\eps^{(m-1)}, \quad\Delta := \frac{\eps^{(m-1)}}{k+2}.
\end{equation}
Then for \( z \in [k] \), redefine $A_z$ such that 
\begin{equation}\label{eq:A-Z}
   \set{\alpha+ z\cdot \Delta } \prec A_z \prec \set{\alpha + (z+1)\cdot \Delta } 
\end{equation}
and $|A_z| = N_m$ for $z \in [k]$. By construction, we have $A_1\prec\dots\prec A_k$.

    Consider the bid vector $\val^* = (\val^{(m-1)},\val')$ where $\val' \in \left(\prod^{k}_{z=1}A_z\right) \times \left(\prod^{m-1}_{\ell=1}S^{(m-1)}_\ell\right).$ 
    Then $|\val^*| = n+k\cdot m$.
    Recall the definition of $\phi(\cdot)$ from \cref{eqn:defn-phi-B-rand}.
    Again, by \cref{lem:subvec-mr} and \ref{rand:prop2_val}$^{(m-1)}$, $\mu(\val^*) \ge \mu(\val) > 0$. 
    Therefore, by budget feasibility, there exists $i \in [n + k\cdot m]$ such that $x_i(\val^*)>0$, and thus, $\phi(\val^*)$ is well defined.
    The $n+ m \cdot k$ possible values of $\phi(\val^*)$ naturally give the following partition of $\left(\prod^{k}_{z=1}A_z\right) \times \left(\prod^{m-1}_{\ell=1}S^{(m-1)}_\ell\right)$.
\begin{align*}
    T_i &:= \left\{ 
    \val' \in 
    \left(\prod^{k}_{z=1}A_z\right) \times \left(\prod^{m-1}_{\ell=1}S^{(m-1)}_\ell\right) 
    : \phi( \val^{(m-1)},\val') = i 
    \right\} 
\end{align*}
    for $i\in[n+k\cdot m]$. 
    The sequence $\{T_i\}_{i\in[n+k\cdot m]}$ is a partition of the product of $k +m -1$ sets $A_1,\dots,A_k,S_1,\dots S_{m-1}$, each of size $N_m$ by \ref{rand:prop1}$^{(m-1)}$.
    Thus, by the pigeon hole principle, let $i^*$ be the smallest index such that $|T_{i^*}| \ge \frac{1}{n + m\cdot k} \cdot N_m^{ k+m-1}$. 
   By \cref{thm:G-R} and the definitions of $N_m$ and $N_{m+1}$, there exist subsets
\[
B_z \subseteq A_z \ \text{for all } z \in [k], 
\quad\text{and}\quad 
B_z \subseteq S^{(m-1)}_{z-k} \ \text{for all } z \in \{k+1, \dots, k+m-1\},
\]
each satisfying $|B_z| = N_{m+1}$, such that
\[
B_1 \times \cdots \times B_{k+m-1} \subseteq T_{i^*}.
\]
    
    Note that in any vector $\val^* = (\val^{(m-1)},\val')$ where $\val' \in \left(\prod^{k}_{z=1}A_z\right) \times \left(\prod^{m-1}_{\ell=1}S^{(m-1)}_\ell\right)$, the indices can be understood as follows: 
    \begin{itemize}[leftmargin = 5mm]
        \item The bids $\val^*_i$ for $1\leq i\leq n+(k-1)(m-1)$ come from $\val^{(m-1)}$: $\val^*_i = \val^{(m-1)}_i$.
        \item The bids $\val^*_i$ for $n+ (k-1)\cdot (m-1) + 1\leq i\leq n+ (k-1)\cdot (m-1) +k$ come from $A_1\times\dots\times A_k$: $\val^*_i \in A_{i-n-(k-1)(m-1)}$. 
        \item The bids $\val^*_i$ for $n+ (k-1)\cdot (m-1) + k+1\leq i\leq n+ k\cdot m$ come from $S_1^{(m-1)}\times\dots\times S^{(m-1)}_{m-1}$: $\val^*_i \in S^{(m-1)}_{i-n-(k-1)(m-1)-k}$. 
    \end{itemize}
    Let $L := n+ (k-1)\cdot (m-1)$ and $R := n+ (k-1)\cdot (m-1) +k$ be the left and right boundary indices of bids that come from $A_1$ through $A_k$.
    \begin{claim}\label{claim:ind-rand-i-range-down}
        $L <i^*$.
    \end{claim}
    \begin{proof}
    Suppose for the sake of contradiction that $i^* \le L$. 
    Let $\bid = (b_1 ,\dots,b_{m-1} ) \in B_{k+1} \times \dots \times B_{k+m-1}$, and
       choose a two-element set $V\subset(\val^{(m-1)}_{i^*},\min B_1)\setminus\{\val^{(m-1)}_1,\ldots,\val^{(m-1)}_L,b_1,\ldots,b_{m-1}\}$. This is possible because $\min B_1>\max(\val^{(m-1)})$, and gives $V\prec B_1$. Raising bid $i^*$ to any value in $V$ preserves its positive confirmation probability by Myerson's Lemma.
    Similarly to the base case, by construction, the hypothesis of \cref{lem:W-size-set} is satisfied w.r.t $(\val^{(m-1)}_{-i^*}, b_1,\dots,b_{m-1})$ and the sets $V,B_1,\dots,B_k$. 
     Thus, we can conclude that there exists $\bid^*$ such that $|W(\bid^*)|>k,$ which contradicts \cref{lem:limited-space}. 
    \end{proof} 
    \begin{claim}\label{claim:ind-rand-i-range-up}
        $i^* \le R$.
    \end{claim}
    \begin{proof}
Suppose, for the sake of contradiction, that $i^*>R$.
Let $j^*:=i^*-R$, so the bid selected by $\phi(\cdot)$ comes from $S_{j^*}^{(m-1)}$.
Pick
\[
\val'=(a_1,\dots,a_k,b_1,\dots,b_{m-1})
\in
\prod_{z=1}^{k+m-1}\rminf(B_z).
\]
In particular, $b_{j^*}\in\rminf(B_{k+j^*})$.
Since $B_1\times\cdots\times B_{k+m-1}\subseteq T_{i^*},$
the bid $b_{j^*}$ has positive confirmation probability for every choice of the coordinates from the corresponding $B_z$'s.
If we replace $b_{j^*}$ by $\min B_{k+j^*}$ while keeping all other bids fixed, the resulting bid vector is still in $T_{i^*}$, so this lower bid also has positive confirmation probability.
Therefore, by \cref{cor:increasing}, user $i^*$ has strictly positive utility at bid $b_{j^*}$.

Consider the world where $(\val^{(m-1)},\val')$ is the honest bid vector.
Let $u$ be the user with bid $a_k$.
By \cref{eq:A-Z} and \ref{rand:prop3}$^{(m-1)}$, we have $a_k<b_{j^*}.$

Since $\phi(\val^{(m-1)},\val')=i^*$ selects a lowest-valued bid with positive confirmation probability, user $u$ must have zero confirmation probability in the honest execution.
Thus, the coalition consisting of the miner and user $u$ has honest joint utility
\[
\mu(\val^{(m-1)},\val').
\]

Now consider the following deviation.
The miner censors:
\begin{itemize}[leftmargin = 5mm]
    \item the bid $b_{j^*}$ from user $i^*$;
    \item the bids $a_1,\dots,a_{k-1}$;
    \item all bids in $\val^{(m-1)}$ that are not in $\val^{(j^*)}$;
    \item all bids $b_\ell$ for $\ell\in\{j^*+1,\dots,m-1\}$.
\end{itemize}
User $u$ changes its bid from $a_k$ to $b_{j^*}$, and the miner then performs the honest inclusion rule on the resulting bid vector
\[
\val^*=(\val^{(j^*)},b_1,\dots,b_{j^*}).
\]

We first compare miner revenue.
Except for the bid $b_{j^*}$, which is replaced by user $u$'s bid of the same amount, every bid removed by the deviation has value strictly smaller than $b_{j^*}$.
For the bids $a_1,\dots,a_{k-1}$ this follows from \cref{eq:A-Z}; for $b_{j^*+1},\dots,b_{m-1}$ it follows from \ref{rand:prop3}$^{(m-1)}$; and for the bids in $\val^{(m-1)}\setminus\val^{(j^*)}$ it follows from \ref{rand:prop2_val} and \ref{rand:prop3} applied to the subsequent inductive steps.

Since $\phi(\val^{(m-1)},\val')=i^*$ selects a lowest-valued bid with positive confirmation probability, all these strictly lower bids have zero confirmation probability, and hence zero payment.
Replacing the original bid $b_{j^*}$ by user $u$'s bid of the same amount does not change the mechanism outcome by weak symmetry.
Therefore, by \cref{lem:miner-revenue-bound},
\[
\mu(\val^*)\ge\mu(\val^{(m-1)},\val').
\]
Conversely, $\val^*$ is a subvector of $(\val^{(m-1)},\val')$ in bid amounts, so \cref{lem:subvec-mr} gives
\[
\mu(\val^*)\le\mu(\val^{(m-1)},\val').
\]
Hence $\mu(\val^*)=\mu(\val^{(m-1)},\val').$

It remains to consider user $u$'s utility.
Since $(b_1,\dots,b_{j^*})
\in
S_1^{(m-1)}\times\cdots\times S_{j^*}^{(m-1)},$
\ref{rand:prop4}$^{(m-1)}$ gives
\[
\frac{\widetilde p_{b_{j^*}}(\val^*)}
{\widetilde x_{b_{j^*}}(\val^*)}
\le
\gamma-\eps^{(m-1)}.
\]
Moreover, by \cref{eq:A-Z}, we have $a_k>\gamma-\eps^{(m-1)}.$
Therefore,
\begin{align*}
a_k\widetilde x_{b_{j^*}}(\val^*)
-\widetilde p_{b_{j^*}}(\val^*)
>
(\gamma-\eps^{(m-1)})
\widetilde x_{b_{j^*}}(\val^*)
-\widetilde p_{b_{j^*}}(\val^*)\ge 0.
\end{align*}
Thus user $u$ obtains strictly positive utility under the deviation, while the miner revenue is unchanged.
The coalition therefore strictly increases its joint utility.

At the same time, user $i^*$ had strictly positive utility in the honest execution and its bid is censored by the deviation.
This contradicts weak $1$-CRHP.
\end{proof}
    By \cref{claim:ind-rand-i-range-down} and \cref{claim:ind-rand-i-range-up}, for every $(\val^{(m-1)},\val')$ where $\val' \in B_1\times\cdots\times B_{k+m-1}\subseteq T_{i^*}$, we know that a bid with positive confirmation probability exists within the set $B_{i^* - L} \subseteq A_{i^* - L}$.
    Pick an arbitrary $\widehat{\val}\in B_1\times \dots\times B_k$, and define $\val^{(m)} = \left(\val^{(m-1)},\widehat{\val}\right)_{-i^*}$.
  
Recall that $\gamma := \min S^{(m-1)}_{m-1}$.
Consider the value
\begin{equation}\label{eq:rand-delta}
\delta:=
\min_{\bid \in B_{k+1}\times\cdots\times B_{k+m-1}}
\left(
\gamma-
\frac{p_R(\val^{(m)},\gamma,\bid)}
{x_R(\val^{(m)},\gamma,\bid)}
\right).
\end{equation}

We show that $\delta>0$.
Fix any $\bid\in B_{k+1}\times\cdots\times B_{k+m-1}.$
Choose $\gamma'$ such that $\max A_k<\gamma'<\gamma.$
Since $B_1\times\cdots\times B_{k+m-1}\subseteq T_{i^*},$
the bid at index $i^*$ has positive confirmation probability in $(\val^{(m-1)},\widehat{\val},\bid).$
This bid comes from one of $A_1,\dots,A_k$, so its value is smaller than $\gamma'$.
By \cref{cor:increasing}, raising this bid to $\gamma'$ gives it strictly positive utility.
The resulting bid vector is a permutation of $(\val^{(m)},\gamma',\bid),$
so by weak symmetry, the bid $\gamma'$ at index $R$ has strictly positive utility in $(\val^{(m)},\gamma',\bid)$.

Since $\gamma>\gamma'$, \cref{cor:increasing} further implies that the bid $\gamma$ at index $R$ has strictly positive utility in $(\val^{(m)},\gamma,\bid)$.
Therefore,
\[
\frac{p_R(\val^{(m)},\gamma,\bid)}
{x_R(\val^{(m)},\gamma,\bid)}
<\gamma.
\]
This holds for every
$\bid\in B_{k+1}\times\cdots\times B_{k+m-1}$.
Since this set is finite, we have $\delta>0$.
Thus, we can define the following:
\begin{itemize}[leftmargin = 5mm]
\item $\val^{(m)} = \left(\val^{(m-1)},\widehat{\val}\right)_{-i^*}$;
    \item $S^{(m)}_1,\dots, S^{(m)}_{m-1}= B_{k+1}, \dots,B_{k+m-1}$;
    \item \( \eps^{(m)} := \frac{1}{2} \min\left(\delta,\Delta\right)\);
    \item \( S^{(m)}_m \text{ s.t. } \set{\gamma - \eps^{(m)}}\prec S^{(m)}_m\prec \set{\gamma} \text{ and } |S^{(m)}_m| = N_{m+1}\).
\end{itemize}
    We now prove the properties:
\begin{itemize}[leftmargin = 5mm]
    \item \ref{rand:prop1}$^{(m)}$,  \ref{rand:prop2_val}$^{(m)}$, and \ref{rand:prop2}$^{(m)}$ are satisfied by construction.
    
    \item \ref{rand:prop3}$^{(m)}$ follows from construction.
Since $S^{(m)}_\ell=B_{k+\ell}\subseteq S^{(m-1)}_\ell$ for every $\ell\in[m-1]$ and $S^{(m)}_m\prec\{\gamma\}\preceq S^{(m)}_{m-1}$, we have $S^{(m)}_m\prec\cdots\prec S^{(m)}_1.$
Moreover, by \ref{rand:prop3}$^{(m-1)}$ and \cref{eq:A-Z},
\[
\max(\val^{(m)})\le \max A_k<\gamma-\Delta.
\]
Since $2\eps^{(m)}\le\Delta$ and $\min S^{(m)}_m>\gamma-\eps^{(m)}$, we have
\[
\max(\val^{(m)})
<
\gamma-\Delta
\le
\gamma-2\eps^{(m)}
<
\min S^{(m)}_m-\eps^{(m)}.
\]
Thus, \ref{rand:prop3}$^{(m)}$ holds.
    

\item We first consider $\ell\in[m-1]$.
Since $S^{(m)}_z\subseteq S^{(m-1)}_z$ for every $z\in[m-1]$, the positive-confirmation guarantee in \ref{rand:prop4}$^{(m-1)}$ continues to hold.
Moreover,
\[
\min S^{(m)}_m-\eps^{(m)}
>
\gamma-2\eps^{(m)}
\ge
\gamma-\Delta
>
\gamma-\eps^{(m-1)}
=
\min S^{(m-1)}_{m-1}-\eps^{(m-1)}.
\]
Therefore, the payment-ratio bound in \ref{rand:prop4}$^{(m-1)}$ also implies the required bound in \ref{rand:prop4}$^{(m)}$ for every $\ell\in[m-1]$.

It remains to prove \ref{rand:prop4}$^{(m)}$ for $\ell=m$.
Let $\bid=(b_1,\dots,b_m)\in S^{(m)}_1\times\cdots\times S^{(m)}_m$
and let $\ell^*:=i^*-L$.
By construction,
\[
(b_1,\dots,b_{m-1})\in B_{k+1}\times\cdots\times B_{k+m-1}.
\]
Fix any $b'\in B_{\ell^*}$.
Together with the fixed bids in $\val^{(m)}$ and $(b_1,\dots,b_{m-1})$, reinserting $b'$ gives, up to a permutation, a bid vector whose coordinates outside $\val^{(m-1)}$ lie in $B_1\times\cdots\times B_{k+m-1}\subseteq T_{i^*}.$
Hence, by weak symmetry and the definition of $T_{i^*}$,
\[
\widetilde{x}_{b'}(\val^{(m)},\bid_{-m},b')>0.
\]
On the other hand, by \cref{eq:A-Z} and the choice of $\eps^{(m)}$,
\[
b_m>\gamma-\eps^{(m)}>\gamma-\Delta>\max A_k\ge b'.
\]
Therefore, Myerson's Lemma gives
\[
\widetilde{x}_{b_m}(\val^{(m)},\bid)>0.
\]

It remains to prove the payment-ratio bound.
Since $|\val^{(m)}|=R-1$ and $b_m$ is unique, weak symmetry gives
\[
\frac{\widetilde{p}_{b_m}(\val^{(m)},\bid)}
{\widetilde{x}_{b_m}(\val^{(m)},\bid)}
=
\frac{p_R(\val^{(m)},b_m,\bid_{-m})}
{x_R(\val^{(m)},b_m,\bid_{-m})}.
\]
Since $b_m<\gamma$, \cref{cor:p/x-increasing} and the definition of $\delta$ imply
\begin{align*}
\frac{\widetilde{p}_{b_m}(\val^{(m)},\bid)}
{\widetilde{x}_{b_m}(\val^{(m)},\bid)}
&\le
\frac{p_R(\val^{(m)},\gamma,\bid_{-m})}
{x_R(\val^{(m)},\gamma,\bid_{-m})}\\
&\le
\gamma-\delta\\
&\le
\gamma-2\eps^{(m)}\\
&<
\min S^{(m)}_m-\eps^{(m)}.
\end{align*}
Thus, \ref{rand:prop4}$^{(m)}$ holds.
\end{itemize}

\begin{corollary} \label{cor:weak-trivial}
    Only trivial TFMs, where no bids are ever confirmed, satisfy UIC, MIC, MRHP, and weak $1$-CRHP in the plain model.
\end{corollary}
\begin{proof}
    By \cref{thm:rand-CRHP-imp-k+1}, a TFM that satisfies UIC, MIC, and weak $1$-CRHP in the plain model must have zero miner revenue.
    Suppose for the sake of contradiction that for some bid vector $\bid$,  $x_i(\bid)>0$. By \cref{cor:increasing}, this implies a bid vector $\bid'$ exists where $\util_i(\bid')>\util_i(\bid)\ge 0$.
    Since the miner revenue is always $0$, in a world where $\bid'$ is the honest bid vector, the miner can ignore all bids to harm user $i$ at no cost, contradicting MRHP.
\end{proof}

\section{Characterizations in the MPC-Assisted Model}
\label{sec:mpc}
In this section, we present the feasibility and impossibility results in the MPC-assisted model.
Recall that in the MPC-assisted model, a committee of $M$ miners jointly run a multi-party computation (MPC) to implement the TFM.

\subsection{Feasibility in the MPC-assisted Model}
\label{sec:mpc-feasibility}
In this section, we present a mechanism that achieves all desired properties, UIC, MIC, all RHP variants, and positive revenue in the MPC-assisted model.

\begin{mdframed}
\begin{center}
        \textbf{MPC-assisted, posted price with random selection}
    \end{center}
    
    \noindent //Reserve price ${\sf res} > 0$, revenue parameter $0 < \epsilon \leq \res$.
    \begin{itemize}[leftmargin = 5mm, label=\textbullet]
        \item \textbf{Inclusion \& Confirmation}: Let $E$ be the set of bids at least the reserve price ${\sf res}$. If $|E|\le k$, include and confirm all bids in $E$. Otherwise, uniformly at random choose $k$ bids in $E$ to include and confirm.
        \item \textbf{Payment \& Revenue}: Each confirmed bid pays the reserve price $\res$.
        If at least one bid is confirmed, the total miner revenue is $\epsilon$, which is divided equally among the $M$ miners; otherwise, the total miner revenue is zero.
    \end{itemize}
\end{mdframed}

\revision{
\begin{theorem}
\label{thm:ppw-rand-mpc}
In the MPC-assisted model, posted price with uniform random selection satisfies UIC, MIC, URHP, MRHP, and $d$-CRHP for every $1\le d\le k$.
Total miner revenue is positive when $|E|\geq 1$.
\end{theorem}
}
\begin{proof}
Let $E$ be the set of eligible submitted bids at least reserve, and let $N=|E|$. If $N=0$, no bid is confirmed and total miner revenue is zero. If $N>0$, each eligible bid is confirmed with probability $q(N):=\min\{1,k/N\}$ under uniform random selection, pays $\res$ conditional on confirmation, and the total miner revenue is $\epsilon$.

\smallskip\noindent\textbf{UIC:} For a single real bid and fixed other bids, the user's allocation probability is zero below reserve and equal to a constant $q$ at or above reserve; the expected payment is $\res\cdot q$. This is the Myerson threshold payment. Fake bids cannot help: every eligible fake bid has true value zero and, whenever confirmed, pays reserve; it also weakly lowers the real bid's confirmation probability. If the user's true value is below reserve, every strategy gives nonpositive utility. Otherwise, fake eligible bids cannot increase the confirmation probability of its real bid. Omitting the real bid also gives nonpositive utility. Therefore truthful bidding with no fake bids is optimal.

\smallskip\noindent\textbf{MIC:} A miner coalition of size $m<M/2$ can inject fake bids. It cannot censor honest bids or alter the random selection, and the MPC provides guaranteed output delivery. If $N>0$, honest total miner revenue is already $\epsilon$, the maximum possible. Fake bids cannot increase it and may require fake payments. If $N=0$, fake eligible bids can create revenue $\epsilon$, but the coalition receives only $m\epsilon/M$ while paying at least $\res$ for each confirmed fake bid. Since $m<M/2$ and $\epsilon\le \res$, this is not profitable. Thus MIC holds.

\smallskip\noindent\textbf{URHP:} A user-only deviation can harm another user only by reducing that user's confirmation probability. This requires adding eligible competition. Adding eligible fake bids imposes strictly positive expected reserve payments on fake bids and weakly lowers the deviator's own real-bid confirmation probability. Turning a below-reserve real bid into an eligible bid gives strictly negative expected surplus. Changing an already eligible real bid has no effect on the allocation except through fake bids. Therefore any user-harming deterministic bidding strategy strictly lowers the deviating user's expected utility. 
By UIC, no deterministic bidding strategy improves on the deviating user's honest expected utility. Thus a randomized deviation that does not lower the deviator's expected utility uses deterministic bidding strategies that preserve this expected utility with probability one, and each such strategy weakly preserves every honest user's expected utility. Hence URHP holds.

\smallskip\noindent\textbf{MRHP:}
A miner-only coalition can harm an honest user only by adding eligible fake bids that dilute the user's confirmation probability.
Adding eligible fake bids cannot increase the total miner revenue beyond the fixed $\epsilon$ and imposes payments whenever a fake bid is confirmed.
If there is an eligible honest bid, the block remains nonempty and every honest miner outside the coalition continues to receive $\epsilon/M$; if there is no eligible honest bid, every honest miner's baseline revenue is zero.
Hence any deterministic bidding strategy that harms a protected honest user makes the deviating miner coalition strictly worse off, while honest miners outside the coalition cannot be harmed. By MIC, no deterministic strategy improves on honest utility. Thus a randomized deviation that does not lower expected utility uses utility-neutral deterministic strategies with probability one, each preserving every protected user's and miner's utility.

\revision{
\smallskip\noindent\textbf{$d$-CRHP:}
Let $\C$ be a coalition containing $1\le m<M/2$ miners and a set $\mathcal U$ of at most $d\le k$ users.
We show that any deviation that does not decrease the coalition's expected utility cannot harm any user or miner outside $\C$.

Let $h$ be the number of eligible honest users outside $\C$ and $t$ be the number of users in $\mathcal U$ whose true values are at least $\res$, and define
\[
s:=\sum_{i\in\mathcal U}\max\{v_i-\res,0\}
\text{ and }
a:=\frac{m\epsilon}{M}.
\]
Recall that $q(N):=\min\{1,k/N\}$ for $N>0$, and set $q(0):=0$.
Under honest behavior, there are $h+t$ eligible bids, so the coalition's expected utility is
\[
\util_{\C}(\val;H_{\C})
=
a\cdot \mathbf 1_{\{h+t>0\}}+q(h+t)\cdot s.
\]

Now consider an arbitrary (possibly randomized) strategy $S_{\C}$ by $\C$.
Let $T\subseteq\mathcal U$ be the set of colluding users whose primary bids are eligible under $S_{\C}$, and let $f$ be the number of eligible fake bids.
Then under strategy $S_{\C}$, we have $N = h+|T|+f$. 
Let $z = \sum_{i\in T}(v_i-\res) - f\cdot \res$.
Then $z\leq s$.

Conditioned on an arbitrary realization of input bid vector under $S_{\C}$, every eligible bid is confirmed with probability $q(N)$.
Hence the coalition's expected utility, conditioned on this realization of bid vector, is $a\cdot \mathbf 1_{\{N>0\}}+q(N)\cdot z.$
We split into two cases.

\noindent{\bf Case 1: $h>0$}
Then $N>0$ under every deviation, so every honest miner outside $\C$ continues to receive revenue $\epsilon/M$.
An eligible honest user $j\notin\C$ has expected utility $(v_j-\res)\E[q(N)]$
under the deviation, where the expectation is taken over randomness used in $S_{\C}$. 
Compared with $j$'s utility $(v_j-\res)q(h+t)$ under honest behavior, it can be harmed only if
\[
\E[q(N)]<q(h+t).
\]

Suppose first that $s>0$.
If an honest user is harmed, the coalition's expected utility under the deviation satisfies
\[
\util_{\C}(\val;S_{\C})
=
a+\E[q(N)z]
\le
a+s\E[q(N)] < a + s\cdot q(h+t)=
\util_{\C}(\val;H_{\C}).
\]
Therefore, the coalition must also be harmed.

It remains to consider $s=0$.
In this case, $z\le0$ for every realized bidding strategy, while the coalition's honest utility is $a$.
Since $q(N)>0$ when $h>0$, a deviation with expected utility at least $a$ must satisfy $z=0$ with probability one.
Therefore, it must be that $f = 0$ and for every $i\in T$, we have $v_i = \res$.
Thus under the strategy $S_{\C}$, every eligible bid from the coalition must come from one of the $t$ users whose true value is exactly $\res$.
Hence, $N\le h+t$ with probability one.
Thus $q(N)\ge q(h+t)$ with probability one, so no honest user is harmed.

\noindent{\bf Case 2: $h = 0$.}
All honest users outside $\C$ are below the reserve and therefore have utility zero under both honest behavior and any deviation, so only honest miners outside $\C$ need to be considered.
If $t=0$, the honest total miner revenue is zero, so no honest miner outside $\C$ can be harmed.
Therefore, we focus on the case where $t>0$.
Since $t\le|\mathcal U|\le d\le k$, all $t$ honestly eligible colluding users are confirmed under honest behavior, and
\[
\util_{\C}(\val;H_{\C})=a+s.
\]
For every realized bid vector where $N > 0$, $q(N)z\le s,$ because $z\le s$ and $q(N)\le1$.
Hence conditioned on any realized bid vector such that $N > 0$, coalition's conditional utility is at most $a+s$; while under any bid vector such that $N = 0$, coalition's utility is zero.
If $\rho:=\Pr[N=0]$, then
\[
\util_{\C}(\val;S_{\C})
\le
(1-\rho)(a+s).
\]
Since $a>0$, any deviation satisfying
\[
\util_{\C}(\val;S_{\C})
\ge
\util_{\C}(\val;H_{\C})=a+s
\]
must have $\rho=0$.
Therefore the block remains nonempty with probability one, and every honest miner outside $\C$ continues to receive $\epsilon/M$.

Thus every deviation that weakly increases the coalition's expected utility weakly preserves the utility of every honest user and miner outside the coalition.
The mechanism therefore satisfies $d$-CRHP for every $1\le d\le k$.}
\end{proof}

\smallskip\noindent{\bf Not $1$-SCP:} Assume $M\ge3$ and choose $1\le m<M/2$. Suppose no honest bid is eligible and a colluding user has true value $v\in(\res-m\epsilon/M,\res)$, where $m$ is the number of colluding miners. If the user bids at least reserve, it is confirmed, pays reserve, and the miner coalition obtains revenue share $m\epsilon/M$. The coalition's gain is $m\epsilon/M-(\res-v)>0$. This deviation does not harm any honest user, consistent with CRHP.

\subsection{Characterization of Deterministic Mechanisms}
\subsubsection{Deterministic Feasibility in the MPC-Assisted Model}
While the previous mechanism achieves all desired properties, it relies on randomness to choose the set of $k$ bids to confirm.
In this section, we show that the all-or-nothing posted price is a deterministic mechanism satisfies all desired properties in the MPC-assisted model.
\begin{mdframed}
\begin{center}
    {\bf MPC-assisted, all or nothing posted price auction}. 
\end{center}    
\noindent //Reserve price ${\sf res} > 0$, revenue parameter $0 < \epsilon \leq \res$.
       
    \noindent \textbf{Mechanism:}
    \begin{itemize}[label=\textbullet]
        \item \textbf{Inclusion Rule \& Confirmation Rule:} If there are $\le k$ bids and each bid is strictly greater than ${\sf res}$, include and confirm all bids. Otherwise, include and confirm no bids. The strict threshold is intentional for the URHP boundary case. 
        \item \textbf{Payment Rule \& Miner Revenue Rule:} Each confirmed bid pays the reserve price $\res$.
        The total miner revenue is $\epsilon$ per confirmed bid and is divided equally among the $M$ miners.
    \end{itemize}
\end{mdframed}
\begin{lemma} \label{cr-ub}
The above TFM satisfies UIC, MIC, URHP, MRHP, $d$-CRHP for any $d\ge 1$, and has positive miner revenue in the MPC-assisted model.
\end{lemma}
\begin{proof}
The strict-threshold all-or-nothing argument from \cref{plain-AorN} gives UIC and URHP, and proves the user-protection requirements of MRHP and $d$-CRHP. In the MPC-assisted model, MRHP and $d$-CRHP additionally protect honest miners outside the strategic coalition. We verify this additional requirement and MIC below.

\smallskip\noindent\textbf{MIC:}
A miner coalition of size $1\le m<M/2$ cannot censor honest bids and can only inject fake bids.
The MPC provides guaranteed output delivery.
An input bid vector fails only if it contains more than $k$ bids or a bid at or below reserve; injecting fake bids cannot fix either condition.
Since $m<M/2$ and $\epsilon\le\res$, let $a:=m\epsilon/M<\res$ denote the coalition's share of the miner revenue generated by each confirmed bid. Suppose first that the honest bid vector contains $N\le k$ bids, all strictly above $\res$, so that all $N$ bids are confirmed and the coalition's honest utility is $Na$. If the coalition injects $f$ fake bids and the resulting bid vector still satisfies the confirmation condition, its utility becomes
\[
(N+f)a-f\res=Na+f(a-\res)\le Na.
\]
If instead the injection causes the mechanism to confirm no bids, the coalition obtains utility $0\le Na$. Thus injecting fake bids cannot increase the coalition's utility.
Therefore MIC holds, including for randomized strategies by averaging.

It remains to prove that honest miners outside the coalition are protected against miner coalition or miner-user coalition.
Fix a coalition $\C$ containing $1\le m<M/2$ miners and possibly some users.

If the honest execution confirms no bid, then every honest miner outside $\C$ receives zero revenue and therefore cannot be harmed.
Suppose instead that the honest execution confirms $N>0$ bids.
Then $N\le k$, every bid is strictly greater than $\res$, and every miner receives revenue $N\epsilon/M$.

Fix an arbitrary strategy of $\C$.
If the deviation ends up with no bid being confirmed, then the coalition loses its positive honest miner revenue, so the deviation strictly decreases the coalition's utility.

It remains to consider a deviation under which the mechanism confirms a positive number of bids.
Changing a colluding user's bid while keeping it strictly above $\res$ does not affect either the allocation or payment.
Dropping a colluding user's bid with value $v_i>\res$ decreases the coalition's utility by $(v_i-\res)+\frac{m\epsilon}{M}>0,$
because the coalition loses both the user's utility and its share of the miner revenue generated by this bid.
Similarly, adding a fake bid decreases the coalition's utility by $\res-\frac{m\epsilon}{M}>0$, 
because the fake bid pays $\res$, while the colluding miners receives only an $m/M$ fraction of the additional revenue $\epsilon$.

Therefore, every deviation that either drops a colluding user's bid or adds a fake bid strictly decreases the coalition's utility.
Any deviation that does not decrease the coalition's utility must therefore preserve the same number of confirmed bids as under honest behavior.
Hence every honest miner outside $\C$ continues to receive revenue $N\epsilon/M$.

The same argument applies to randomized strategies, since every realization gives the coalition utility at most its honest utility, a randomized deviation can have expected utility at least the honest utility only if the realized bidding strategy preserves the same number of confirmed bids and miner revenue with probability one.
Thus MRHP and $d$-CRHP also protect honest miners outside the coalition.

\end{proof}
\smallskip\noindent{\bf Not $1$-SCP:} 
Assume $M\ge3$, so a coalition with one miner is allowed. Imagine a world with $k$ users, where $k-1$ of them have true value above $\res$, but one has a true value $(\res -\frac{\eps}{2M})$.
Consider the coalition consisting of one miner and this user with true value below $\res$.
In the honest case, no one gets confirmed, so this coalition gets a joint utility of $0$.
However, if the user raises their bid to $\res + 1$, the joint utility now becomes $k\cdot\frac{\eps}{M} + (\res -\frac{\eps}{2M}) - \res >0$.

\subsubsection{Deterministic Impossibility in the MPC-Assisted Model} \label{sec:mpc-det-imp}
The above all-or-nothing posted price is degenerate in the sense that it confirms no bid when there are more than $k$ users. 
Unfortunately, this is necessary for deterministic mechanisms, even in the MPC-assisted model.
\begin{lemma} \label{lem:cr-include-above}
   Suppose a deterministic TFM $({\bf I},{\bf C},{\bf P},{\bf R})$ satisfies UIC and URHP in the MPC-assisted model. For any bid vector $\bid = (b_1,\dots, b_n)$, if there exists some $b_i \in \bid$ such that $\util_{i}(\bid) >0$, then for every $j$ where $b_j$ is unique in $\bid$ and $b_j > b_i$, it must hold that $x_{j}(\bid) = 1$.
\end{lemma}
\begin{proof}
    For the sake of contradiction, suppose that there exists $\bid$ where $\util_{i}(\bid) >0$ and $x_{j}(\bid) = 0$ for some unique $b_j>b_i$. Consider the bid vector $\bid'$ where the user $i$ and user $j$ swap bid values: $b'_i=b_j$, $b'_j=b_i$, and $b'_\ell=b_\ell$ for $\ell\notin\{i,j\}$. By weak symmetry, the bid at $b_j$ must now be unconfirmed, i.e. $x_i(\bid')= 0$, since $\bid'$ and $\bid$ only differ by metadata and the bid at $b_j$ is unique.
    
    However,  consider the world where $\bid$ is the honest bid vector and the intermediate bid vector $\bid^* = (\bid_{-j}, b_i)$ is achieved by the strategy where user $j$ underbids to $b_i$.
    By Myerson's Lemma (\cref{Myerson's}), the confirmation function $x_j(\bid_{-j}, v)$ is weakly increasing in v, and thus $0 \le x_j(\bid_{-j}, b_i) \le x_j(\bid_{-j},b_j) = 0$.
    Thus, the underbidding strategy cannot harm user $j$.
    Consequently, by URHP, $\util_i(\bid^*) \ge \util_i(\bid)>0$ which implies $x_i(\bid^*) = 1$ in a deterministic TFM. 
    Now, by Myerson's Lemma (\cref{Myerson's}), when user $i$ raises their bid to $b_j$ to transform $\bid^*$ into $\bid'$, $x_i(\bid') = 1$, which is a contradiction. 
\end{proof}
\begin{lemma}\label{lem: include below}
   Suppose a deterministic TFM $({\bf I},{\bf C},{\bf P},{\bf R})$ satisfies UIC and URHP in the MPC-assisted model. 
   For any bid vector $\bid = (b_1,\dots, b_n)$, if there exists some $b_i \in \bid$ such that $\util_{i}(\bid) >0$, then for every $j$ where $b_j$ is unique and $b_j > p_i(\bid)$, it must hold that $x_{j}(\bid) = 1$.
\end{lemma}
\begin{proof}
    For the sake of contradiction, suppose that there exists $\bid$ where $\util_{i}(\bid) >0$ and $x_{j}(\bid) = 0$ for some unique $b_j>p_i(\bid)$.
    We then define $\bid' = (\bid_{-j}, b'_j)$ where $b_j' \in (p_i(\bid),b_j) \setminus \set{b_z: z \in [n]}$.
    Note that $b_j'$ is unique in $\bid'$.
   By Myerson's Lemma (\cref{Myerson's}), the confirmation function $x_j(\bid_{-j}, v)$ is weakly increasing in v, and thus $0 \le x_j(\bid_{-j}, b'_j) \le x_j(\bid_{-j},b_j) = 0$.
   This means the expected utility of user $j$ is still $0$ under the underbidding strategy.
    In a world where $\bid$ is the honest bid vector, user $j$ can adopt the strategy to underbid to $b'_j$ to achieve the realized bid vector $\bid'$. 
    Thus, by URHP, $\util_i(\bid') \ge \util_i(\bid)>0$ which implies $p_i(\bid')\le p_i(\bid)$ in a deterministic TFM. 
    \begin{claim}
        For $v \in (p_i(\bid),b_j')$, $\util_j(\bid_{-i},v)>0$.
    \end{claim}
    \begin{proof}
        It is sufficient to prove that $x_j(\bid'_{-i},v) =1$ as \cref{cor:increasing} would imply that $\util_j(\bid_{-i},v) >\util_j(\bid'_{-i},v)  \ge 0$.
        By Myerson's Lemma (\cref{Myerson's}), $x_i(\bid'_{-i},v) = 1$ since $v > p_i(\bid) \ge p_i(\bid')$.
        Furthermore, $v< b'_j,$ so by \cref{lem:cr-include-above}, $x_j(\bid'_{-i},v) =1$.
    \end{proof}
    Consider any $v\in(p_i(\bid),b_j')$ and the true-value profile $\widehat{\val}:=(\bid_{-i},v).$
    By the above claim, user $j$ has strictly positive utility when everyone behaves honestly under $\widehat{\val}$.

Now consider the deviation in which user $i$ reports $b_i$, producing the submitted bid vector $\bid$.
Since $\util_i(\bid)>0$, user $i$ is confirmed when bidding $b_i$, and in a deterministic UIC mechanism $p_i(\bid)$ is the threshold payment for user $i$ when the other bids are fixed to $\bid_{-i}$.
Because $v>p_i(\bid)$, honest bid $v$ is also confirmed and pays the same threshold $p_i(\bid)$.
Therefore, with user $i$'s true value fixed to $v$, both truthful reporting and the deviation to report $b_i$ give user $i$ utility $v-p_i(\bid)>0.$
Thus,
\[
\util_i(\widehat{\val};S_i)
=
\util_i(\widehat{\val};H_i),
\]
where $S_i$ denotes the strategy that reports $b_i$.

Under this deviation, the resulting bid vector is $\bid$, where $\util_j(\bid)=0$ by assumption.
Hence user $i$ has a utility-neutral deviation that reduces user $j$'s utility from a strictly positive value to zero, contradicting URHP.
\end{proof}
\begin{theorem} \label{thm:cr-ind-k+1}
    Let $k$ denote the block size.
    If a deterministic TFM $({\bf I},{\bf C},{\bf P},{\bf R})$ satisfies UIC and URHP in the MPC-assisted model, then no bid can be confirmed whenever the number of bids exceeds $k$.
\end{theorem}
\begin{proof}
Suppose, for the sake of contradiction, that there is a bid vector $\val=(v_1,\dots,v_n)$ with $n>k$ and some confirmed bid. Choose a confirmed bid $i^*$ and define $b_1=\max(\val)+1$ and $\val^{(1)}=(\val_{-i^*},b_1)$. For $m>1$, choose $b_m=\frac12(b_{m-1}+\max(\val))$ and set $\val^{(m)}=(\val^{(m-1)},b_m)$. 

We use the following claim, proved below by induction:
\begin{equation}
\label{eqn:deterministic-induction}
    \widetilde{x}_{b_m}(\val^{(m)})=1
\qquad\text{and}\qquad
\widetilde{p}_{b_m}(\val^{(m)})\le \max(\val)
\quad\text{for every }m\ge1.
\end{equation}
Applying to $m=k+1$, the bid $b_{k+1}$ is confirmed and pays at most $\max(\val)$, so it has strictly positive utility.
Since $b_1>\cdots>b_{k+1}>\max(\val)$ and all these bids are unique, \cref{lem:cr-include-above} implies that $b_1,\ldots,b_k$ are also confirmed in $\val^{(k+1)}$.
Hence $\val^{(k+1)}$ contains $k+1$ confirmed bids, contradicting block capacity $k$.
\end{proof}
\begin{proof}[of Claim \ref{eqn:deterministic-induction}]
The proof is by induction. The base case follows from UIC and Myerson's threshold-payment formula because $b_1$ replaces a confirmed bid by a higher value, so $\widetilde{x}_{b_1}(\val^{(1)})=1$ and $\widetilde{p}_{b_1}(\val^{(1)})\le \max(\val)$. For the inductive step, fix $b'\in(\max(\val),b_{m-1})$, keeping the appended bid's index and metadata fixed as its amount varies. If the unique fake bid $b'$ were unconfirmed in $(\val^{(m-1)},b')$, an unconfirmed user (which exists since $|\val^{(m-1)}|>k$) could inject that unconfirmed fake bid without lowering its zero honest utility: the fake bid pays zero, and individual rationality gives nonnegative utility from its unchanged truthful primary bid. By \cref{lem: include below}, the bid $b_{m-1}$ must then become unconfirmed or pay at least $b'$: if it remained confirmed and paid less than $b'<b_{m-1}$, it would have positive utility and the lemma would force $b'$ to be confirmed. Since $b'>\max(\val)\ge\widetilde{p}_{b_{m-1}}(\val^{(m-1)})$, in either case the utility of the bid $b_{m-1}$ would be strictly lower than in the original vector $\val^{(m-1)}$, violating URHP. Hence every $b'\in(\max(\val),b_{m-1})$ is confirmed. Myerson's lemma then gives the payment bound for $b_m$: its threshold is at most the infimum of this interval, namely $\max(\val)$, and the induction follows.
\end{proof}

\section{IC and RHP are Incomparable}
\label{sec:compare}
This section presents the mechanisms witness to show that IC and RHP are incomparable with respect to all three types of strategic players and in both the plain and the MPC-assisted models.
We summarize the mechanisms in \cref{tab:compare-plain} for the plain model and \cref{tab:compare-mpc} for the MPC-assisted model. 

\subsection{Comparison in the Plain Model}
\newcolumntype{L}{>{\raggedright\arraybackslash}X}
\newcolumntype{J}{>{\raggedright\arraybackslash}m{32mm}}
\newcolumntype{U}{>{\centering\arraybackslash}m{28mm}}

\begin{table}
\caption{Mechanisms for the incomparability of IC and RHP in the plain model.}
\label{tab:compare-plain}
\centering
\small
\setlength{\tabcolsep}{4pt}
\renewcommand{\arraystretch}{1.2}
\begin{tabularx}{\linewidth}{@{} J | L U U @{}}
\toprule
\textbf{Strategic Player} & \textbf{Mechanism} & \textbf{IC} & \textbf{RHP} \\
\midrule
\multirow{2}{*}{User} 
  & Burning first-price auction (\cref{sec:user-compare-plain}) & \xmark~UIC  & \cmark~URHP \\
  & Second price auction (\cref{sec:user-compare-plain}) & \cmark~UIC  & \xmark~URHP \\
\midrule
\multirow{2}{*}{Miner} 
  & All-or-nothing posted price (\cref{sec:aon-posted}) & \xmark~MIC  & \cmark~MRHP \\
  & 2-winner second-price (\cref{sec:miner-compare-plain}) & \cmark~MIC  & \xmark~MRHP \\
\midrule
\multirow{2}{*}{Miner-user coalition} 
  & All-or-nothing posted price (\cref{sec:aon-posted}) & \xmark~$1$-SCP  & \cmark~$1$-CRHP \\
  & 2-winner second-price (\cref{sec:miner-compare-plain}) & \cmark~$1$-SCP  & \xmark~$1$-CRHP \\
\bottomrule
\end{tabularx}
\end{table}

\subsubsection{UIC vs. URHP in the Plain Model}
\label{sec:user-compare-plain}
\begin{mdframed}
    \begin{center}
        {\bf Second-price auction}
    \end{center}
    \begin{itemize}[label=\textbullet, leftmargin = 5mm]
        \item \textbf{Inclusion \& Confirmation:} Include highest two bids, and confirm the top bid. Break ties arbitrarily. 
        \item \textbf{Payment \& Revenue:} All confirmed bids pay the price of the lowest included bid (if only one bid is included, is pays $0$) and all payments go to the miner.
    \end{itemize}
\end{mdframed}
\begin{lemma}[UIC does not imply URHP in the plain model]
\label{lem:uic-not-rhp}
The above second-price auction satisfies UIC but not URHP in the plain model.
\end{lemma}
\begin{proof}
\textbf{UIC:} For deviations without fake bids, UIC follows from the fact that for any fixed user $i$, for fixed other users' bids $\bid_{-i}$, user $i$'s allocation is monotone and the price is exactly as defined in \cref{Myerson's}. To include fake bids, let $r$ be the highest outside bid, or zero if none exists. Truthful utility is $\max\{v_i-r,0\}$. Under any deviation, a confirmed primary bid pays at least $r$, while confirmed fake bids only add payment costs; without a confirmed primary bid, utility is nonpositive. Thus every deviation is bounded by truthful utility, also after averaging over mixed strategies.\\
\textbf{Not URHP:} Suppose the honest bid vector was $(0,v_2)$ where $0<v_2$. Then, in the honest case, user $1$ has utility $0$ and user $2$ has utility $v_2$.
When user $1$ raises their bid to $\frac{v_2}{2}$, their utility remains the same.
However, user $2$'s payment rises, resulting in half the utility $v_2-\frac{v_2}{2} = \frac{v_2}{2}$.
\end{proof}

\begin{mdframed}
    \begin{center}
        {\bf Burning first-price auction}
    \end{center}
    \begin{itemize}[label=\textbullet, leftmargin = 5mm]
        \item \textbf{Inclusion \& Confirmation:} Include and confirm the top $k$ bids, breaking ties arbitrarily. 
        \item \textbf{Payment \& Revenue:} All confirmed bid price pay their bids, and all payments are burned.
    \end{itemize}
\end{mdframed}
\begin{lemma}[URHP does not imply UIC in the plain model]
\label{lem:urhp-not-uic}
The above burning first-price auction satisfies URHP but not UIC in the plain model. 
\end{lemma}
\begin{proof}
    \textbf{URHP:} Every protected honest user has utility exactly $0$ under every deviation: a confirmed truthful bid pays its value, and an unconfirmed bid pays nothing. Therefore no deviation, including a mixed strategy, can harm an honest user. URHP does not protect miner revenue.

    \textbf{Not UIC:} There is an incentive for strategic underbidding, where a user can bid below their true valuation to increase utility while still being confirmed. 
\end{proof}

\subsubsection{MIC vs. MRHP in the Plain Model}
\label{sec:miner-compare-plain}
\begin{mdframed}
    \begin{center}
        \textbf{2-winner second-price auction} \hfill //Block size $k \ge 2$.
    \end{center}
    \begin{itemize}[label=\textbullet]
        \item \textbf{Inclusion Rule \& Confirmation Rule:} The honest inclusion rule includes the two highest bids, or all bids if there is only one submitted bid. 
        The confirmation rule confirms the top two included bids, or the single bid if only one is included. Break ties deterministically based on metadata.
        \item \textbf{Payment Rule:} If two bids are confirmed, both pay the lower confirmed bid. If only one bid is confirmed, it pays $0$.
        \item \textbf{Miner Revenue Rule:} The miner receives the common payment of the confirmed bids as revenue, while the remaining payments are burned. 
    \end{itemize}
\end{mdframed}
\begin{lemma}[MIC does not imply MRHP in the plain model]
\label{lem:2-winners}
    The above 2-winner second price auction is MIC but not MRHP.
\end{lemma}
\textbf{MIC:} Let $h_2$ denote the second-highest honest bid, taking $h_2=0$ if fewer than two honest bids are submitted. Under honest behavior, the miner's utility is $h_2$.

Consider any miner deviation. If fewer than two bids are confirmed, the miner revenue is $0$. Suppose two bids are confirmed and let $p$ be the lower confirmed bid, so the miner revenue is $p$. If both confirmed bids are real user bids, then $p\leq h_2$. If at least one confirmed bid is fake bid from the miner, then every confirmed fake bid pays $p$, while the miner receives total revenue only $p$; hence the miner's utility is at most $0\leq h_2$. Thus no deviation gives the miner utility greater than the honest utility, and MIC holds.

\textbf{Not MRHP:} Consider the honest bid vector $(1,1,3)$. The lowest confirmed bid is $1$, so the miner revenue is $1$ and user $3$ has utility $2$. The miner can censor bid $3$. The two remaining bids are $(1,1)$, so the miner revenue remains $1$, while user $3$ becomes unconfirmed and its utility falls to $0$. Hence MRHP fails.
\begin{lemma}[MRHP does not imply MIC in the plain model]
    All-or-nothing posted-price auction (\cref{sec:aon-posted}) with $\eps > 0$ is MRHP but not MIC.
\end{lemma}
\begin{proof}
    By \cref{plain-AorN}.
\end{proof}
\subsubsection{1-SCP vs. 1-CRHP in the Plain Model}
\begin{lemma}[$1$-SCP  does not imply $1$-CRHP in the plain model]
     The $2$-winner second-price auction is $1$-SCP but not $1$-CRHP in the plain model.
\end{lemma}
\begin{proof}
\textbf{$1$-SCP:} Suppose, for the sake of contradiction, that $\val$ is the true value vector and there exists a coalition 
$\mathcal{C}$ consisting of the miner and a user~$i$ with valuation~$v_i$, 
and a joint strategy $S_{\mathcal{C}}$ such that $\util_\C(\val;S_\C) > \util_\C(\val;H_\C)$.
Let $p$ denote the payment associated with each confirmed bid under $S_{\mathcal{C}}$ and let $p_H$ denote the payment associated with each confirmed bid in the honest outcome.
We bound utility for each realized deviation; averaging these bounds then covers mixed strategies.

\paragraph{Case 1: $i$ is confirmed under $S_{\mathcal{C}}$.} 
Here, $\mathcal{C}$'s joint utility is at most the user’s valuation minus the payment, 
plus the payment received by the miner:
\[
 v_i - p + p = v_i.
\]
The coalition’s utility could be strictly less than $v_i$ if the deviation
involves fake bids that incur additional payments.  
If $i$ were also confirmed under the honest strategy, the coalition’s 
total utility would again be $v_i$; thus the deviation yields no gain.
Therefore, we can assume that $i$ is unconfirmed under the honest strategy.  
The joint utility in the honest case comes solely from the miner revenue which equals $p_H.$ 
Since $i$ is unconfirmed, we know $v_i \le p_H$. This shows $S_{\mathcal{C}}$ does not benefit the coalition.
\paragraph{Case 2: $i$ is unconfirmed under $S_{\mathcal{C}}$.} 
Here, $\mathcal{C}$'s joint utility comes solely from the miner revenue, $p$, but could be lower due to confirmed fake bids.
In the honest case, the joint utility is at least $p_H$, the miner revenue.
If $p \le p_H$, the deviation does not benefit $\mathcal{C}$. 
If $p > p_H$, then there must be at least one confirmed fake bid injected at or above $p$, meaning the coalition's joint utility under $S_{\mathcal{C}}$ would be at most $0 \le p_H$.

In both cases, $S_{\mathcal{C}}$ does not produce a joint utility higher than the honest joint utility for the coalition. 
This is a contradiction.\\
\textbf{Not $1$-CRHP:} 
Consider the honest value vector $(1,2)$ and the coalition $\mathcal{C}$ consisting of the miner and user $1$. 
Under the honest strategies, $\mathcal{C}$'s joint utility is $1$ while user $2$'s utility is $1$.
However, if the miner censors user $2$, $\mathcal{C}$'s joint utility is still $1$ while user $2$'s utility is $0$. Thus, this strategy harms user $2$ without harming the coalition.
\end{proof}

\begin{lemma}[$1$-CRHP  does not imply $1$-SCP in the plain model]
    All-or-nothing posted-price auction (\cref{sec:aon-posted}) with $\eps > 0$ is $d$-CRHP for any $d\geq 1$ but not $1$-SCP.
\end{lemma}
\begin{proof}
    By \cref{plain-AorN}.
\end{proof}

\subsection{Comparison in the MPC-Assisted Model}
\begin{table}
\caption{Mechanisms for the incomparability of IC and RHP in the MPC-assisted model.}
\label{tab:compare-mpc}
\centering
\small
\setlength{\tabcolsep}{4pt}
\renewcommand{\arraystretch}{1.2}
\begin{tabularx}{\linewidth}{@{} J | L U U @{}}
\toprule
\textbf{Strategic Player} & \textbf{Mechanism} & \textbf{IC} & \textbf{RHP} \\
\midrule
\multirow{2}{*}{User} 
  & Burning first-price auction (\cref{sec:user-compare-plain}) & \xmark~UIC  & \cmark~URHP \\
  & Second price auction (\cref{sec:user-compare-plain}) & \cmark~UIC  & \xmark~URHP \\
\midrule
\multirow{2}{*}{Miner} 
  & Revenue-capped first-price (\cref{sec:miner-compare-mpc}) & \xmark~MIC  & \cmark~MRHP \\
  & Trigger-posted-price (\cref{sec:miner-compare-mpc}) & \cmark~MIC  & \xmark~MRHP \\
\midrule
\multirow{2}{*}{Miner-user coalition} 
  & Burning first-price (\cref{sec:user-compare-plain}) & \xmark~$1$-SCP  & \cmark~$1$-CRHP \\
    & Trigger-posted-price (\cref{sec:miner-compare-mpc}) & \cmark~$1$-SCP & \xmark~$1$-CRHP \\
\bottomrule
\end{tabularx}
\end{table}
\subsubsection{UIC vs. URHP in the MPC-Assisted Model}
\cref{lem:uic-not-rhp} and \cref{lem:urhp-not-uic} still hold in the MPC-assisted model.
\subsubsection{MIC vs. MRHP in the MPC-Assisted Model}
\label{sec:miner-compare-mpc}
\begin{mdframed}
    \begin{center}
        \textbf{MPC-assisted, trigger-posted-price auction} \hfill //Reserve price ${\sf res} > 0$.
    \end{center}
    \begin{itemize}[label=\textbullet]
        \item \textbf{Inclusion Rule \& Confirmation Rule:} If there is exactly one bid strictly above reserve, include and confirm that bid. Otherwise, no bids are included or confirmed.
        \item \textbf{Payment \& Miner Revenue Rule:} All confirmed bids pay the reserve price, and all payments are burned.
    \end{itemize}
\end{mdframed}
\begin{lemma}[MIC does not imply MRHP in the MPC-assisted model.] 
\label{tpp-mpc}
    The above MPC assisted, trigger posted price auction satisfies MIC but not MRHP.
\end{lemma}
\begin{proof}
All confirmed bids pay the same amount, so we will only be referring to the confirmation probability in the proof below.\\
\textbf{MIC: } The miner receives no revenue so MIC, is trivially satisfied.\\
\textbf{Not MRHP:} 
Consider when there is an unconfirmed user $1$ and a confirmed user $2$ with positive utility in the honest case.
Then the miner coalition has expected honest utility $0$, and user $2$ has positive expected honest utility. If the miner coalition injects two bids above reserve, its expected utility is still $0$ while user $2$'s bid is now unconfirmed, leading to zero utility. This witnesses failure of MRHP without relying on a user-only deviation.
\end{proof}
\begin{mdframed}
    \begin{center}
        \textbf{Revenue-capped first-price auction} \hfill //Block size $k\ge 2$
    \end{center}

    \begin{itemize}
    \item \textbf{Inclusion Rule \& Confirmation Rule:} Include the top $2$ bids but confirm only the highest bid, breaking ties arbitrarily.
    \item \textbf{Payment Rule:} The confirmed bid pays its bid.
    \item \textbf{Miner Revenue Rule:} The total miner revenue is the highest non-confirmed included bid, or $0$ if there is no non-confirmed included bid. 
    Total revenue is divided equally among the $M$ miners.
    \end{itemize}
\end{mdframed}
\begin{lemma}[MRHP does not imply MIC in the MPC-assisted model]
   The above mechanism satisfies MRHP but not MIC in the MPC-assisted model. 
\end{lemma}
\begin{proof}
\textbf{MRHP:} Fix a coalition $\C$ of $1\leq m<M/2$ miners. Every honest user always has utility $0$: if confirmed, the user pays its truthful bid, and if unconfirmed, it pays nothing. Hence a miner deviation cannot harm an honest user.

Let $R_H$ denote the total miner revenue under honest behavior and let $R'$ denote the expected total miner revenue under an arbitrary, possibly randomized, deviation. Let $F\geq 0$ denote the expected total payment incurred by confirmed fake bids injected by $\C$. Since revenue is divided equally, the coalition's honest utility is $\frac{m}{M}R_H$, while its expected utility under the deviation is
\[
\frac{m}{M}R'-F.
\]
Therefore, if the deviation does not decrease the coalition's utility, then
\[
\frac{m}{M}R'-F\geq \frac{m}{M}R_H,
\]
which implies $R'\geq R_H$. Hence every honest miner outside $\C$ receives expected revenue at least $R_H/M$, so no honest miner is harmed. Thus the mechanism satisfies MRHP.

\textbf{Not MIC:} Consider the honest bid vector $(1,0)$. The honest total miner revenue is $0$. A nonempty coalition of $m<M/2$ miners can inject a fake bid of value $1/2$. The real bid $1$ remains the unique confirmed bid, while the fake bid $1/2$ is the highest non-confirmed bid and therefore raises the total miner revenue to $1/2$. The fake bid is unconfirmed and pays nothing, so the coalition's utility increases by $m/(2M)>0$. Hence MIC fails.
\end{proof}
\subsubsection{1-SCP vs. 1-CRHP in the MPC-Assisted Model}
\begin{lemma}[$1$-SCP does not imply $1$-CRHP in the MPC-assisted model]
    The trigger-posted-price satisfies $1$-SCP but not $1$-CRHP in the MPC-assisted model.
\end{lemma}
\begin{proof}
\textbf{$1$-SCP:} Fix a coalition $\C$ consisting of $1\leq m<M/2$ miners and one user $i$. The miners revenue is always zero. Moreover, injecting an eligible fake bid cannot help user $i$: a confirmed fake bid incurs a payment of $\res$, while an additional eligible fake bid can only change an instance with exactly one eligible bid into an instance with at least two eligible bids, in which no bid is confirmed.

It therefore suffices to consider deviations of user $i$ without eligible fake bids. If there is no eligible honest bid outside $\C$, then user $i$ faces a posted price $\res$: bidding above reserve gives utility $v_i-\res$, while bidding at or below reserve gives utility $0$, so truthful bidding is optimal. If there is an eligible honest bid outside $\C$, the coalition cannot censor it, and user $i$ cannot make its own bid the unique eligible bid. Hence the coalition has no strictly profitable deviation, and the mechanism is $1$-SCP.

\textbf{Not $1$-CRHP:} Consider a coalition consisting of $1\leq m<M/2$ miners and a user with value $0$, together with an honest user of value $2\res$. Under honest behavior, the honest user is the unique eligible bidder and obtains utility $\res$, while the coalition has utility $0$. The miner coalition can inject one fake bid strictly above $\res$. There are then at least two eligible bids, so no bid is confirmed. The fake bid is unconfirmed and pays nothing, so the coalition's utility remains $0$, while the honest user's utility falls from $\res$ to $0$. Thus $1$-CRHP fails.
\end{proof}
\begin{lemma}[$1$-CRHP does not imply $1$-SCP in the MPC-assisted model]
    The burning first-price auction satisfies $d$-CRHP for every $d\geq 1$ but not $1$-SCP in the MPC-assisted model.
\end{lemma}
\begin{proof}
\textbf{$d$-CRHP:} Every honest user always has utility $0$: if confirmed, the user pays its truthful bid, and if unconfirmed, it pays nothing. Moreover, all payments are burned, so every miner receives zero revenue under every outcome. Consequently, no deviation by a miner-user coalition can make any protected honest user or honest miner strictly worse off. Thus the mechanism satisfies $d$-CRHP for every $d\geq 1$.

\textbf{Not $1$-SCP:} Consider a single user with value $v>0$ together with any nonempty allowed miner coalition. Under truthful bidding, the user is confirmed, pays $v$, and obtains utility $0$. If the user instead bids some $b\in(0,v)$ while the miners follow the prescribed mechanism, the user remains confirmed and obtains utility $v-b>0$. Since miner revenue remains zero, this strictly increases the coalition's joint utility. Hence the mechanism is not $1$-SCP.
\end{proof}

    
    

\paragraph{Acknowledgements.}
This work is supported by the 2024 \href{https://research.stellar.org/research-grants}{Stellar Academic Research Grant}.

\begingroup
\phantomsection
\addcontentsline{toc}{section}{References}
\let\url\nolinkurl
\providecommand{\doi}[1]{doi:\nolinkurl{#1}}
\bibliographystyle{alpha}
\bibliography{ref-arxiv}

@article{Erdos1964,
  author  = {Erd{\H{o}}s, P.},
  title   = {On extremal problems of graphs and generalized graphs},
  journal = {Israel Journal of Mathematics},
  volume  = {2},
  number  = {3},
  pages   = {183--190},
  year    = {1964},
  url     = {https://www.renyi.hu/~p_erdos/1964-13.pdf}
}

@inproceedings{GaneshThomasWeinberg-EC24,
  author    = {Aadityan Ganesh and Clayton Thomas and S. Matthew Weinberg},
  title     = {Revisiting the Primitives of Transaction Fee Mechanism Design},
  booktitle = {Proceedings of the 25th ACM Conference on Economics and Computation (EC)},
  year      = {2024}
}

@article{Klaus2001,
  author  = {Klaus, Bettina},
  title   = {Coalitional Strategy-Proofness in Economies with Single-Dipped Preferences and the Assignment of an Indivisible Object},
  journal = {Games and Economic Behavior},
  year    = {2001},
  volume  = {34},
  number  = {1},
  pages   = {64--82},
  doi     = {10.1006/game.1999.0789},
  issn    = {0899-8256},
  url     = {https://www.sciencedirect.com/science/article/pii/S0899825699907893}
}

@article{angeris2024multidimensional,
  title={Multidimensional blockchain fees are (essentially) optimal},
  author={Angeris, Guillermo and Diamandis, Theo and Moallemi, Ciamac},
  journal={arXiv preprint arXiv:2402.08661},
  year={2024}
}

@inproceedings{aumayr2023domino,
  author = {Aumayr, Lukas and Moreno-Sanchez, Pedro and Kate, Aniket and Maffei, Matteo},
  title = {Breaking and Fixing Virtual Channels: Domino Attack and Donner},
  booktitle = {Network and Distributed System Security Symposium (NDSS)},
  pages = {1--18},
  year = {2023},
  publisher = {The Internet Society},
  doi = {10.14722/ndss.2023.24370}
}

@misc{buterin2017triangle,
author = {Buterin, Vitalik},
title = {The Triangle of Harm},
year = {2017},
howpublished = {\url{https://vitalik.ca/general/2017/07/16/triangle_of_harm.html}}
}

@misc{buterin2018griefing,
author = {Buterin, Vitalik},
title = {A Griefing Factor Analysis Model},
year = {2018},
howpublished = {\url{https://ethresear.ch/t/a-griefing-factor-analysis-model/2338}}
}

@article{carroll2019robustness,
  title={Robustness in mechanism design and contracting},
  author={Carroll, Gabriel},
  journal={Annual Review of Economics},
  volume={11},
  number={1},
  pages={139--166},
  year={2019},
  publisher={Annual Reviews}
}

@article{chen2025bayesian,
  title={Bayesian mechanism design for blockchain transaction fee allocation},
  author={Chen, Xi and Simchi-Levi, David and Zhao, Zishuo and Zhou, Yuan},
  journal={Operations Research},
  year={2025},
  publisher={INFORMS}
}

@inproceedings{chung2024collusion,
  title={Collusion-resilience in transaction fee mechanism design},
  author={Chung, Hao and Roughgarden, Tim and Shi, Elaine},
  booktitle={Proceedings of the 25th ACM Conference on Economics and Computation},
  pages={1045--1073},
  year={2024}
}

@article{credibleauction,
 author = {Mohammad Akbarpour and Shengwu Li},
  year = 2020,  
  title = {Credible Auctions: A Trilemma}, 
  journal = {Econometrica, Econometric Society}, 
}

@inproceedings{credibleauction-comm00,
  author    = {Matheus V. X. Ferreira and
               S. Matthew Weinberg},
  editor    = {P{\'{e}}ter Bir{\'{o}} and
               Jason D. Hartline and
               Michael Ostrovsky and
               Ariel D. Procaccia},
  title     = {Credible, Truthful, and Two-Round (Optimal) Auctions via Cryptographic
               Commitments},
  booktitle = {{EC} '20: The 21st {ACM} Conference on Economics and Computation,
               Virtual Event, Hungary, July 13-17, 2020},
  pages     = {683--712},
  publisher = {{ACM}},
  year      = {2020},
}

@inproceedings{credibleauction-comm01,
  author    = {Meryem Essaidi and
               Matheus V. X. Ferreira and
               S. Matthew Weinberg},
  editor    = {Mark Braverman},
  title     = {Credible, Strategyproof, Optimal, and Bounded Expected-Round Single-Item
               Auctions for All Distributions},
  booktitle = {13th Innovations in Theoretical Computer Science Conference, {ITCS}
               2022, January 31 - February 3, 2022, Berkeley, CA, {USA}},
  series    = {LIPIcs},
  volume    = {215},
  pages     = {66:1--66:19},
  year      = {2022},
}

@InProceedings{crypto-tfm,
  author =	{Shi, Elaine and Chung, Hao and Wu, Ke},
  title =	{{What Can Cryptography Do for Decentralized Mechanism Design?}},
  booktitle =	{ITCS 2023},
  pages =	{97:1--97:22},
  series =	{LIPIcs},
  ISBN =	{978-3-95977-263-1},
  ISSN =	{1868-8969},
  year =	{2023},
  volume =	{251},
  editor =	{Tauman Kalai, Yael},
  publisher =	{Schloss Dagstuhl -- Leibniz-Zentrum f{\"u}r Informatik},
  address =	{Dagstuhl, Germany},
  URL =		{https://drops.dagstuhl.de/entities/document/10.4230/LIPIcs.ITCS.2023.97},
  URN =		{urn:nbn:de:0030-drops-176005},
  doi =		{10.4230/LIPIcs.ITCS.2023.97}
}

@misc{cryptoeprint:2025/1086,
      author = {Hao Chung and Elisaweta Masserova and Elaine Shi and Sri AravindaKrishnan Thyagarajan},
      title = {Fairness in the Wild: Secure Atomic Swap with External Incentives},
      howpublished = {Cryptology {ePrint} Archive, Paper 2025/1086},
      year = {2025},
      url = {https://eprint.iacr.org/2025/1086}
}

@inproceedings{diamandis2023designing,
  title={Designing multidimensional blockchain fee markets},
  author={Diamandis, Theo and Evans, Alex and Chitra, Tarun and Angeris, Guillermo},
  booktitle={5th Conference on Advances in Financial Technologies (AFT 2023)},
  pages={4:1--4:23},
  volume={282},
  series={Leibniz International Proceedings in Informatics (LIPIcs)},
  doi={10.4230/LIPIcs.AFT.2023.4},
  year={2023},
  organization={Schloss Dagstuhl--Leibniz-Zentrum f{\"u}r Informatik}
}

@book{duque2024local,
  title={Local Non-Bossiness and Preferences Over Colleagues},
  author={Duque, Eduardo and Pereyra, Juan and Torres-Mart{\'\i}nez, Juan Pablo},
  year={2024},
  publisher={Universidad de Chile, Departamento de Econom{\'\i}a}
}

@inproceedings{dynamicpostedprice,
  title={Dynamic posted-price mechanisms for the blockchain transaction-fee market},
  author={Ferreira, Matheus VX and Moroz, Daniel J and Parkes, David C and Stern, Mitchell},
  booktitle={Proceedings of the 3rd ACM Conference on Advances in Financial Technologies},
  pages={86--99},
  year={2021}
}

@misc{eip1559,
  author={Vitalik Buterin and Eric Conner and Rick Dudley and Matthew Slipper and Ian Norden and Abdelhamid Bakhta},
  title={{EIP-1559}: Fee Market Change for {ETH} 1.0 Chain},
  howpublished={Ethereum Improvement Proposals, no. 1559, \url{https://eips.ethereum.org/EIPS/eip-1559}},
  month=apr,
  year={2019}
}

@misc{eip7805,
  author       = {Thomas Thiery and Francesco {D'Amato} and Julian Ma and Barnab\'e Monnot and Terence Tsao and Jacob Kaufmann and Jihoon Song},
  title        = {{EIP-7805}: Fork-choice enforced Inclusion Lists (FOCIL)},
  howpublished = {Ethereum Improvement Proposals},
  year         = {2024},
  month        = nov,
  note         = {Created: 2024-11-01. Accessed: 2026-02-08},
  url          = {https://eips.ethereum.org/EIPS/eip-7805}
}

@misc{ethereum_randao_docs,
  title        = {Block proposal},
  author       = {{Ethereum.org}},
  howpublished = {\url{https://ethereum.org/developers/docs/consensus-mechanisms/pos/block-proposal/}},
  year         = {2026},
  note         = {Ethereum documentation, accessed March 29, 2026}
}

@article{ferreira2024incentive,
  title={Incentive-Compatible Collusion-Resistance via Posted Prices},
  author={Ferreira, Matheus VX and Gafni, Yotam and Resnick, Max},
  journal={arXiv preprint arXiv:2412.20853},
  year={2024}
}

@inproceedings{foundation-tfm,
  title={Foundations of transaction fee mechanism design},
  author={Chung, Hao and Shi, Elaine},
  booktitle={Proceedings of the 2023 Annual ACM-SIAM Symposium on Discrete Algorithms (SODA)},
  pages={3856--3899},
  year={2023},
  organization={SIAM}
}

@article{functional-fee-market,
  author    = {Soumya Basu and
               David A. Easley and
               Maureen O'Hara and
               Emin G{\"{u}}n Sirer},
  title     = {Towards a Functional Fee Market for Cryptocurrencies},
  journal   = {CoRR},
  volume    = {abs/1901.06830},
  year      = {2019},
  url       = {http://arxiv.org/abs/1901.06830},
}

@inproceedings{gafni2024barriers,
  title={Barriers to collusion-resistant transaction fee mechanisms},
  author={Gafni, Yotam and Yaish, Aviv},
  booktitle={Proceedings of the 25th ACM Conference on Economics and Computation},
  pages={1074--1096},
  year={2024}
}

@article{ganesh2025characterizing,
  title={Characterizing Off-Chain Influence Proof Transaction Fee Mechanisms},
  author={Ganesh, Aadityan and Thomas, Clayton and Weinberg, S Matthew},
  journal={arXiv preprint arXiv:2512.02354},
  year={2025}
}

@inproceedings{garimidi2025transaction,
  title={Transaction Fee Mechanism Design for Leaderless Blockchain Protocols},
  author={Garimidi, Pranav and Heimbach, Lioba and Roughgarden, Tim},
  booktitle={International Conference on Financial Cryptography and Data Security},
  pages={20--35},
  year={2025},
  organization={Springer}
}

@article{george2023griefs,
author = {George, William},
title = {An Analysis of Griefs and Griefing Factors},
journal = {Frontiers in Blockchain},
volume = {6},
pages = {1137155},
year = {2023},
publisher = {Frontiers},
doi = {10.3389/fbloc.2023.1137155}
}

@inproceedings{gmw87,
 author = {Goldreich, O. and Micali, S. and Wigderson, A.},
 title = {How to play ANY mental game},
 booktitle = {ACM symposium on Theory of computing (STOC)},
 year = {1987},
}

@inproceedings{heimbach2025early,
  title={The early days of the ethereum blob fee market and lessons learnt},
  author={Heimbach, Lioba and Milionis, Jason},
  booktitle={International Conference on Financial Cryptography and Data Security},
  pages={53--71},
  year={2025},
  organization={Springer}
}

@inproceedings{herlihy2018atomic,
author = {Herlihy, Maurice},
title = {Atomic Cross-Chain Swaps},
booktitle = {PODC},
pages = {245--254},
year = {2018},
publisher = {ACM},
doi = {10.1145/3212734.3212736}
}

@incollection{jehiel2005allocative,
  title={Allocative and Informational Externalities in Auctions and Related Mechanisms},
  author={Jehiel, Philippe and Moldovanu, Benny},
  editor={Blundell, Richard and Newey, Whitney K. and Persson, Torsten},
  booktitle={Advances in Economics and Econometrics: Theory and Applications, Ninth World Congress},
  volume={1},
  chapter={3},
  pages={102--135},
  publisher={Cambridge University Press},
  year={2006},
  doi={10.1017/CBO9781139052269.005}
}

@article{lavee2025does,
  title={Does Your Blockchain Need Multidimensional Transaction Fees?},
  author={Lavee, Nir and Nisan, Noam and Pai, Mallesh and Resnick, Max},
  journal={arXiv preprint arXiv:2504.15438},
  year={2025}
}

@article{leme2023nonbossy,
  title={Nonbossy mechanisms: Mechanism design robust to secondary goals},
  author={Leme, Renato Paes and Schneider, Jon and Zhang, Hanrui},
  journal={arXiv preprint arXiv:2307.11967},
  year={2023}
}

@incollection{leonardos2023griefing,
author = {Leonardos, Stefanos and Sridhar, Shyam and Cheung, Yun Kuen and Piliouras, Georgios},
title = {Griefing Factors and Evolutionary In-Stabilities in Blockchain Mining Games},
booktitle = {Mathematical Research for Blockchain Economy},
editor = {Pardalos, Panos M. and Kotsireas, Ilias S. and Guo, Yike and Knottenbelt, William J.},
pages = {75--94},
year = {2023},
publisher = {Springer},
doi = {10.1007/978-3-031-18679-0_5}
}

@article{lopomo2021uncertainty,
  title={Uncertainty in mechanism design},
  author={Lopomo, Giuseppe and Rigotti, Luca and Shannon, Chris},
  journal={arXiv preprint arXiv:2108.12633},
  year={2021}
}

@misc{lopp2017feeestimation,
  author       = {Lopp, Jameson},
  title        = {The Challenges of {Bitcoin} Transaction Fee Estimation},
  howpublished = {BitGo blog post},
  year         = {2017},
  month        = may,
  url          = {https://blog.bitgo.com/the-challenges-of-bitcoin-transaction-feeestimation-e47a64a61c72},
  note         = {Accessed: 2026-02-08}
}

@article{mazumdar2023strategic,
author = {Mazumdar, Subhra and Banerjee, Prabal and Sinha, Abhinandan and Ruj, Sushmita and Roy, Bimal Kumar},
title = {Strategic Analysis of Griefing Attack in Lightning Network},
journal = {IEEE Transactions on Network and Service Management},
volume = {20},
number = {2},
pages = {1790--1803},
year = {2023},
publisher = {IEEE},
doi = {10.1109/TNSM.2022.3230768}
}

@article{myerson,
author = {Myerson, Roger B.},
title = {Optimal Auction Design},
year = {1981},
issue_date = {February 1981},
publisher = {INFORMS},
volume = {6},
number = {1},
issn = {0364-765X},
journal = {Math. Oper. Res.},
}

@inproceedings{raikwar2024sok,
  title={SoK: Dag-based consensus protocols},
  author={Raikwar, Mayank and Polyanskii, Nikita and M{\"u}ller, Sebastian},
  booktitle={2024 IEEE International Conference on Blockchain and Cryptocurrency (ICBC)},
  pages={1--18},
  year={2024},
  organization={IEEE}
}

@InProceedings{reasonable-tfm,
  author =	{Wu, Ke and Shi, Elaine and Chung, Hao},
  title =	{{Maximizing Miner Revenue in Transaction Fee Mechanism Design}},
  booktitle =	{ITCS 2024},
  pages =	{98:1--98:23},
  series =	{LIPIcs},
  ISBN =	{978-3-95977-309-6},
  ISSN =	{1868-8969},
  year =	{2024},
  volume =	{287},
  editor =	{Guruswami, Venkatesan},
  publisher =	{Schloss Dagstuhl -- Leibniz-Zentrum f{\"u}r Informatik},
  address =	{Dagstuhl, Germany},
  URN =		{urn:nbn:de:0030-drops-196266},
  doi =		{10.4230/LIPIcs.ITCS.2024.98}
}

@article{ritz1983restricted,
  title={Restricted domains, arrow-social welfare functions and noncorruptible and non-manipulable social choice correspondences: the case of private alternatives},
  author={Ritz, Zvi},
  journal={Mathematical Social Sciences},
  volume={4},
  number={2},
  pages={155--179},
  year={1983},
  publisher={Elsevier}
}

@article{roughgarden2021transaction,
  title={Transaction fee mechanism design},
  author={Roughgarden, Tim},
  journal={ACM SIGecom Exchanges},
  volume={19},
  number={1},
  pages={52--55},
  year={2021},
  publisher={ACM New York, NY, USA}
}

@article{saijo2007secure,
  title={Secure implementation},
  author={Saijo, Tatsuyoshi and Sjostrom, Tomas and Yamato, Takehiko},
  journal={Theoretical Economics},
  volume={2},
  number={3},
  pages={203--229},
  year={2007}
}

@article{satterthwaite1981strategy,
  title={Strategy-proof allocation mechanisms at differentiable points},
  author={Satterthwaite, Mark A and Sonnenschein, Hugo},
  journal={The Review of Economic Studies},
  volume={48},
  number={4},
  pages={587--597},
  year={1981},
  publisher={Wiley-Blackwell}
}

@article{stouka2025multiple,
  title={Multiple Proposer Transaction Fee Mechanism Design: Robust Incentives Against Censorship and Bribery},
  author={Stouka, Aikaterini-Panagiota and Ma, Julian and Thiery, Thomas},
  journal={arXiv preprint arXiv:2505.13751},
  year={2025}
}

@article{thomson2016non,
  title={Non-bossiness},
  author={Thomson, William},
  journal={Social Choice and Welfare},
  volume={47},
  number={3},
  pages={665--696},
  year={2016},
  publisher={Springer}
}

@inproceedings{tsabary2021madhtlc,
  author = {Tsabary, Itay and Yechieli, Matan and Manuskin, Alex and Eyal, Ittay},
  title = {{MAD-HTLC}: Because {HTLC} is Crazy-Cheap to Attack},
  booktitle = {2021 IEEE Symposium on Security and Privacy (S\&P)},
  pages = {1230--1248},
  year = {2021},
  publisher = {IEEE},
  doi = {10.1109/SP40001.2021.00080}
}

@inproceedings{uc,
 author = {Canetti, R.},
 title = {Universally Composable Security: A New Paradigm for Cryptographic Protocols},
 booktitle = {FOCS},
 year = {2001},
}

@article{vickrey1961counterspeculation,
  title={Counterspeculation, auctions, and competitive sealed tenders},
  author={Vickrey, William},
  journal={The Journal of finance},
  volume={16},
  number={1},
  pages={8--37},
  year={1961},
  publisher={JSTOR}
}

@inproceedings{xue2021hedging,
author = {Xue, Yingjie and Herlihy, Maurice},
title = {Hedging Against Sore Loser Attacks in Cross-Chain Transactions},
booktitle = {PODC},
pages = {155--164},
year = {2021},
publisher = {ACM},
doi = {10.1145/3465084.3467904}
}

@inproceedings{yaofeemech,
  author={Yao, Andrew Chi-Chih},
  title={{An Incentive Analysis of Some Bitcoin Fee Designs (Invited Talk)}},
  booktitle={47th International Colloquium on Automata, Languages, and Programming (ICALP 2020)},
  pages={1:1--1:12},
  series={Leibniz International Proceedings in Informatics (LIPIcs)},
  volume={168},
  year={2020},
  publisher={Schloss Dagstuhl -- Leibniz-Zentrum f{\"u}r Informatik},
  url={https://drops.dagstuhl.de/entities/document/10.4230/LIPIcs.ICALP.2020.1},
  doi={10.4230/LIPIcs.ICALP.2020.1}
}

@inproceedings{zoharfeemech,
  author    = {Ron Lavi and
               Or Sattath and
               Aviv Zohar},
  title     = {Redesigning Bitcoin's fee market},
  booktitle = {The World Wide Web Conference, {WWW} 2019},
  pages     = {2950--2956},
  year      = {2019},
}
\endgroup
\end{document}